\documentclass[pra, twocolumn, nofootinbib, preprintnumbers, superscriptaddress]{revtex4-1}

\setcitestyle{super}


\usepackage{graphicx, color, graphpap}      
\usepackage{amsmath}
\usepackage{enumitem}
\usepackage{amssymb}
\usepackage{amsthm}
\usepackage{pstricks}
\usepackage{float}
\usepackage[colorlinks=true, citecolor=blue, linkcolor=red]{hyperref}
\usepackage[T1]{fontenc}
\usepackage{bbm}
\usepackage{dsfont}
\usepackage[linesnumbered,ruled,vlined]{algorithm2e}
\SetKwInput{kwInit}{Init}
\usepackage{mathtools}

\usepackage{tikz}
\usepackage{graphicx}
\usetikzlibrary{positioning}
\usepackage{graphicx}
\usepackage{xcolor}
\usepackage[caption=false]{subfig}
\usepackage[ruled,vlined]{algorithm2e}
\usepackage{siunitx}
\usepackage{qcircuit}
\usepackage{pifont} 
\usepackage{tabularx}

\usepackage{wrapfig}

\usepackage{url}

\usepackage{breakurl}
\usepackage{relsize}

\DeclareUnicodeCharacter{2212}{-}

\newcommand{\secref}[1]{Section \ref{#1}}
\newcommand{\appref}[1]{Appendix \ref{#1}}
\newcommand{\thmref}[1]{Theorem \ref{#1}}
\newcommand{\lemref}[1]{Lemma \ref{#1}}
\newcommand{\figref}[1]{Fig.~\ref{#1}}
\renewcommand{\eqref}[1]{Eq. (\ref{#1})}

\newcommand{\attref}[1]{Attack \ref{#1}}

\newcommand{\braket}[2]{\langle #1 \hspace{1pt} | \hspace{1pt} #2 \rangle}
\newcommand{\ket}[1]{| #1 \rangle}
\newcommand{\bra}[1]{\langle #1 |}

\newcommand{\ketbra}[2]{| \hspace{1pt} #1 \rangle \langle #2 \hspace{1pt} |}

\newcommand{\norm}[2][]{#1| \! #1| #2 #1| \! #1|}

\newcommand{\sq}{\textsf{sq}}

\newcommand{\Tr}{{\rm Tr}}

\renewcommand{\geq}{\geqslant}
\renewcommand{\leq}{\leqslant}

\newcommand{\paramtheta}{\boldsymbol{\theta}}

\DeclareMathOperator{\Lbs}{\textsf{L}}
\DeclareMathOperator{\Gbs}{\textsf{G}}
\DeclareMathOperator{\Cbs}{\textsf{C}}

\newcommand{\CZ}{\mathsf{CZ}}

\newcommand{\SWAP}{\mathsf{SWAP}}

\newcommand{\Ansatze}{\textrm{Ans\"{a}tze}}
\newcommand{\Ansatz}{\textrm{Ans\"{a}tz}}

\newtheoremstyle{example}{\topsep}{\topsep}%
{}
{}
{\bfseries}
{:}
{   }
{\thmname{#1}\thmnumber{ #2}}
\theoremstyle{example}
\newtheorem{theorem}{Theorem}
\newtheorem{lemma}{Lemma}
\newtheorem{corollary}{Corollary}

\theoremstyle{definition}
\newtheorem{definition}{Definition}

\newtheorem*{theorem*}{Theorem}

\def\orcid#1{\kern -0.4em\href{https://orcid.org/#1}{\includegraphics[keepaspectratio,width=0.7em]{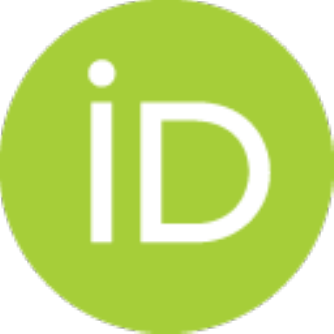}}}

\DeclareMathOperator*{\argmin}{arg\,min}

\newcommand{\opt}{\mathrm{opt}}

\definecolor{ForestGreen}{RGB}{34, 139, 34}

\long\def\ca#1\cb{} 
\newcommand{\computerfont}[1]{{\fontfamily{cmtt}\selectfont #1}}

\def\orcid#1{\kern -0.4em\href{https://orcid.org/#1}{\includegraphics[keepaspectratio,width=0.7em]{orcid_logo.pdf}}}

\newcounter{attack}
\newenvironment{attack}[1]
  {\par\addvspace{\topsep}
   \noindent
   \tabularx{\linewidth}{@{} X @{}}
    \hline
    \refstepcounter{attack}\textbf{Attack \theattack} #1 \\
    \hline}
  { \\
    \hline
   \endtabularx
   \par\addvspace{\topsep}}

\newcommand{\sbline}{\\[.5\normalbaselineskip]}

\begin{document}
\title{Variational Quantum Cloning: Improving Practicality for Quantum Cryptanalysis}
 
\author{Brian Coyle~\orcid{0000-0002-3436-8458}}
\affiliation{School of Informatics, 10 Crichton Street, Edinburgh, United Kingdom, EH8 9AB.}

\author{Mina Doosti}
\affiliation{School of Informatics, 10 Crichton Street, Edinburgh, United Kingdom, EH8 9AB.}

\author{Elham Kashefi}
\affiliation{School of Informatics, 10 Crichton Street, Edinburgh, United Kingdom, EH8 9AB.}

\affiliation{CNRS, LIP6, Sorbonne Universit\'{e}, 4 place Jussieu, 75005 Paris, France.}

\author{Niraj Kumar~\orcid{0000-0003-3037-1083}}
\affiliation{School of Informatics, 10 Crichton Street, Edinburgh, United Kingdom, EH8 9AB.}

\begin{abstract}
  Cryptanalysis on standard quantum cryptographic systems generally involves finding optimal adversarial attack strategies on the underlying protocols. The core principle of modelling quantum attacks in many cases reduces to the adversary's ability to clone unknown quantum states which facilitates the extraction of some meaningful secret information.
   Explicit optimal attack strategies typically require high computational resources due to large circuit depths or, in many cases, are unknown. In this work, we propose variational quantum cloning (VQC), a quantum machine learning based cryptanalysis algorithm which allows an adversary to obtain optimal (approximate) cloning strategies with short depth quantum circuits, trained using hybrid classical-quantum techniques. The algorithm contains operationally meaningful cost functions with theoretical guarantees, quantum circuit structure learning and gradient descent based optimisation. Our approach enables the end-to-end discovery of hardware efficient quantum circuits to clone specific families of quantum states, which in turn leads to an improvement in cloning fidelites when implemented on quantum hardware: the Rigetti Aspen chip. Finally, we connect these results to quantum cryptographic primitives, in particular quantum coin flipping. We derive attacks on two protocols as examples, based on quantum cloning and facilitated by VQC. As a result, our algorithm can improve near term attacks on these protocols, using approximate quantum cloning as a resource.
\end{abstract}

\maketitle

\section{Introduction} \label{ssec:intro}

In recent times, small scale quantum computers which can support on the order of $50 - 200$ qubits have come into existence, which showcases that we are now firmly in the noisy intermediate scale quantum (NISQ) era\cite{preskill_quantum_2018}.  These devices lack the capabilities of quantum error correction\cite{brun_quantum_2019, devitt_quantum_2013}, and this coupled with their small sizes puts a speedup in, for example, factoring large prime numbers~\cite{shor_algorithms_1994} out of reach. However, the availability of such devices over the cloud~\cite{larose_overview_2019, bergholm_pennylane_2020, broughton_tensorflow_2020} has led to an increasing study of their capabilities and usefulness. Concurrently to the rapid development of the quantum hardware, many proposals have been made for quantum algorithms and applications which are tailored for 
NISQ devices. The promise of these approaches is further boosted by the recent implementations of the quantum computation which provably cannot be simulated by any classical device in reasonable time~\cite{arute_quantum_2019, zhong_quantum_2020}.
\begin{figure}[ht]
    \centering
\includegraphics[width=0.8\columnwidth, height=0.7\columnwidth]{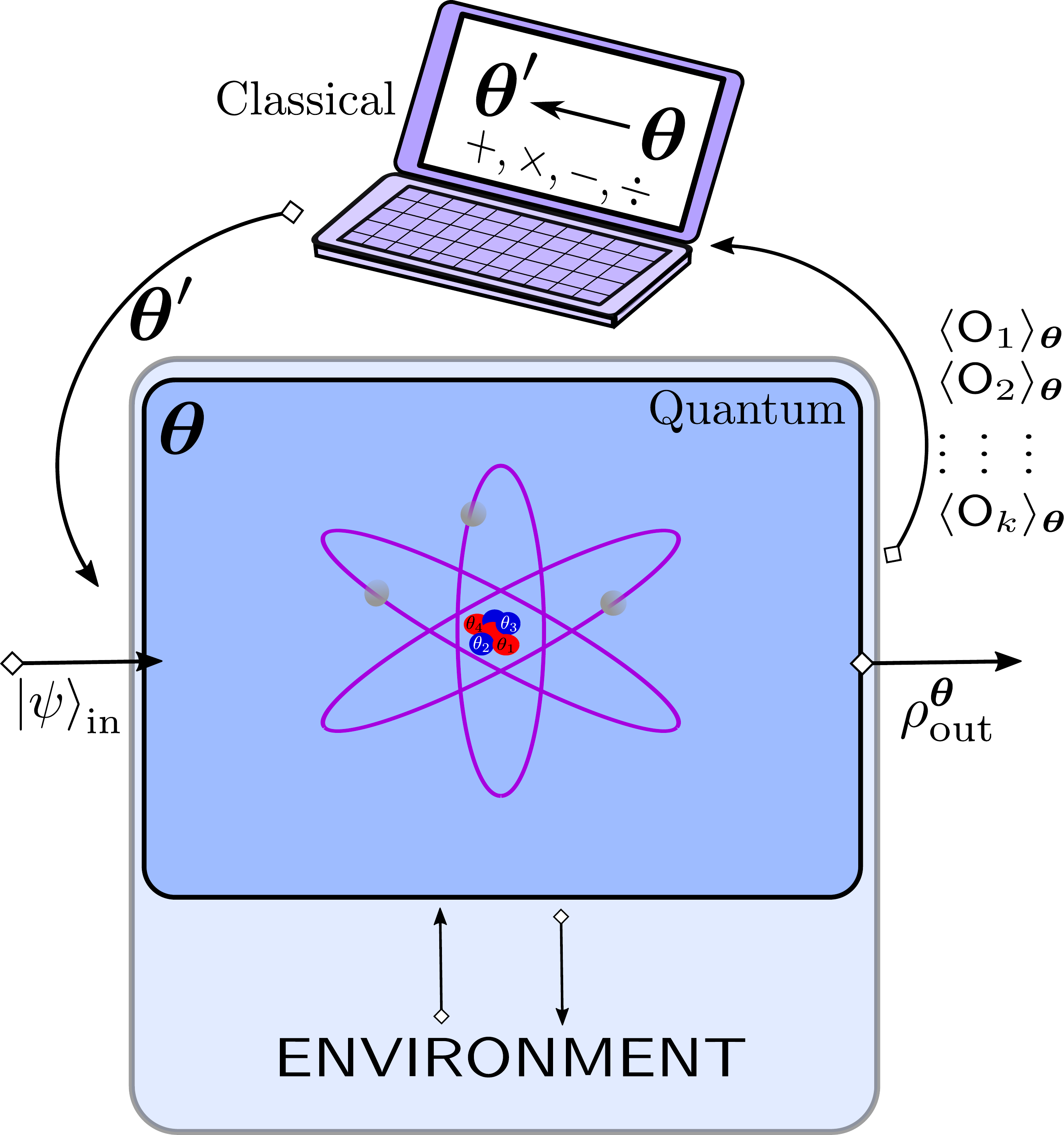}
\caption{Typical procedure of a variational algorithm. A quantum computer produces a parameterised quantum state, $\rho_\mathrm{out}^{\paramtheta}$, given some input $\ket{\psi}_{\mathrm{in}}$, which is then used for some task. The parameters of the state are updated using a classical computer with $k$ observables, $\mathsf{O}$, measured from the quantum state. Most generally, the quantum computer can interact with some environment to implement general quantum transformations.} \label{fig:variational_algorithm}
\end{figure}
%
A prominent family of algorithms suitable for these quantum devices have become known as variational quantum algorithms (VQAs)~\cite{mcclean_theory_2016, biamonte_universal_2019, endo_hybrid_2020, wecker_progress_2015, cerezo_variational_2020_rev} that rely heavily on a synergistic relationship between classical and quantum co-processors. To maximise use of coherence time, the quantum component is utilised only for the subroutine which would be difficult (or impossible in reasonable time) for a purely classical device to implement. The core quantum component is typically made of a parameterised quantum circuit (PQC)\cite{sim_expressibility_2019}. When VQAs are applied to machine learning problems, they have come to be seen as quantum neural networks (QNN's) \cite{benedetti_parameterized_2019, killoran_continuous-variable_2019} in the flourishing field of \emph{quantum} machine learning\cite{wittek_quantum_2014, biamonte_quantum_2017, kopczyk_quantum_2018, schuld_supervised_2018} (QML). This is since they can achieve many of the same tasks as classical neural networks, \cite{mitarai_quantum_2018, grant_hierarchical_2018} and may outperform them in certain cases~\cite{wright_capacity_2020, coyle_born_2020, cong_quantum_2019}.

VQAs typically use a NISQ computer to evaluate an objective function and a classical computer to adjust input parameters to optimise said function. They have been proposed or used for many applications including quantum chemistry~\cite{peruzzo_variational_2014} in the variational quantum eigensolver (VQE) and combinatorial optimisation~\cite{farhi_quantum_2014} in the quantum approximate optimisation algorithm (QAOA).

More recently, an interesting line of work has focused on using variational techniques for quantum algorithm discovery (that is finding novel algorithms for specific tasks) ranging from learning Grover's algorithm~\cite{morales_variational_2018} and compiling quantum circuits~\cite{khatri_quantum-assisted_2019, jones_quantum_2018, heya_variational_2018} to solving linear systems of equations~\cite{bravo-prieto_variational_2019, xu_variational_2019, huang_near-term_2019} and even extending to the foundations of quantum mechanics~\cite{arrasmith_variational_2019}, among others~\cite{larose_variational_2019, carolan_variational_2020, bravo-prieto_quantum_2020, anschuetz_variational_2018, verdon_quantum_2019}. Deeper fundamental questions about the computational complexity~\cite{mcclean_theory_2016, biamonte_universal_2019}, trainability~\cite{mcclean_barren_2018, grant_initialization_2019, cerezo_cost-function-dependent_2020, arrasmith_effect_2020, cerezo_variational_2020, cerezo_variational_2020-1, vinci_optimally_2018, stokes_quantum_2020} and noise-resilience~\cite{sharma_noise_2020, larose_robust_2020, colless_computation_2018} of VQAs have also been considered. While all of the above are tightly related, each application and problem domain presents its own unique challenges, for example, requiring domain specific knowledge, efficiency, interpretability of solution etc. The tangential relationship of these algorithms to machine learning techniques also opens the door to the wealth of information and techniques available in that field~\cite{schmidhuber_deep_2015, goodfellow_deep_2016}. A parallel and related line of research has focused on purely classical machine learning techniques (for example reinforcement learning) to discover new quantum experiments\cite{krenn_automated_2016, melnikov_active_2018} and quantum communication protocols\cite{wallnofer_machine_2020}.

Here we extend the application of variational approaches in two directions, quantum foundations and quantum cryptography, by focusing on one concrete pillar of quantum mechanics: the no-cloning theorem. It is well known that cloning arbitrary quantum information perfectly and deterministically is forbidden by quantum mechanics~\cite{wootters_single_1982}.

Specifically, given a general quantum state, it is impossible produce two perfect `clones' of it, since any information extracting measurement would, by its nature, disturb the coherence of the quantum state. This is in stark contrast to classical information theory, in which one can deterministically read and copy classical bits.

Furthermore, the no-cloning theorem is a base under which much of modern quantum cryptography, for example quantum key distribution (QKD), is built. If an adversary is capable of intercepting and making perfect copies of quantum states sent between two parties communicating using some secret key (encoded in quantum information) they can, in principle, obtain full information of the secret key. The fact that the adversary is limited in such a way by a foundational quantum mechanical principle leads to many potential advantages in using quantum communication protocols, for example in giving information theoretic security guarantees. 

However, the discovery of Ref.~\cite{buzek_quantum_1996}, showed that, if one is willing to relax two assumptions in the no-cloning theorem then it is in fact possible to clone \emph{some} quantum information. Removing the requirement of `perfect' clones gives \emph{approximate} cloning, and relaxing determinism gives \emph{probabilistic} cloning.
Both of these sub-fields of quantum information have a rich history, and have been widely studied. For comprehensive reviews see Refs.~\cite{scarani_quantum_2005, fan_quantum_2014}. 

In this work, we revisit (approximate) quantum cloning using the tools of variational quantum algorithms with two viewpoints in mind:
\begin{enumerate}
    \item Find unknown optimal cloning fidelities for particular families of quantum states $\leftrightarrow$ quantum foundations.
    \item Improve practicality in cloning transformations for well studied scenarios $\leftrightarrow$ quantum cryptography.
\end{enumerate}
We refer to direction of work in quantum machine learning for quantum cryptography as \emph{variational quantum cryptanalysis}, and we can draw on the relationship between classical machine learning and deep learning, with classical cryptography\cite{ateniese_hacking_2015, maghrebi_breaking_2016,  papernot_towards_2016, alani_applications_2019}. Furthermore, we remark that the techniques developed in this work can enhance the toolkit of cryptographers in constructing secure quantum protocols by keeping in mind the realistic attack strategies we propose.

We mention that the work of Ref.\cite{jasek_experimental_2019} considered a similar idea of using a variational quantum circuit to learn the parameters in optimal phase covariant cloning circuits. In this work, we significantly expand on this idea to propose optimal cloning circuits for a wide class of cloning problems including universal/state-dependent cloning under a symmetric/asymmetric framework. We also significantly expand on the $\Ansatze$ and cost functions used to build a more flexible and powerful algorithm.

To these ends, the rest of the paper is organised as follows. In \secref{sec:quantum_cloning}, we discuss background on approximate quantum cloning, including figures of merit, and the two specific cases we consider. Next, in \secref{sec:variational}, we introduce the variational methods we use, including several cost functions, their gradients and their provable guarantees (notions of faithfulness and barren plateau avoidance). In \secref{sec:quantum_coing_flipping_and_cloning}, we discuss two quantum coin flipping protocols and describe cloning based attacks on them, making connections between cloning and quantum state discrimination. Finally, in \secref{sec:results} we present the results of VQC in learning to clone phase covariant and state-dependent states as examples, and elucidate the connection to the previous coin flipping protocols. We conclude in \secref{sec:discussion}.


\section{Quantum Cloning} \label{sec:quantum_cloning}

As discussed, the subfield of quantum cloning is a rich area of study since the discovery of approximate cloning~\cite{buzek_quantum_1996}. Throughout this manuscript, we focus on the motivation of quantum cryptographic attacks for perspective, but we stress the fundamental primitive is that of quantum cloning. As such, these tools are useful whenever the need to approximately clone quantum states rears its head. As a motivating example, let us assume that Alice wishes to transmit quantum information to Bob (for example to implement a quantum key distribution (QKD) protocol) but the channel is subject to one (or multiple) eavesdroppers, Eve(s), who wishes to adversarially gain knowledge about the information sent by Alice. This scenario is useful to motivate the example of phase covariant cloning, which we discuss in \secref{ssec:phase_cov_cloning}, but when discussing state-dependent cloning and quantum coin flipping in \secref{sec:quantum_coing_flipping_and_cloning}, we will drop Eve, and have Bob as an adversary.

We illustrate this in \figref{fig:cloning}, for a single Eve. In this picture, Alice ($A$) sends a quantum state\footnote{This will typically be a pure state, $\rho_A := \ketbra{\psi}{\psi}_A$, but it can also be generalised to include mixed states, in which the task is referred to as \emph{broadcasting}\cite{barnum_noncommuting_1996, chen_mixed_2007, dang_optimal_2007}. The \emph{no-broadcasting} theorem is a generalisation of the no-cloning theorem in this setting.}, $\rho_A$, to Bob ($B$). A cloning based attack strategy for Eve ($E$) could be to try and clone Alice's state, producing a second (approximate) copy which she can use later in her attack, with some ancillary register ($E^*$).

\subsection{Properties of Cloning Machines}
There are various quantities to take into account when comparing quantum cloning machines\footnote{Meaning the unitary that Eve implements to perform the cloning.} (QCMs). The three most important and relevant for us are:
\begin{enumerate}
    \item Universality. 
    \item Locality.
    \item Symmetry.
\end{enumerate}
These three properties manifest themselves in the \emph{comparison metric} which is used to compare the clones outputted from the QCM, relative to the ideal input states.

\emph{Universality} refers to the family of states which QCMs are built for, $\mathcal{S}  \subseteq \mathcal{H}$ ($\mathcal{H}$ is the full Hilbert space), as this has a significant effect of their performance. Based on this, QCMs are typically subdivided into two categories, \emph{universal} (UQCM), and  \emph{state dependent} (SDQCM). In the former case \emph{all} states must be cloned equally well ($\mathcal{S} = \mathcal{H}$). In the latter, the cloning machine will depend on the family of states fed into it, so $\mathcal{S} \subset \mathcal{H}$. From a cryptographic point of view, Eve can gain substantial advantages by catering her cloning machine to any partial information she may have (for example, if she knows Alice is sending states from a specific family). 

By \emph{locality}, we mean whether the QCM optimises a local comparison measure (i.e.\@ check the quality of \emph{individual} output clones - a \emph{one-particle test criterion}\cite{werner_optimal_1998, scarani_quantum_2005}) or a global one (i.e.\@ check the quality of the global output state from the QCM - an \emph{all-particle test criterion}).

Finally, in \emph{symmetric} QCMs, we require each `clone' outputted from the QCM to be the same relative to the comparison metric, however asymmetric output is also sometimes desired. This property obviously only applies to local QCMs.  By varying this symmetry, an adversary can choose to tradeoff between their success in gaining information, and their likelihood of detection.

As our comparison metric, we will use the \emph{fidelity}\cite{jozsa_fidelity_1994}, defined between quantum states $\rho, \sigma$ as follows:
\begin{equation} \label{eqn:fidelity_definition}
    F(\rho, \sigma) = \left(\Tr\sqrt{\sqrt{\rho} \sigma \sqrt{\rho}}\right)^2
\end{equation}
which is symmetric with respect to $\rho$ and $\sigma$ and reduces to the state overlap if one of the states is pure.

A UQCM will maximise the fidelity over all possible input states, whereas a SDQCM will only be able to maximise it with respect to the input set, $S$.

The \emph{local} fidelity, $F^j_{\Lbs} := F_{\textsf{L}}(\sigma_j, \rho_A),  j \in\{B, E\}$ compares the ideal input state, $\rho_A$, to the output clones, $\sigma_j$, i.e.\@ the reduced states of the QCM output. In contrast, the \emph{global} fidelity  compares the entire output state of the QCM to a product state of input copies, $F_{\textsf{G}}(\sigma_{BE}, \rho_A \otimes \rho_A)$.
It may seem at first like the most obvious choice to study is the local fidelity, however (as we discuss at length over the next sections) the global fidelity is a relevant quantity for some cryptographic protocols, and we explicitly demonstrate this for the quantum coin flipping protocol of Aharonov et.\@ al.\@ \cite{aharonov_quantum_2000}

In general, Alice can send $M > 1$ copies of the input state. In this case, we can model the cloning task with $M$ `Bobs' $B_1, \dots, B_M$ and $N-M$ `Eves', $E_{N-M}, \dots, E_N$, whose job would be to create $N > M$ approximate clones (known as $N\rightarrow M$ cloning) of the state and return $M$ approximate clones to the Bobs.

Finally, by enforcing \emph{symmetry} in the output clones, we require that
\begin{equation}
F^j_{\Lbs} = F^k_{\Lbs}, \qquad \forall j, k \in \{1, \dots N\}.
\end{equation}
We consider all of these properties when constructing our variational algorithm.

\begin{figure}[ht]
    \centering
\includegraphics[width=\columnwidth, height=0.5\columnwidth]{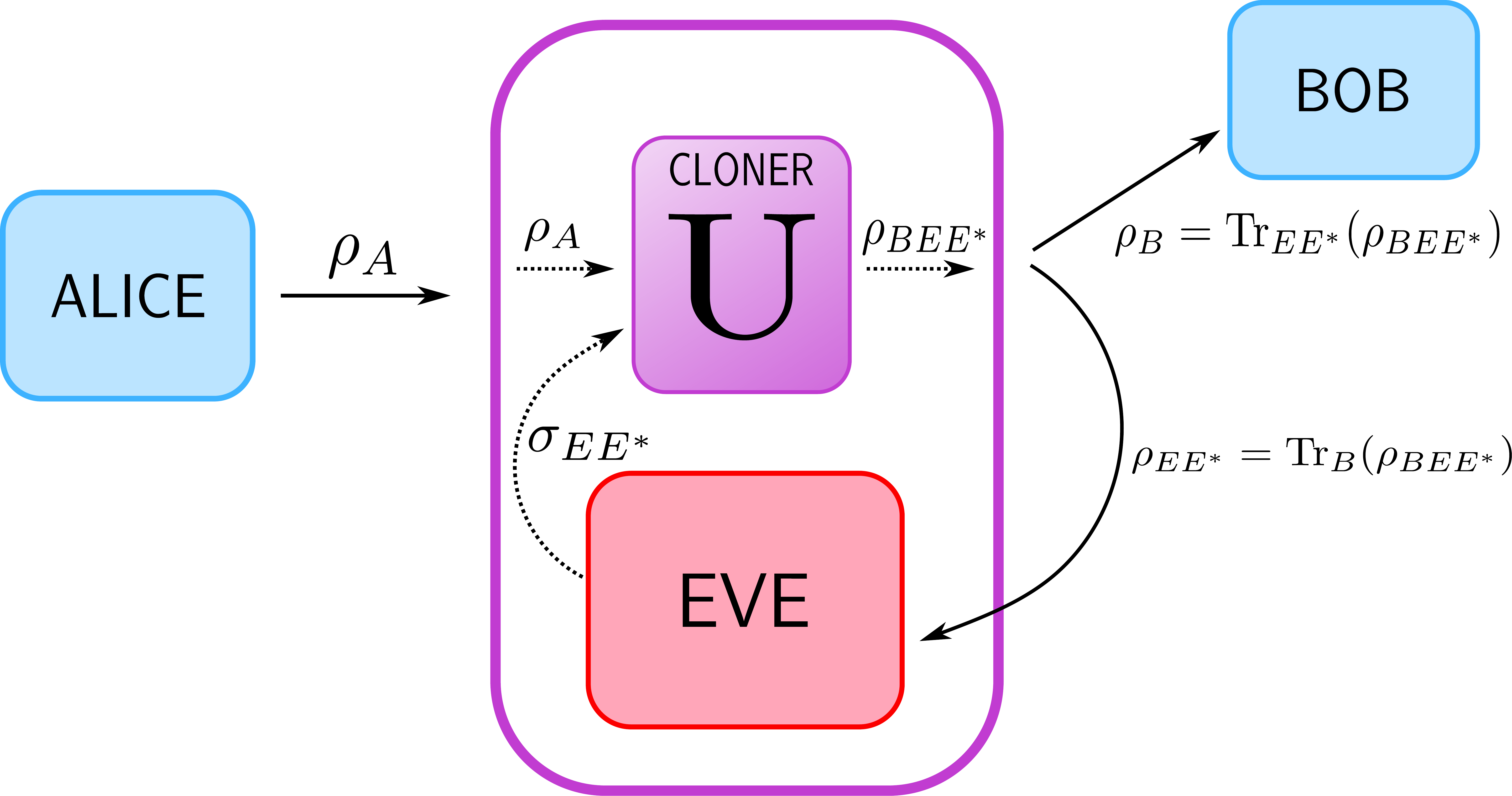}
\caption{Cartoon illustration of an eavesdropping attack by Eve, trying to clone the state, $\rho_A$, Alice sends to Bob. Eve injects a `blank' state (which can be a specific state or an arbitrary state depending on the scenario, which ends up as a clone of $\rho_A$) and an ancillary system ($E^*$). She then applies the cloning unitary, $U$. We can assume the cloner is manufactured by Eve, to give her the greatest advantage. The state Bob receives will be the partial trace over Eve's subsystems, $\rho_{B} = \Tr_{EE^*}(\rho_{BEE^*})$, and Eve's clone will be $\rho_E$, where $\rho_{BEE^*}$ is the full output state from the QCM.} 
\label{fig:cloning}
\end{figure}

\subsection{Phase-Covariant Cloning}\label{ssec:phase_cov_cloning}
The earliest result\cite{buzek_quantum_1996} in approximate cloning was that a universal symmetric cloning machine for qubits can be designed to achieve an optimal cloning fidelity of $ 5/6 \approx 0.8333$, which is notably higher than trivial copying strategies\cite{scarani_quantum_2005}. In other words, if Eve is required to clone \emph{every} single qubit state in the Bloch sphere equally well, the best local fidelity she and Bob can jointly receive is $F_{\mathsf{L}, \textrm{opt}}^{\text{U}, \mathrm{B}} = F_{\mathsf{L}, \textrm{opt}}^{\text{U}, \mathrm{E}} = 5/6$.

However, as mentioned above, Eve can do better still if she has some knowledge of the input state. For example, if Alice sends only \emph{phase-covariant}\cite{brus_phase-covariant_2000} ($\mathsf{X} - \mathsf{Y}$ plane in the Bloch sphere) states of the form:
\begin{equation} \label{eqn:x_y_plane_states}
    \ket{\psi_{xy}(\eta)} = \frac{1}{\sqrt{2}}\left(\ket{0} + e^{i\eta}\ket{1}\right)
\end{equation}
Then Eve can construct a cloning machine with fidelity $F_{\mathsf{L},  \text{opt}}^{\text{PC}, \mathrm{E}} \approx 0.85 > 5/6$.

These states are relevant since they are used in BB-84 QKD protocols and also in universal blind quantum computation protocols\cite{bennett_quantum_2014, broadbent_universal_2009}. Interestingly, the cloning of phase-covariant states can be accomplished in an \emph{economical} manner, meaning without needing an ancilla system for Eve, $E^*$\cite{niu_two-qubit_1999}. However, as noted in Refs.\cite{scarani_quantum_2001, scarani_quantum_2005}, removing the ancilla is useful to reduce resources if one is \emph{only} interested in performing cloning, but if Eve wishes to attack Alice and Bobs communication, it is more beneficial to apply an ancilla-based attack. Intuitively, this is because the ancilla also contains information of the input state which Eve can extract.
Of interest to our purposes, is an explicit quantum circuit which implements the cloning transformation. A unified circuit\cite{buzek_quantum_1997, fan_quantum_2014, fan_quantum_2001} for the above cases (universal and $\mathsf{X} - \mathsf{Y}$ phase covariant) can be seen in \figref{fig:qubit_cloning_ideal_circ}.  The parameters of the circuit, $\boldsymbol{\alpha} = \{\alpha_1, \alpha_2, \alpha_3\}$, are given by the family of states the circuit is built for~\cite{buzek_quantum_1997, fan_quantum_2014, fan_quantum_2001}.

\begin{figure}[ht]
    \centering
\includegraphics[width=0.95\columnwidth]{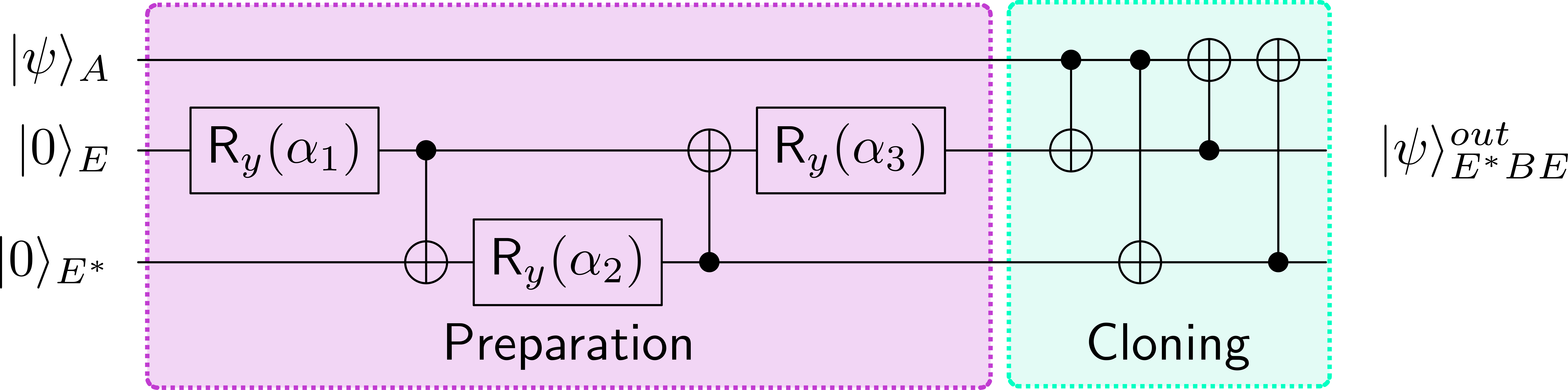}
\caption{Ideal cloning circuit for universal and phase covariant cloning. The \textsf{Preparation} circuit prepares Eve's system to receive the cloned states, while the \textsf{Cloning} circuit transfers information. Notice, the output registers which contain the two clones of $\ket{\psi}_A$ to Bob and Eve in this circuit are registers $2$ and $3$ respectively.} \label{fig:qubit_cloning_ideal_circ}
\end{figure}

\subsection{Cloning of Fixed Overlap States}\label{ssec:non-ortho-states}

In the above examples, the side information available to Eve is their inhabitance of a particular plane in the Bloch sphere. Alternatively, Alice may want to implement a protocol using states which have a fixed overlap\footnote{Confusingly, cloning states with this property is historically referred to as `state-dependent', so we herein use this term referring to this scenario.}. This was one of the original scenarios studied in the realm of approximate cloning \cite{brus_optimal_1998}, and has been used to demonstrate advantage related to quantum contextuality\cite{lostaglio_contextual_2020}, but is difficult to tackle analytically. For example, one may consider two states of the type:
\begin{equation} \label{eqn:state_dependent_cloning_states}
\begin{split}
    \ket{\psi_1} = \cos\phi\ket{0} +\sin\phi\ket{1}\\
    \ket{\psi_2} = \sin\phi\ket{0} +\cos\phi\ket{1} 
\end{split}
\end{equation}%
which have a fixed overlap, $s = \braket{\psi_1}{\psi_2} = \sin 2\phi$.

 The optimal local fidelity for this scenario\cite{brus_optimal_1998} is:
\begin{equation} \label{eqn:local_optimal_non_ortho_fidelity_1to2}
\begin{split}
    &F^{\textrm{FO}, j}_{\mathsf{L}, \textrm{opt}}  = \frac{1}{2} + \frac{\sqrt{2}}{32 s}(1+s)(3−3s+\sqrt{1−2s+ 9s^2})\\
    & \times\sqrt{−1 + 2s + 3s^2 + (1−s)\sqrt{1−2s+ 9s^2}}, ~j \in \{B, E\}
\end{split}
\end{equation}
It can be shown that the \emph{minimum} value for this expression is achieved when $s=\frac{1}{2}$ and gives $F^{\textrm{FO}, j}_{\mathsf{L}, \textrm{opt}} \approx 0.987$, which is much better than the symmetric phase-covariant cloner. We will use this scenario as a case study for quantum coin flipping protocols.

\subsection{Cloning with Multiple Input States} \label{ssec:m_to_n_cloning} 

As mentioned above, we can also provide multiple ($M$) copies of an states to the cloner and request $N$ output approximate clones. This is referred to as $M\rightarrow N$ cloning\cite{gisin_optimal_1997}. In the limit $M\rightarrow \infty$, an optimal cloning machine becomes equivalent to an quantum state estimation machine\cite{scarani_quantum_2005} for universal cloning. In this case, the optimal local fidelity becomes:
\begin{equation} \label{eqn:mton_universal_optimal_fidelity}
    F^{\textrm{U}, j}_{\Lbs, \textrm{opt}}(M, N) = F^{\textrm{U}}_{\Lbs, \textrm{opt}}(M, N) = \frac{M}{N} + \frac{(N-M)(M+1)}{N(M+2)}
\end{equation}
We will primarily focus here on the example of state dependent cloning of the states in \eqref{eqn:state_dependent_cloning_states}, and we explicitly revisit this case in the numerical results in \secref{sec:results}.

In this procedure, the adversaries $E_1\dots E_N$ receive $M$ copies of either $\ket{\psi_1}$ or $\ket{\psi_2}$, and use $N - M$ ancilla qubits to assist, so the initial state is $\ket{\psi_i}^{\otimes M}\otimes\ket{0}^{\otimes N-M}, i = 1, 2$.

For these states, the optimal \emph{global} fidelity of cloning the two states in \eqref{eqn:state_dependent_cloning_states} is given by:
\begin{equation}\label{eqn:optimal_global_non_ortho_state_fidelity}
    F^{{\mathrm{FO}}}_{\Gbs, \textrm{opt}}(M,N) = \frac{1}{2}\left( 1 + s^{M+N} + \sqrt{1-s^{2M}}\sqrt{1-s^{2N}} \right)
\end{equation}
Interestingly, it can be shown that the SDQCM which achieves this optimal \emph{global} fidelity, does not actually saturate the optimal \emph{local} fidelity (i.e.\@ the individual clones do not have a fidelity given by \eqref{eqn:local_optimal_non_ortho_fidelity_1to2}). Instead, computing the local fidelity for the globally optimised SDQCM gives\cite{brus_optimal_1998}:
\begin{multline}\label{eqn:state_dep_local_fidelity_from_global}
    F_{\Lbs}^{\mathrm{FO}, *}(M, N) = \frac{1}{4}\left(
    \frac{1+s^M}{1+s^N}\left[1+s^2+2s^N\right] +\right.\\
     \left.\frac{1-s^M}{1-s^N}\left[1+s^2-2s^N\right] +
     2\frac{1-s^{2M}}{1-s^{2N}}\left[1-s^2\right]
    \right)    
\end{multline}
which (taking $M=1, N=2$) is actually a \emph{lower} bound for the optimal local fidelity, $F^{\textrm{FO}, j}_{\mathsf{L}, \textrm{opt}} $ in \eqref{eqn:local_optimal_non_ortho_fidelity_1to2}. This point is crucially relevant in our designs for a variational cloning algorithm, and affects our ability to prove faithfulness arguments as we discuss later.


\section{Variational Quantum Cloning: Cost Functions and Gradients} \label{sec:variational}

Now we are in a position to outline details of our variational quantum cloning algorithm. To reiterate, our motivation is to find short-depth circuits to clone a given family of states, and also use this toolkit to investigate state families where the optimal figure of merit is unknown.

As illustrated in \figref{fig:variational_algorithm}, a variational method uses a parameterised state denoted by $\rho_{\paramtheta}$, typically prepared by some short-depth parameterised unitary on some initial state, $\rho_{\paramtheta} := U(\paramtheta)\ket{0}\bra{0}U^{\dagger}(\paramtheta)$. The parameters are then optimised by minimising (or maximising) a \emph{cost function}, typically a function of $k$ observable measurements on $\rho_{\paramtheta}$, $\mathsf{O}_k$. Since this resembles a classical neural network, techniques and ideas from classical machine learning can be borrowed and adapted.

The variational approach has been useful in quantum information, and has been applied successfully to learn quantum algorithms. More interestingly, given the flexibility of the method, it can learn alternate versions of quantum primitives, or even \emph{improved} versions in some cases to achieve a particular task. For example the work of Ref.\cite{cincio_learning_2018} found alternative and novel methods to compute quantum state overlap and Ref.\cite{cincio_machine_2020} is able to learn circuits which are better suited to a given hardware. We adopt these techniques here and find comparable results. An overview of the main ingredients of VQC is given in \figref{fig:vqc_overview}.

\begin{figure*}
    \includegraphics[width=1.9\columnwidth,height=0.3\textwidth]{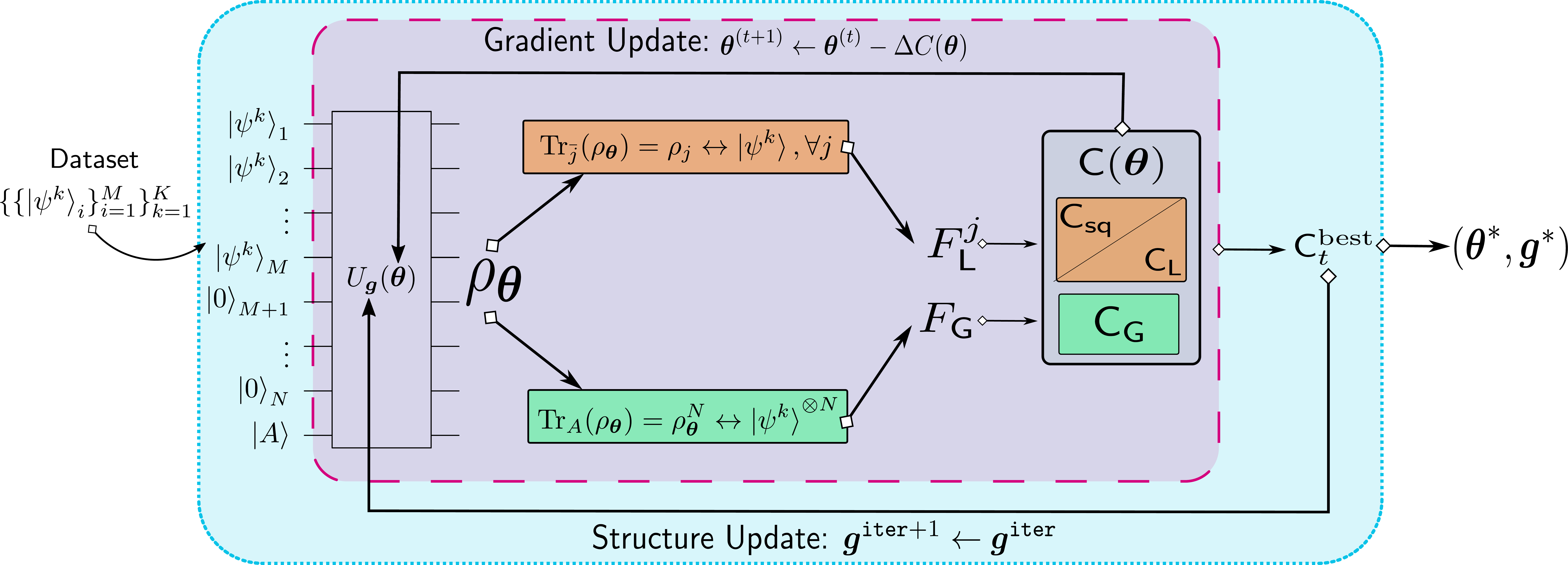}
    \caption{Illustration of VQC for $M \rightarrow N$ cloning. A dataset of $K$ states is chosen from $S$, with $M$ copies of each. These are fed with $N-M$ blank states, and possibly another ancilla, $\ket{\phi}_A$ into the variable structure Ansatz, $U_{\boldsymbol{g}}(\paramtheta)$. Depending on the problem, either the global, or local fidelities of the output state, $\rho_{\paramtheta}$, is compared to the input states, $\ket{\psi^k}$, and the corresponding local or global cost function, $\Cbs(\paramtheta)$ is computed, along with its gradient. We have two optimisation loops, one over the continuous parameters, $\paramtheta$, by gradient descent, and the second over the circuit structure, $\boldsymbol{g}$. Gradient descent over $\paramtheta$ in each structure update step outputs, upon convergence, the `minimum' cost function value, $\Cbs^{\textrm{best}}_t$, for the chosen cost function, $t \in \{\Lbs, \sq, \Gbs\}$.}
    \label{fig:vqc_overview}
\end{figure*}

%

\subsection{Cost Functions} \label{ssec:cost_functions}
We introduce three cost functions here suitable for approximate cloning tasks, one inspired by Ref.~\cite{jasek_experimental_2019}, and the other two adapted from the literature on variational algorithms\cite{larose_variational_2019, cerezo_cost-function-dependent_2020, cerezo_variational_2020, khatri_quantum-assisted_2019, sharma_noise_2020}. We begin by stating the functions, and then discussing the various ingredients and their relative advantages.

The first, we call the `\emph{local cost}', given by:
\begin{align} \label{eqn:local_cost_full}
    \Cbs_{\Lbs}^{M\rightarrow N}(\paramtheta) &:= \mathop{\mathbb{E}}_{\substack{\ket{\psi} \in \mathcal{S}}} \left[\Cbs^{\psi}_{\Lbs}(\paramtheta) \right] :=  \mathop{\mathbb{E}}_{\substack{\ket{\psi} \in \mathcal{S}}} \left[\Tr( \mathsf{O}^{\psi}_{\Lbs} \rho_{\paramtheta})\right] \\
    \mathsf{O}^{\psi}_{\Lbs}  &:= \mathds{1} - \frac{1}{N}\sum\limits_{j=1}^N\ketbra{\psi}{\psi}_j \otimes \mathds{1}_{\Bar{j}}
\end{align}
where $\ket{\psi} \in \mathcal{S}$ is the family of states to be cloned. The subscripts $j\ (\Bar{j})$ indicate operators acting on subsystem $j$ (everything except subsystem $j$) respectively.

The second, we refer to as the `\emph{squared local cost}' or just `squared cost' for brevity:
\begin{multline} \label{eqn:squared_local_cost_mton}
    \Cbs_{\sq}^{M\rightarrow N}(\paramtheta) := \mathop{\mathbb{E}}_{\substack{\ket{\psi} \in \mathcal{S}}}\left[ \sum\limits_{i=1}^N (1-F^i_{\Lbs}(\paramtheta))^2\right. \\
    \left.+ \sum\limits_{i<j}^N (F^i_{\Lbs}(\paramtheta)-F^j_{\Lbs}(\paramtheta))^2\right] 
\end{multline}

The notation $F^j_{\Lbs}(\boldsymbol{\theta}) := F_{\Lbs}(\ketbra{\psi}{\psi}, \rho^j_{\boldsymbol{\theta}})$ indicates the fidelity of Alice's input state, relative to the reduced state of output qubit $j$, given by $\rho^j_{\boldsymbol{\theta}} = \Tr_{\Bar{j}}\left(\rho_{\paramtheta}\right)$. These first two cost functions are related only in that they are both functions of \emph{local} observables, i.e.\@ the local fidelities.

The third and final cost is fundamentally different to the other two, in that it instead uses global observables, and as such, we refer to it as the `\emph{global cost}':
\begin{align} \label{eqn:global_cost_full}
    \Cbs_{\Gbs}^{M\rightarrow N}(\paramtheta) &:= \mathop{\mathbb{E}}_{\substack{\ket{\psi} \in \mathcal{S}}} \left[\Tr(\mathsf{O}^{\psi}_{\Gbs}\rho_{\paramtheta})\right] \\
     \mathsf{O}^{\psi}_{\Gbs} &:= \mathds{1} - \ketbra{\psi}{\psi}^{\otimes N}
\end{align}
For compactness, we will drop the superscript $M \rightarrow N$ when the meaning is clear from context. 

Now, we motivate our choices for the above cost functions. For \eqref{eqn:squared_local_cost_mton}, if we restrict to the special case of $1\rightarrow 2$ cloning (i.e.\@ we have only two output parties, $j\in \{B, E\}$), and remove the expectation value over states, we recover the cost function used in Ref.\cite{jasek_experimental_2019}. A useful feature of this cost is that symmetry is explicitly enforced by the difference term $(F_i(\paramtheta) - F_j(\paramtheta))^2$. 

In contrast, the local and global cost functions are inspired by other variational algorithm literature\cite{larose_variational_2019, cerezo_cost-function-dependent_2020, cerezo_variational_2020, khatri_quantum-assisted_2019, sharma_noise_2020} where their properties have been extensively studied, particularly in relation to the phenomenon of `barren plateaus'\cite{mcclean_barren_2018, cerezo_cost-function-dependent_2020}. It has been demonstrated that hardware efficient $\Ansatze$ are untrainable (with either differentiable or non-differentiable methods) using a global cost function similar to \eqref{eqn:global_cost_full}, since they have exponentially vanishing gradients\cite{schuld_evaluating_2019}. In contrast, local cost functions (\eqref{eqn:local_cost_full}, \eqref{eqn:squared_local_cost_mton}) are shown to be efficiently trainable with $\mathcal{O}(\log N)$ depth hardware efficient $\Ansatze$\cite{cerezo_cost-function-dependent_2020}. We explicitly prove this property also for our local cost (\eqref{eqn:local_cost_full}) in \appref{app_ssec:pc_cloning_fixed_hardware_efficient_ansatz}.

We also remark that typically global cost functions are usually more favourable from the point of view of \emph{operational meaning}, for example in variational compilation\cite{khatri_quantum-assisted_2019}, this cost function compares the closeness of two global unitaries. In this respect, local cost functions are usually used as a proxy to optimise a global cost function, meaning optimisation with respect to the local function typically provides insights into the convergence of desired global properties. In contrast to many previous applications, by the nature of quantum cloning, VQC allows the local cost functions to have immediate operational meaning, illustrated through the following example (using the local cost, \eqref{eqn:local_cost_full}) for $1\rightarrow 2$ cloning:
\begin{align*}
    \Cbs^{\psi}_{\Lbs}(\paramtheta) &= \Tr\left[\left(\mathds{1} - \frac{1}{2}\sum\limits_{j=1}^2\ketbra{\psi}{\psi}_j \otimes \mathds{1}_{\Bar{j}}\right) \rho_{\paramtheta}\right]\\
     \implies  \Cbs_{\Lbs}(\paramtheta)  &= 1 - \frac{1}{2}\mathbb{E}\left[F_{\Lbs}\left(\ketbra{\psi}{\psi}, \rho_{\paramtheta}^1\right) + F_{\Lbs}\left(\ketbra{\psi}{\psi}, \rho_{\paramtheta}^2\right)\right]
\end{align*}
where $\mathbb{E}[F_{\Lbs}]$ is the average fidelity\cite{scarani_quantum_2005} over the possible input states. The final expression of $\Cbs_{\Lbs}(\paramtheta)$ in the above equation follows from the expression of fidelity when one of the states is pure. Similarly, the global cost function relates to the global fidelity of the output state with respect to the input state(s).

In practice, we estimate the expectation values in the above by drawing $K$ input samples uniformly at random from $\mathcal{S}$. For example, $\Cbs_{\Lbs}(\paramtheta)$ in the above equation can be estimated with $K$ samples denoted by $\{\ket{\psi^k}\}^{K}_{k=1}$ as,

\begin{equation*}
\small
     \Cbs_{\Lbs}(\paramtheta)  \approx 1 - \frac{1}{2K}\sum_{k=1}^{K}\left[F_{\Lbs}\left(\ketbra{\psi^k}{\psi^k}, \rho_{\paramtheta}^1\right) + F_{\Lbs}\left(\ketbra{\psi^k}{\psi^k}, \rho_{\paramtheta}^2\right)\right]
\end{equation*}

\subsection{Cost Function Gradients} \label{ssec:cost_function_gradients}
Typically, in machine learning applications, the cost functions are ideally minimised using gradient based optimisation, a differentiable method requiring the computation of gradients of the cost function. In contrast, the work of Ref.\cite{jasek_experimental_2019} considers a black box, gradient-free training using the Nelder-Mead optimiser\cite{gill_practical_2019}. In this work, we opt for a gradient based approach and derive the analytic gradient for our cost functions. If we assume that the $\Ansatze$, $U(\paramtheta)$ is composed of unitary gates, with each parameterised gate of the form: $\exp(i\theta_l \Sigma)$, where $\Sigma$ is a generator with two distinct eigenvalues\cite{mitarai_quantum_2018, schuld_evaluating_2019}, then we can derive the gradient (see \appref{app_a:analytic_gradient}) of \eqref{eqn:squared_local_cost_mton}, with respect to a parameter, $\theta_l$:
\begin{multline}  \label{eqn:analytic_squared_grad_mton}
\small
    \frac{\partial \Cbs_{\sq}(\boldsymbol{\theta})}{\partial \theta_l}  = 
    2{\mathbb{E}}\left[\sum\limits_{i<j} (F^i_{\Lbs}- F^j_{\Lbs}) \times \right.\\
    \left. \left(F^{i, l+ \frac{\pi}{2} }_{\Lbs} - F^{i,  l- \frac{\pi}{2} }_{\Lbs} -F^{j,  l+ \frac{\pi}{2}}_{\Lbs} + F^{j, l-\frac{\pi}{2}}_{\Lbs} \right) \right.\\
    \left. - \sum\limits_{i} (1- F^i_{\Lbs})(F^{i, l+\frac{\pi}{2}}_{\Lbs}
    - F^{i, l- \frac{\pi}{2}}_{\Lbs})  \right]
\end{multline}

where $F^{j, l\pm\pi/2}_{\Lbs}(\boldsymbol{\theta})$ denotes the fidelity of a particular state, $\ket{\psi}$, with $j^{th}$ reduced state of the VQC circuit which has the $l^{th}$ parameter shifted by $\pm \pi/2$. We suppress the $\paramtheta$ dependence in the above. Using the same method, we can also derive the gradient of the local cost, \eqref{eqn:local_cost_full} with $N$ output clones as:
\begin{equation}\label{eqn:gradient_local_cost_full}
    \frac{\partial \Cbs_{\Lbs}(\paramtheta)}{\partial \theta_l} = \mathbb{E}\left(\sum_{i=1}^N \left[F_{\Lbs}^{i, l-\pi/2} -  F_{\Lbs}^{i, l+\pi/2} \right]\right)
\end{equation}
Finally, similar techniques result in the analytical expression of the gradient of the global cost function:
\begin{equation}\label{eqn:gradient_global_cost_full}
    \frac{\partial \Cbs_{\Gbs}(\paramtheta)}{\partial \theta_l} = \mathbb{E}\left(F_{\Gbs}^{l-\pi/2} - F_{\Gbs}^{l+\pi/2}\right)    
\end{equation}
where $F_{\Gbs}:= F(\ketbra{\psi}{\psi}^{\otimes N}, \rho_{\paramtheta})$ is the global fidelity between the parameterised output state and an $n$-fold tensor product of input states to be cloned.

%
\subsection{Asymmetric Cloning} \label{ssec:asymmetric cloning}
Our above local cost functions are defined in a way that they enforce \emph{symmetry} in the output clones. However, from the purposes of eavesdropping, this may not be the optimal attack for Eve to implement. In particular, she may wish to impose less disturbance on Bob's state so she reduces the chance of getting detected. This subtlety was first addressed in Refs.\cite{fuchs_information_1996, fuchs_optimal_1997}. 

For example, if she wishes to leave Bob with a fidelity parameterised by a specific value, $p$, $F^{p, B}_{\Lbs} = 1 - p^2/2$, we can define an asymmetric version of the $1\rightarrow 2$ squared local cost function in \eqref{eqn:squared_local_cost_mton}:
\begin{equation} \label{eqn:asymmetric_cost_function_maintext}
      \Cbs_{\Lbs, \textrm{asym}}(\paramtheta) := \mathbb{E}\left[F^{p, B}_{\Lbs} - F^{B}_{\Lbs}(\paramtheta)\right]^2 + \mathbb{E}\left[F^{p, E}_{\Lbs} - F^{E}_{\Lbs}(\paramtheta)\right]^2 \\
\end{equation}
The corresponding value for Eve's fidelity ($F^{p, E}_{\Lbs}$) in this case can be derived from the `\emph{no-cloning inequality}'\cite{scarani_quantum_2005}:
\begin{align} \label{eqn:no_cloning_inequality_maintext}
    &\sqrt{(1 - F^{p, B}_{\Lbs})(1 -  F^{p, E}_{\Lbs})} \geqslant F^{p, B}_{\Lbs}+  F^{p, E}_{\Lbs} -\frac{3}{2}  \\
    \implies &F^{p, E}_{\Lbs} = 1 - \frac{1}{4}(2 - p^2 - p\sqrt{4 - 3p^2}) 
\end{align}
where again the expectation is taken over the family of states to be cloned. The cost function in \eqref{eqn:asymmetric_cost_function_maintext} can be naturally generalised to for the case of $M \rightarrow N$ cloning. We note that $\Cbs_{\Lbs, \textrm{asym}}(\paramtheta)$ can also be used for symmetric cloning by enforcing $p = \frac{1}{\sqrt{3}}$ in \eqref{eqn:no_cloning_inequality_maintext}. However, it comes with an obvious disadvantage in the requirement to have knowledge of the optimal clone fidelity values a-priori. In contrast, our local cost functions (\eqref{eqn:local_cost_full}, \eqref{eqn:squared_local_cost_mton}) do not have this requirement, and thus are more suitable to be applied in general cloning scenarios.


\subsection{Cost Function Guarantees} \label{ssec:cost_function_guarantees}
We would like to have theoretical guarantees about the above cost functions in order to use them. Specifically, due to the nature of the problems in other variational algorithms, for example in quantum circuit compilation\cite{khatri_quantum-assisted_2019} or linear systems solving\cite{bravo-prieto_variational_2019}, the costs defined therein are \emph{faithful}, meaning they approach zero as the solution approaches optimality. 

Unfortunately, due to the hard limits on approximate quantum cloning, the above costs cannot have a minimum at $0$, but instead at some finite value (say $\Cbs^{\textrm{opt}}_{\Lbs}$ for the local cost). If one has knowledge of the optimal cloning fidelities for the problem at hand, then normalised cost functions with a minimum at zero can be defined. Otherwise, one must take the lowest value found to be the approximation of the cost minimum.

Despite this, we can still derive certain theoretical guarantees about them. Specifically, we consider notions of \emph{strong} and \emph{weak} faithfulness, relative to the error in our solution. Our goal is to provide statements about the \emph{generalisation performance} of the cost functions, by considering how close the states we output from our cloning machine are to those which would be outputted from the `\emph{optimal}' cloner, relative to some metrics. In the following, we denote $\rho_{\mathrm{opt}}^{\psi, j}$ ($\rho^{\psi, j}_{\paramtheta}$) to be the optimal (VQC learned) reduced state for qubit $j$, for a particular input state, $\ket{\psi}$. If the superscript $j$ is not present, we mean the global state of all clones. 
\begin{definition}[Strong Faithfulness]
    A cloning cost function, $\Cbs$, is strongly faithful if:
    \begin{equation}\label{eqn:strongly_faithful_cost_defn}
        \Cbs(\paramtheta) = \Cbs^{\mathrm{opt}} \implies \rho_{\paramtheta}^\psi = \rho_{\opt}^{\psi} \qquad \forall \ket{\psi} \in \mathcal{S}
    \end{equation}
    where $\Cbs^{\mathrm{opt}}$ is the minimum value achievable for the cost, $\Cbs$, according to quantum mechanics, and $ \mathcal{S}$ is the given set of states to be cloned.
\end{definition}
\begin{definition}[$\epsilon$-Weak Local Faithfulness]
    A local cloning cost function, $\Cbs_{\Lbs}$, is $\epsilon$-weakly faithful if:
    \begin{multline}\label{eqn:weakly_faithful_local_cost_defn}
        |\Cbs_{\Lbs}(\paramtheta) - \Cbs_{\Lbs}^{\mathrm{opt}}| \leqslant \epsilon \implies D(\rho_{\paramtheta}^{\psi, j}, \rho_{\opt}^{\psi, j}) \leqslant f(\epsilon) \\ \forall \ket{\psi} \in \mathcal{S}, \forall j 
    \end{multline}
\end{definition}
where $D( \cdot, \cdot)$ is a chosen metric in the Hilbert space between the two states and $f$ is a polynomial function.
\begin{definition}[$\epsilon$-Weak Global Faithfulness]
    A global cloning cost function, $\Cbs_{\Gbs}$, is $\epsilon$-weakly faithful if:
    \begin{equation}\label{eqn:weakly_faithful_global_cost_defn}
        |\Cbs_{\Gbs}(\paramtheta) - \Cbs_{\Gbs}^{\mathrm{opt}}| \leqslant \epsilon \implies D(\rho_{\paramtheta}^{\psi}, \rho_{\opt}^{\psi}) \leqslant f(\epsilon) \ \ \forall \ket{\psi} \in \mathcal{S}
    \end{equation}
\end{definition}
One could also define local and global versions of strong faithfulness, but this is less interesting so we do not focus on it here.

Next, we prove that our cost functions satisfy these requirements if we take the metric, $D$, to be the Fubini-Study\cite{fubini_sulle_1904, study_kurzeste_1905, nielsen_quantum_2010} metric, $\textrm{D}_{\mathrm{FS}}$, between $\rho$ and $\sigma$, and defined via the fidelity:
\begin{equation}
    \textrm{D}_{\mathrm{FS}}(\rho,\sigma) = \arccos\hspace{1mm} \sqrt{F(\rho, \sigma)}
    \label{eqn:fubini_study_distance_maintext}
\end{equation}

We also state the theorems for weak faithfulness, which are less trivial, and present the discussion about strong faithfulness in \appref{app_sec:faithfulness}. We state the weak faithfulness theorem specifically for the $M \rightarrow N$ squared cost function, (\eqref{eqn:squared_local_cost_mton}) and present the results for the local, global and asymmetric (\eqref{eqn:asymmetric_cost_function_maintext}) costs in \appref{app_sec:faithfulness} since they are analogous. For this case, we find the following:
\begin{theorem}\label{thm:squared_local_cost_squared_FS_weak_faithful}
The squared cost function as defined \eqref{eqn:squared_local_cost_mton}, is $\epsilon$-weakly faithful with respect to $\textrm{D}_{\mathrm{FS}}$.
If the squared cost function, \eqref{eqn:squared_local_cost_mton}, is $\epsilon$-close to its minimum, i.e.\@:
\begin{equation}\label{eq:squared_cost_to_epsilon_maintext}
    \Cbs_{\sq}(\paramtheta) - \Cbs^{\mathrm{opt}}_{\sq} \leqslant \epsilon
\end{equation}
where $\Cbs^{\textrm{opt}}_{\sq} := \underset{\paramtheta}{\textrm{min}}\sum\limits_{i}^N (1-F_i(\paramtheta))^2 + \sum\limits_{i<j}^N  (F_i(\paramtheta)-F_j(\paramtheta))^2 = N(1-F_{\mathrm{opt}})^2$ is the optimal theoretical cost, then:
\begin{multline}     \label{eqn:fubini_study_bound_squared_maintext}
    \textrm{D}_{\mathrm{FS}}(\rho^{\psi,j}_{\paramtheta}, \rho^{\psi,j}_{\text{opt}}) \leqslant \frac{\mathcal{N}\epsilon}{2(1 - F_{\mathrm{opt}})\sin(F_{\mathrm{opt}})} := f_1(\epsilon), \\ \forall \ket{\psi} \in \mathcal{S}, \forall j \in [N]
\end{multline}
\end{theorem}
Furthermore, when $\rho^{\psi, j}_{\paramtheta}\in \mathcal{H}$ where $\textrm{dim}(\mathcal{H})=2$, we get the following: 

\begin{theorem} \label{thm:squared_local_cost_squared_trace_weak_faithful}
The squared cost function, \eqref{eqn:squared_local_cost_mton}, is $\epsilon$-weakly faithful with respect to the trace distance $\textrm{D}_{\Tr}$.
\begin{equation}     \label{eqn:trace_distance_bound_squared_maintext}
        \textrm{D}_{\Tr}(\rho_{\opt}^{\psi, j},  \rho^{\psi, j}_{\paramtheta})  \leqslant g_1(\epsilon), \qquad \forall j \in [N]
\end{equation}
where:
\begin{equation} \label{eqn:trace_distance_squared_cost_bound_function}
   g_1(\epsilon) \approx \frac{1}{2}\sqrt{4F_{\mathrm{opt}}(1 - F_{\mathrm{opt}}) + \epsilon\frac{\mathcal{N}(1 - 2F_{\mathrm{opt}})}{2(1 - F_{\mathrm{opt}})}}
\end{equation}

\end{theorem}
and $D_{\Tr}(\rho, \sigma)$ is the trace distance between $\rho, \sigma$
\begin{equation}\label{eqn:trace_distance_definition}
    D_{\Tr}(\rho, \sigma) := \frac{1}{2}||\rho-\sigma||_1 := \frac{1}{2}\Tr\left(\sqrt{(\rho-\sigma)^{\dagger}(\rho - \sigma)}\right)
\end{equation}
$\mathcal{N} = \int_{\ket{\psi} \in \mathcal{S}}d\psi$ is a normalisation factor which depends on the family of states to be cloned. For example, when cloning phase-covariant states in \eqref{eqn:x_y_plane_states}, we have $\mathcal{N}=4\pi$ and:
\begin{equation} \label{eqn:fubini_study_bound_phase_cov_qubits_explicit}
     \mathrm{D}_{\mathrm{FS}}(\rho^{\psi,j}_{\paramtheta}, \rho^{\psi,j}_{\text{opt}}) \leqslant 56 \cdot \epsilon
\end{equation}
An immediate consequence of \eqref{eqn:trace_distance_squared_cost_bound_function} is that there is a non-vanishing gap of $\sqrt{F_{\mathrm{opt}}(1 - F_{\mathrm{opt}})}$ between the states. This is due to the fact that the two output states are only $\epsilon$-close in distance when projected on a specific state $\ket{\psi}$. However, the trace distance is a projection independent measure and captures the maximum over all the projectors.

\subsection{Global versus Local Faithfulness} \label{ssec:global_v_local_faithfulness}

The last remaining point to address is the relationship between the global and local cost functions, $\Cbs_{\Lbs}$ in \eqref{eqn:local_cost_full} and $\Cbs_{\Gbs}$ in \eqref{eqn:global_cost_full}. 

In \appref{app_ssec:global_faithfulness}, we prove similar strong, and $\epsilon$-weak faithfulness guarantees as with the local cost functions above, namely, if we assume $|\Cbs_{\Gbs}(\paramtheta) - \Cbs^{\textrm{opt}}_{\Gbs}| \leq \epsilon$, then:
\begin{align} \label{eqn:global_cost_FS_bound_maintext}
    &\textrm{D}_{\mathrm{FS}}(\rho^{\psi}_{\textrm{opt}}, \rho^{\psi}_{\paramtheta}) \leqslant \mathcal{N}\epsilon/ \sin(F_{\Gbs}^{\textrm{opt}})\\
    \label{eqn:global_cost_trace_bound_maintext}
    &\textrm{D}_{\mathrm{Tr}}(\rho^{\psi}_{\textrm{opt}}, \rho^{\psi}_{\paramtheta}) \leqslant g_3(\epsilon), \qquad \forall \ket{\psi} \in \mathcal{S}
\\
    &g_3(\epsilon) := \frac{1}{2}\sqrt{4F^{\textrm{opt}}_{\Gbs}(1 - F^{\textrm{opt}}_{\Gbs}) + \mathcal{N}\epsilon(1 - 2F^{\textrm{opt}}_{\Gbs})} \nonumber 
\end{align}
We also note that \eqref{eqn:global_cost_FS_bound_maintext} and \eqref{eqn:global_cost_trace_bound_maintext} are taken with respect to the \emph{global} state of all clones (tracing out any ancillary system).

The relationship between $\Cbs_{\Lbs}$ and $\Cbs_{\Gbs}$ is fundamentally important for the works of Refs.\cite{larose_variational_2019, khatri_quantum-assisted_2019, cerezo_variational_2020} since in those cases, the local cost function is defined only as a proxy to train with respect to the global cost, in order to avoid barren plateaus\cite{cerezo_cost-function-dependent_2020}. Nevertheless, for applications in which the global cloning fidelity is relevant, the question is important.

Firstly, we can show a similar bound between the costs as Refs.\cite{khatri_quantum-assisted_2019, bravo-prieto_variational_2019}:
\begin{theorem}\label{thm:relationship_local_global_maintext}
For  $M \rightarrow N$ cloning, the global cost function $\Cbs_{\Gbs}(\paramtheta)$ and the local cost function $\Cbs_{\Lbs}(\paramtheta)$ satisfy the inequality,
\begin{equation} \label{eqn:local_and_global_cost_relationship_inequality_maintext}
 \Cbs_{\Lbs}(\paramtheta) \leqslant \Cbs_{\Gbs}(\paramtheta) \leqslant N\cdot\Cbs_{\Lbs}(\paramtheta)  
\end{equation}
\end{theorem}

However, we notice a subtlety here arising from the fact that our costs do not have a minimum at zero, but instead at whatever \emph{finite} value is permitted by quantum mechanics. This fact does not allow us to make direct use of \eqref{eqn:local_and_global_cost_relationship_inequality_maintext}, since even if the gap between the global cost function and its optimum is arbitrarily small, the corresponding gap between the local cost and its minimum may be finite in general (see \appref{app_ssec:global_vs_local_cost_relationship}). 

However, this does not discount the ability to make such statements of faithfulness in \emph{specific} cases. In particular we show that the global cost function allows us to make arguments about strong \emph{local} faithfulness for universal and phase-covariant cloning. The proofs of the following theorems can be found in \appref{app_ssec:local_faithfulness_from_global_optimisation}:

\begin{theorem} \label{thm:local_clones_from_global_cost_universal_maintext_universal}
The \emph{global} cost, $\Cbs_{\Gbs}$, is \emph{locally} strongly faithful for a \emph{symmetric universal} $M\rightarrow N$ cloner, i.e.:
\begin{multline}\label{eqn:strongly_faithful_global_cost_local_clone_strong_maintext_universal}
    \Cbs_{\Gbs}(\paramtheta) = \Cbs^{\mathrm{opt}}_{\Gbs} \iff \rho_{\paramtheta}^{\psi, j} = \rho_{\opt}^{\psi, j} , \forall j \in [N], \forall \ket{\psi} \in \mathcal{H}.
        \end{multline}
\end{theorem}
\begin{theorem} \label{thm:local_clones_from_global_cost_universal_maintext_phase_covariant}
The \emph{global}  cost, $\Cbs_{\Gbs}$, is \emph{locally} strongly faithful for a \emph{symmetric phase-covariant} $1\rightarrow 2$ cloner:
\begin{multline}\label{eqn:strongly_faithful_global_cost_local_clone_strong_maintext_phase_covariant}
    \Cbs_{\Gbs}(\paramtheta) = \Cbs^{\mathrm{opt}}_{\Gbs} \iff \rho_{\paramtheta}^{\psi, j} = \rho_{\opt}^{\psi, j} , j \in \{1, 2\} \\
    \forall \ket{\psi} \in \{2^{-1/2}(\ket{0} + e^{\mathrm{i}\eta}\ket{1})\}_{\eta}.
\end{multline}
\end{theorem}
To prove these theorems, we use the uniqueness of the optimal global and local cloning machines\cite{werner_optimal_1998, keyl_optimal_1999} and also the formalism proposed by Cerf et. al. ~\cite{cerf_cloning_2002} for the case of \thmref{thm:local_clones_from_global_cost_universal_maintext_phase_covariant}. As a result, depending on the problem, we may make arguments that optimising one cost function is a useful proxy for optimising another. In contrast, for other cases of state-dependent cloning, since the optimal cloning operation differs depending on the figure of merit\cite{brus_optimal_1998}, we cannot provide such guarantees.

\subsection{Sample Complexity} \label{ssec:sample_complexity}

As a penultimate note on the cost functions, we remark that although VQC is a heuristic algorithm, we can nevertheless provide guarantees on the required number of input state samples, $\ket{\psi}$, to estimate the cost function value at each $\paramtheta$. Using  H\"{o}effding's inequality\cite{hoeffding_probability_1963}, we establish that the number of samples required is independent of the size of the distribution $\mathcal{S}$, and only depends on the desired accuracy of the estimate and confidence level (proof in \appref{app_ssec:sample_complexity_appendix}):
\begin{theorem}\label{thm:sample-complexity_main_text}
The number of samples, $\mathrm{Samp}$, required to estimate a cost function $\Cbs(\paramtheta)$ up to $\gamma$-additive error with a success probability $\delta$ is,
\begin{equation}\label{eqn:number-of-samples_main_text}
    \mathrm{Samp} =  L \times K =  \mathcal{O}\left(\gamma^{-2}\log\left(2/\delta\right)\right)
\end{equation}
where $K$ is the number of distinct states $\ket{\psi}$ randomly picked from the distribution $\mathcal{S}$, and $L$ is the number of copies of each of the $K$ input states. 
\end{theorem}


\section{Quantum Coin Flipping Protocols and Cloning Attacks} \label{sec:quantum_coing_flipping_and_cloning}
With our primary objective being the improvement of practicality in attacking quantum secure communication based protocols, this section describes the explicit protocols whose security we analyse through the lens of VQC. In particular, we focus on the primitive of \emph{quantum coin flipping}\cite{mayers_unconditionally_1999, aharonov_quantum_2000} and the use of states which have a fixed overlap. Protocols of this nature are a nice case study for our purposes since they provide a test-bed for cloning states with a fixed overlap, and in many cases explicit security analyses are missing. In particular in this work, to the best of our knowledge, we provide the first purely cloning based attack on the protocols we analyse.  

We note that cloning of phase covariant states we described above can be used to attack BB-84-like quantum key distribution protocols \cite{bennett_quantum_2014}. However, these attacks are not the focus of our work here. Nonetheless, in these cases, an optimal cloning strategy would correspond to the optimal `individual' attack on these QKD protocols\cite{scarani_quantum_2005}.
 
\subsection{Quantum Coin Flipping}\label{ssec:quantum_coin_flipping}
Coin flipping is a fundamental cryptographic primitive that allows two parties (without mutual trust) to remotely agree on a random bit without any of the parties being able to bias the coin in its favour. A `biased coin' has one outcome more likely than the other, for example with the following probabilities:
\begin{equation}\label{eqn:coin_flipping_epsilon_biased}
    \begin{split}
    \Pr(y = 0) = 1/2 + \epsilon\\
    \Pr(y = 1) = 1/2 -  \epsilon
    \end{split}
\end{equation}
where $y$ is a bit outputted by the coin. We can associate $y=0$ to heads ($\mathsf{H}$) and $y=1$ to tails ($\mathsf{T}$). The above coin is an $\epsilon$-biased coin with a bias towards $\mathsf{H}$. In contrast, a fair coin would correspond to $\epsilon=0$.

It has been shown that it is impossible\footnote{Meaning it is not possible to define a coin-flipping protocol such that \emph{neither} party can enforce any bias.} in an information theoretic manner, to achieve a secure coin flipping protocol with $\epsilon=0$ in both the classical and quantum setting\cite{blum_coin_1983, lo_why_1998, mayers_unconditionally_1999}. Furthermore, are two notions of coin-flipping studied in the literature: \emph{weak} coin-flipping, (where it is a-priori known that both parties prefer opposite outcomes), and \emph{strong} coin-flipping, (where neither party knows the desired bias of the other party). In the quantum setting, the lowest possible bias achieveable by any strong coin-flipping
protocol is limited by $\sim 0.207$\cite{alexei_kitaev_quantum_2003}. Although several protocols have been suggested for $\epsilon$-biased strong coin flipping\cite{mayers_unconditionally_1999, aharonov_quantum_2000, bennett_quantum_2014, berlin_fair_2009}, the states used in them share a common structure. First, we introduce said states used in different strong coin flipping protocols and then we show how approximate state-dependent cloning can implement practical attacks on them.

\subsubsection{Quantum states for strong coin flipping} \label{ssec:quantum_states_coin_flipping}
Multiple qubit coin-flipping protocols utilise the following set of states (illustrated in \figref{fig:coinflipstates}):
\begin{equation}\label{eq:coinflip-st}
    \ket{\phi_{x, a}} = 
    \begin{cases}
    \ket{\phi_{x,0}} = \cos\phi\ket{0} + (-1)^x\sin\phi\ket{1} \\
    \ket{\phi_{x,1}} =  \sin\phi\ket{0} + (-1)^{x \oplus 1}\cos\phi\ket{1}
  \end{cases} 
\end{equation}
where $x \in \{0,1\}$.

Such coin flipping protocols usually have a common structure. Alice will encode some random classical bits into some of the above states and Bob may do the same. They will then exchange classical or quantum information (or both) as part of the protocol. Attacks (attempts to bias the coin) by either party usually reduce to how much one party can learn about the classical bits of the other.

We explicitly treat two cases:
\begin{enumerate}
    \item The protocol of Mayers\cite{mayers_unconditionally_1999} in which the states, $\{\ket{\phi_{0,0}}, \ket{\phi_{1,0}}\}$ are used (which have a fixed overlap $s = \cos\left(2\phi\right)$). We denote this protocol $\mathcal{P}_1$.
    
    \item The protocol of Aharonov et. al.\cite{aharonov_quantum_2000}, which uses the full set, i.e.\@ $\{\ket{\phi_{x, a}}\}$. We denote this protocol $\mathcal{P}_2$.
\end{enumerate} 

These set of states are all conveniently related through a reparameterisation of the angle $\phi$\cite{brus_approximate_2008}, which makes them easier to deal with mathematically.

\begin{figure}[!ht]
    \centering
    \includegraphics[width=0.6\columnwidth, height = 0.6\columnwidth]{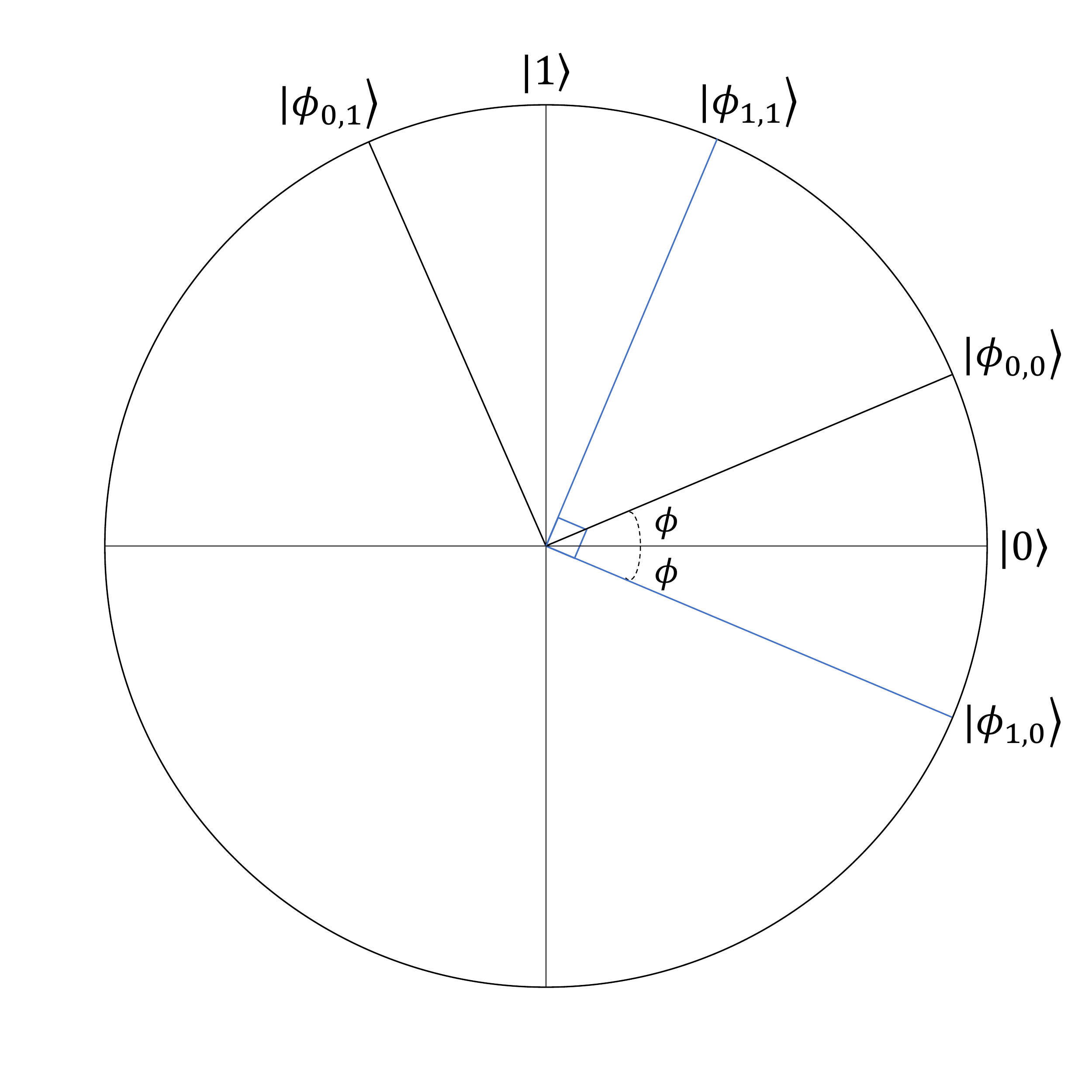}
    \caption{States used for quantum coin-flipping. The first bit represents the `basis', while the other represents one of the two orthogonal states.}
    \label{fig:coinflipstates}
\end{figure}
In all strong coin-flipping protocols, the security or fairness of the final shared bit lies on the impossibility of perfect discrimination of the underlying non-orthogonal quantum states. In general, the protocol can be analysed with either Alice or Bob being dishonest. Here we focus, for illustration, on a dishonest Bob who tries to bias the bit by cloning the non-orthogonal states sent by Alice. We analytically compute the maximum bias that can be achieved by this type of attack in the two protocols, and compare against our variational approaches.

\subsubsection{2-State Coin Flipping Protocol (\texorpdfstring{$\mathcal{P}_1$}{})} \label{ssec:mayers_protocol}
Here we give an overview of the protocol of Mayers et. al.~\cite{mayers_unconditionally_1999} and a possible cloning attack on it. This was incidentally one of the first protocols proposed for strong quantum coin-flipping. Here, Alice utilises the states\footnote{Since the value of the overlap is the only relevant quantity, the different parameterisation of these states to those of \eqref{eqn:state_dependent_cloning_states} does not make a difference for our purposes. However, we note that explicit cloning unitary would be different in both cases.} $\ket{\phi_0} := \ket{\phi_{0,0}}$ and $\ket{\phi_1} := \ket{\phi_{1, 0}}$ such that the angle between them is $\phi := \frac{\pi}{18} \implies s := \cos(\frac{\pi}{9})$.

First, we give a sketch of the protocol of Ref.~\cite{mayers_unconditionally_1999} which is generally described by $k$ rounds. However, for our purposes, it is sufficient to assume we have only a single round ($k=1$). Reasoning for this, and further details about the protocol can be found in \appref{app_ssec:mayers_protocol_and_attack}. Furthermore, the protocol is symmetric with respect to a cheating strategy by either party, but as mentioned we assume that Bob is dishonest. Firstly, Alice and Bob each pick random bits $a$ and $b$ respectively. The output bit will be the \computerfont{XOR} of these input bits i.e.\@
\begin{equation} \label{eqn:mayers_procol_final_bit}
y = a  \oplus b   
\end{equation}
Alice then chooses $n$ random bits, $\{c_i\}_{i=1}^n$, and sends the states $\{\ket{\phi,c_i} := \ket{\phi_{c_{i}}} \otimes\ket{\phi_{\overline{c_{i}}}}\}_{i=1}^n$\footnote{The notation $\overline{i}$ indicates the complement of bit $i$} to Bob.  Likewise, Bob sends the states $\{\ket{\phi, d_i} := \ket{\phi_{d_{i}}} \otimes\ket{\phi_{\overline{d_{i}}}}\}_{i=1}^n$ to Alice, where $d_i$ is a random bit for each $i$ of his choosing. Next, Alice announces the value $a\oplus c_i$ for each $i$. If $a\oplus c_i = 0$, Bob returns the second qubit of $\ket{\phi, c_i}$ back to Alice, and sends the first state otherwise. Once all copies have been transmitted, Alice and Bob announce $a$ and $b$, and (assuming she is the honest one) Alice performs a projective measurement on her remaining $n$ states (which depends on Bob's announced bit, $b$) and another POVM on the $n$ states returned by Bob (depending on her bit, $a$). These POVMs are constructed explicitly from the states $\ket{\phi_0}$ and $\ket{\phi_1}$ in order to give deterministic outcomes, and so can be used to detect Bob's cheating. Although some meaningful optimal attacks have been conjectured for this protocol, a full security proof has not been given previously~\cite{aharonov_quantum_2000}.

Now, we give a sketch of the cloning attack carried out by Bob (illustrated in \figref{fig:coin_flipping_attacks}(a)) here, and further details about the attack are presented in \appref{app_ssec:attack_on_mayers_protocol}. Prior to sending one of the states (one half of $\ket{\phi, c_i}$) back to Alice, Bob could employ the $1\rightarrow 2$ state dependent cloning strategy on the state he is required to send back.  He then sends one of the clones, and performs an optimal state discrimination between the qubit he didn't send, and the remaining clone. 

For example, if Alice announces $a\oplus c_i = 0$ for a given $i$, Bob will in turn send the second state of $\ket{\phi, c_i}$. Then, he must discriminate between the following two states (where $\rho^0_c$ is the remaining clone, i.e.\@ a reduced state from the output of the cloning machine when $\ket{\phi_0}$ is inputted):
\begin{align}\label{eqn:mayers_bob_pairs_discriminate}
   \rho_1 &= \ket{\phi_0}\bra{\phi_0}\otimes\ket{\phi_1}\bra{\phi_1}\\ \textrm{ and }  \qquad \rho_2 &= \ket{\phi_1}\bra{\phi_1}\otimes\rho^0_c
\end{align}

Now, we can use a state discrimination argument (via the Holevo-Helstrom\cite{holevo_statistical_1973, helstrom_quantum_1969} bound) to prove the following:

\begin{theorem}\label{thm:mayers_attack_bias_probability}[Ideal Cloning Attack Bias on $\mathcal{P}_1$]
    \\ Bob can achieve a bias of $\epsilon \approx 0.27$ using a state-dependent cloning attack on the protocol, $\mathcal{P}_1$, with a single copy of Alice's state.
\end{theorem}
The details of the proof can be found in  \appref{app_ssec:attack_on_mayers_protocol}.
We can also show that Bob's probability of guessing $a$ correctly approaches $1$ as the number of states that Alice sends, $n$, increases. In \secref{sec:results}, we compare this success probability to that achieved by the circuit learned by VQC, for these same states.

\begin{figure*}
    \centering
    \includegraphics[width=1.8\columnwidth]{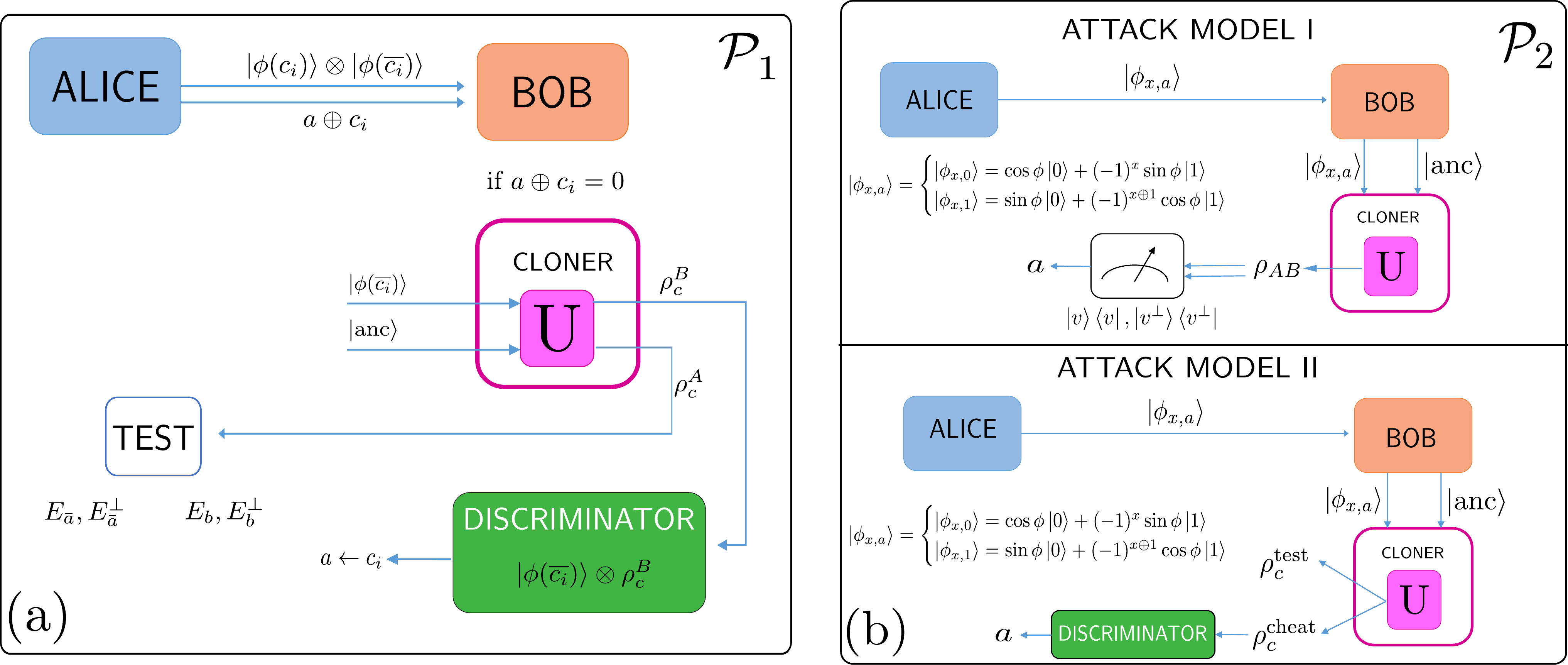}
    \caption{State dependent cloning attacks on the protocols, (a) $\mathcal{P}_1^1$, and (b) $\mathcal{P}_2$. (a) Dishonest Bob uses the cloning machine to produce a copy of the requested state by Alice, sends one copy to pass the test, and keeps the other copy to discriminate the pairs, and guess bit $a$. (b) In the first cloning based attack model (top), Bob measures both output of the cloner with a set of fixed projective measurements and guesses Alice's bit, $a$. In the second attack model (bottom), Bob keeps one of the clones for either testing Alice later or to send back the deposit qubit requested by Alice. He uses then the other local clone to discriminate and guess $a$.}
    \label{fig:coin_flipping_attacks}
\end{figure*}

\subsubsection{4-States Coin Flipping Protocol (\texorpdfstring{$\mathcal{P}_2$}{})} \label{ssec:aharonov_protocol_attack}

Another class of coin-flipping protocols are those which require all the four states in~\eqref{eq:coinflip-st}. One such protocol was proposed by Aharonov et. al.\cite{aharonov_quantum_2000}, where $\phi$ is set as $\frac{\pi}{8}$.

In protocols of this form, Alice encodes her bit in `basis information' of the family of states. More specifically, her random bit is encoded in the state $\ket{\phi_{x, a}}$. For instance, we can take $\{\ket{\phi_{0,0}}, \ket{\phi_{1, 0}}\}$ to encode the bit $a = 0$ and $\{\ket{\phi_{0, 1}}, \ket{\phi_{1,1}}\}$ to encode $a = 1$. The goal again is to produce a final `coin flip' $y = a\oplus b$, while ensuring that no party has biased the bit, $y$. A similar protocol has also been proposed using BB84 states\cite{bennett_quantum_2014} where $\ket{\phi_{0,0}} := \ket{0}, \ket{\phi_{0,1}} := \ket{1}, \ket{\phi_{1,0}} := \ket{+}$ and $\ket{\phi_{1,1}} := \ket{-}$. In this case, the states (also some protocol steps) are different but the angle between them is the same as with the states in $\mathcal{P}_2$. A fault-tolerant version of $\mathcal{P}_2$ has also been proposed in Ref.\cite{berlin_fair_2009}, which uses a generalised angle as in~\eqref{eq:coinflip-st}.

The protocol proceeds as follows. First Alice sends one of the states, $\ket{\phi_{x, a}}$ to Bob. Later, one of two things will happen. Either, Alice will send the bits $x$ and $a$ to Bob, who measures the qubit in the suitable basis to check if Alice was honest, \emph{or} Bob is asked to return the qubit $\ket{\phi_{x, a}}$ to Alice, who measures it and verifies if it is correct. Now, example cheating strategies for Alice involve incorrect preparation of $\ket{\phi_{x, a}}$ and giving Bob the wrong information about $(x, a)$, or for Bob in trying to determine the bits $x, a$ from $\ket{\phi_{x, a}}$ before Alice has revealed them classically. We again focus only on Bob's strategies here to use cloning arguments\footnote{The information theoretic achievable bias of $\epsilon = 0.42$ proven in Ref.~\cite{aharonov_quantum_2000} applies only to Alice's strategy since she has greater control of the protocol (she prepares the original state). In general, a cloning based attack strategy by Bob will be able to achieve a lower bias, as we show.}. As above, Bob randomly selects his own bit $b$ and sends it to Alice. He then builds a QCM to clone all 4 states in \eqref{eqn:aharonov_coinflip_states}.

We next sketch the two cloning attacks on Bob's side of $\mathcal{P}_2$. Again, as with the protocol, $\mathcal{P}_1$, Bob can cheat using as much information as he gains about $a$ and again, once Bob has performed the cloning, his strategy boils down to the problem of state discrimination. In both attacks, Bob will use a (variational) state-dependent cloning machine.

In the first attack model (which we denote I - see \figref{fig:coin_flipping_attacks}(b)) Bob measures \emph{all} the qubits outputted from the cloner to try and guess $(x, a)$. As such, it is the \emph{global} fidelity that will be the relevant quantity. This strategy would be useful in the first possible challenge in the protocol, where Bob is not required to send anything back to Alice. We discuss in \appref{app_sssec:aharonov_attack_I_computation} how the use of cloning in this type of attack can also reduce resources for Bob from a general POVM to projective measurements in the state discrimination, which may be of independent interest. The main attack here boils down to Bob measuring the global output state from his QCM using the projectors, $\ketbra{v}{v}, \ketbra{v^\perp}{v^\perp}$, and from this measurement, guessing $a$. These projectors are constructed explicitly relative to the input states using the Neumark theorem\cite{bae_quantum_2015}.

The second attack model (which we denote II - see \figref{fig:coin_flipping_attacks}(b)) is instead a \emph{local} attack and as such will depend on the optimal local fidelity. It may also be more relevant in the scenario where Bob is required to return a quantum state to Alice. We note that Bob could also apply a global attack in this scenario but we do not consider this possibility here in order to give two contrasting examples. In the below, we compute a bias assuming he does not return a state for Alice for simplicity and so the bias will be equivalent to his discrimination probability. The analysis could be tweaked to take a detection probability for Alice into account also. In this scenario, Bob again applies the QCM, but now he only uses one of the clones to perform state discrimination (given by the \textsf{DISCRIMINATOR} in \figref{fig:coin_flipping_attacks}(b)).

Now, we have the following, for a $4$ state \emph{global} attack on $\mathcal{P}_2$:
\begin{theorem}\label{thm:aharonov_attack_I_bias_probability_maintext}[Ideal Cloning Attack (I) Bias on $\mathcal{P}_2$]
Using a cloning attack on the protocol, $\mathcal{P}_2$, (in attack model I) Bob can achieve a bias:
\begin{equation}
    \epsilon^{\mathrm{I}}_{\mathcal{P}_2, \mathrm{ideal}} \approx 0.35
\end{equation}
\end{theorem}
Similarly, we have the bias which can be achieved with attack II:
\begin{theorem}\label{thm:aharonov_4state_attack_II_bias_probability_maintext}[Ideal Cloning Attack (II) Bias on $\mathcal{P}_2$]
Using a cloning attack on the protocol, $\mathcal{P}_2$, (in attack model II) Bob can achieve a bias:
\begin{equation}\label{eqn:attack_2_aharonov_success_probability_bound_maintext}
    \epsilon^{\mathrm{II}}_{\mathcal{P}_2, \mathrm{ideal}} = 0.25
\end{equation}
\end{theorem}
We prove these results in \appref{app_sec:coin_flipping_cloning_attacks}. We also describe an alternative attack based on a $2$ state cloning machine tailored to the states in \eqref{eqn:aharonov_coinflip_states} which achieves a higher bias in attack model II, but it does not connect as cleanly with the variational framework, so we defer it to \appref{app_sec:coin_flipping_cloning_attacks} also. Note the achievable bias with a cloning attack is lower in attack II, since Bob has information remaining in the leftover clone.


\section{Numerical Results} \label{sec:results}

Here we present numerical results demonstrating the methods described above.


\subsection{\texorpdfstring{$\Ansatz$}{} Choice} \label{ssec:ansatz_choice}

A key element in variational  algorithms is the choice of Ansatz which is used in the PQC. The primary Ansatz we choose is a variation of the variable-structure Ansatz approach in Refs.\cite{cincio_learning_2018, cerezo_variational_2020}. We use this to demonstrate the ability of VQC to learn to clone quantum states \emph{from scratch} in an end-to-end manner. This method is can be viewed as a form of application specific compilation, where the unitary to be implemented is one which clones quantum states, and the gateset to compile to is optional (typically the native gateset of the quantum hardware). 
In \appref{app_ssec:pc_cloning_fixed_ansatz} we also learn the angles in the optimal ideal circuit in \figref{fig:qubit_cloning_ideal_circ} using both the $\SWAP$ test and direct wavefunction simulation, and we also test a fixed-structure hardware efficient Ansatz.

The variable structure approach was given first in Ref.~\cite{cincio_learning_2018}, and variations on this idea have been given in many forms\cite{grimsley_adaptive_2019, ostaszewski_quantum_2019, chivilikhin_mog-vqe_2020, rattew_domain-agnostic_2020, li_quantum_2020} and most recently given the broad classification of \emph{quantum architecture search} (QAS)\cite{zhang_differentiable_2020} to draw parallels with neural architecture search\cite{yao_taking_2019, liu_darts_2019}, (NAS) in classical ML. The goal is to solve the following optimisation problem\cite{li_quantum_2020}:
\begin{align} \label{eqn:variable_structure_ansatz_optimisation_problem}
    (\paramtheta^{*}, \boldsymbol{g}^{*}) = \argmin_{\paramtheta, \boldsymbol{g} \in \mathcal{G}} \Cbs(\paramtheta, \boldsymbol{g})
\end{align}
where $\mathcal{G}$ is a gateset \emph{pool}, from which a particular sequence, $\boldsymbol{g}$ is chosen. As a summary, to solve this problem, we iterate over $\boldsymbol{g}$, swap out gates, and reoptimise the parameters, $\paramtheta$ until a minimum of the cost, $\Cbs(\paramtheta^{*}, \boldsymbol{g}^{*})$ is found. This is a combination of a discrete and continuous optimisation problem, where the discrete parameters are the indices of the gates in $\boldsymbol{g}$ (i.e.\@ the circuit structure), and the continuous parameters are $\paramtheta$. Each time the circuit structure is changed (a subset of gates are altered), the continuous parameters are reoptimised, as in Ref.\cite{cincio_learning_2018}. We provide more extensive details of the specific procedure we choose in \appref{app_b:structure_learning}. Variations of this approach have been proposed in Refs. \cite{du_quantum_2020, li_quantum_2020} which could be easily incorporated, and we leave such investigation to future work. 


\subsection{Phase-Covariant Cloning}

Here we demonstrate VQC with $1\rightarrow 2$ cloning phase covariant states, \eqref{eqn:x_y_plane_states}. We allow $3$ qubits ($2$ output clones plus $1$ ancilla) in the circuit. The fully connected (FC) gateset pool we choose for this problem is the following:
\begin{align} \label{eqn:phase_covariant_cloning_gateset}
    \mathcal{G}_{\textrm{PC}} = \left\{\right. \mathsf{R}^2_{z}(\theta), \mathsf{R}^3_{z}(\theta), \mathsf{R}^4_{z}(\theta),
    \mathsf{R}^2_{x}(\theta), \mathsf{R}^3_{x}(\theta), \mathsf{R}^4_{x}(\theta), \nonumber\\
    \mathsf{R}^2_{y}(\theta), \mathsf{R}^3_{y}(\theta), \mathsf{R}^4_{y}(\theta), 
    \CZ_{2, 3}, \CZ_{3, 4}, \CZ_{2, 4}\left.\right\}
\end{align}

where $\mathsf{R}_j^i$ indicates the $j^{th}$ Pauli rotation acting on the $i^{th}$ qubit and $\CZ$ is the controlled-$Z$ gate. In this case, we use the qubits indexed $2, 3$ and $4$ in an \computerfont{Aspen-8} sublattice.

We test two scenarios: the first is forcing the qubits to appear in registers $2$ and $3$, exactly as in \figref{fig:qubit_cloning_ideal_circ}, and the second is to allow the clones to appear instead in registers $1$ and $2$. The results of this can be seen in \figref{fig:learned_vs_ideal_circ_on_hw}, and clearly demonstrates the advantage of our flexible approach. The ideal circuit in \figref{fig:learned_vs_ideal_circ_on_hw}(b) suffers a degredation in performance when implemented on the physical hardware since it requires $6$ entangling gates as it is attempting to transfer the information across the circuit.

Furthermore, since the \computerfont{Aspen-8} chip does not have any $3$ qubit loops in its topology, it is necessary for the compiler to insert $\SWAP$ gates. The same is required for the circuit in \figref{fig:learned_vs_ideal_circ_on_hw}(a), however, by allowing the clones to appear in registers $1$ and $2$, VQC is able to find much more conservative circuits, having fewer entangling gates, and are directly implementable on a linear topology. A representative example is shown in \figref{fig:learned_vs_ideal_circ_on_hw}(c). This gives a significant improvement in the cloning fidelities, of about $15\%$ when the circuit is run on the QPU, as observed in \figref{fig:learned_vs_ideal_circ_on_hw}(d). We also note, in order to generate these results on the QPU, we use quantum state tomography\cite{dariano_quantum_2003} with the \computerfont{forest-benchmarking} library\cite{gulshen_forest_2019} to reconstruct the output density matrix, rather than using the $\SWAP$ test to compute the fidelities. We do this to mitigate the effect of quantum noise, which we elaborate on in \appref{app_ssec:pc_cloning_fixed_ansatz}.
\begin{figure*}
\centering
    \includegraphics[width=1.3\columnwidth, height=0.85\columnwidth]{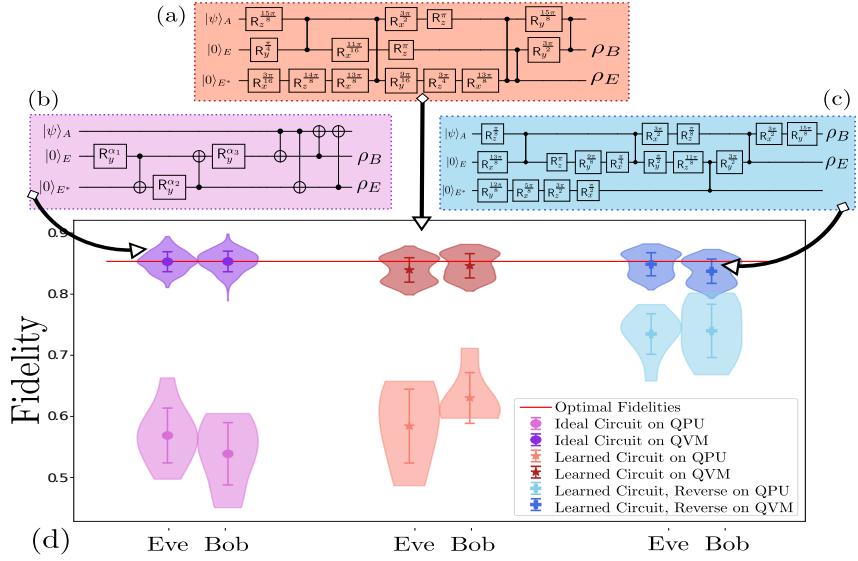}
    \caption{Variational Quantum Cloning implemented on phase-covariant states using $3$ qubits from the Rigetti \computerfont{Aspen-8} chip (QPU), plus simulated results (QVM). Violin plots in (d) show the cloning fidelities, for Bob and Eve, found using each of the circuits shown in (a)-(c) respectively. (b) is the ideal circuit from \figref{fig:qubit_cloning_ideal_circ} with clones appearing in registers $2$ and $3$. (a) shows the structure-learned circuit for the same scenario, using one less entangling gate. (c) demonstrates the effect of allowing clones to appear in registers $1$ and $2$. In the latter case, only $4$ (nearest-neighbour) entangling gates are used, demonstrating a significant boost in performance on the  QPU. The plot (d) also shows the maximal possible fidelity for this problem in red.}
    \label{fig:learned_vs_ideal_circ_on_hw}
\end{figure*}

Here we can also investigate the difference between the global and local fidelities achieved by the circuits VQC (i.e.\@ in \figref{fig:learned_vs_ideal_circ_on_hw}(a)) finds, versus the ideal one (shown in \figref{fig:learned_vs_ideal_circ_on_hw}(b)). We show in the \appref{app:local_vs_global_circuit_fidelities}, that the `ideal' circuit achieves both the optimal local and global fidelities for this problem:
\begin{multline} \label{eqn:ideal_circuit_local_global_fidelities}
\small
    \figref{fig:learned_vs_ideal_circ_on_hw}(b) \implies
    \begin{cases}
    F_{\mathrm{B}}^{(b)}  = F_{\mathrm{E}}^{(b)}  = F_{\mathsf{L}}^{\opt} = \frac{1}{2}\left(1 + \frac{1}{\sqrt{2}}\right) \approx 0.853\\
    F_{\mathsf{G}}^{(b)} = F_{\mathsf{G}}^{\mathrm{opt}} = \frac{1}{8}\left(1 + \sqrt{2}\right)^2 \approx 0.72
    \end{cases}
\end{multline}
In contrast, our learned circuit (\figref{fig:learned_vs_ideal_circ_on_hw}(a)) maximises the local fidelity, but in order to gain an advantage in circuit depth, compromises with respect to the global fidelity:
\begin{equation} \label{eqn:learned_circuit_local_global_fidelities}
    \figref{fig:learned_vs_ideal_circ_on_hw}(a) \implies\begin{cases}
    F_{\mathrm{B}}^{(a)}  \approx F_{\mathrm{E}}^{(a)} \approx F_{\mathsf{L}}^{\opt} = 0.85\\
    F_{\mathsf{G}}^{(a)}  \approx 0.638 < F_{\mathsf{G}}^{\opt}
    \end{cases}
\end{equation}
%

\subsection{State-Dependent Cloning}\label{ssec:results_state_dep_cloning}
Here we present the results of VQC when learning to clone the states used in the two coin flipping protocols above. Firstly, we focus on the states used in the original protocol, $\mathcal{P}_1$ for $1\rightarrow 2$ cloning, and then move to the 4 state protocol, $\mathcal{P}_2$. In the latter we also extend from $1\rightarrow 2$ cloning to $1\rightarrow 3$ and $2\rightarrow 4$ also. These extensions will allow us to probe certain features of VQC, in particular explicit symmetry in the cost functions. In all cases, we use the variable structure Ansatz, and once a suitable candidate has been found, the solution is manually optimised further. The learned circuits used to produce the figures in this section are given in \appref{app_sec:vqc_learned_circuits}.

\subsubsection{Cloning \texorpdfstring{$\mathcal{P}_1$}{} states}
As a reminder, the two states used in this protocol are:
\begin{align} \label{eqn:mayers_states_explicit_numerical}
    \ket{\phi_0} := \ket{\phi_{0, 0}} = \cos\left(\frac{\pi}{18}\right)\ket{0} + \sin\left(\frac{\pi}{18}\right)\ket{1}\\
    \ket{\phi_1} := \ket{\phi_{0, 1}} = \cos\left(\frac{\pi}{18}\right)\ket{0} - \sin\left(\frac{\pi}{18}\right)\ket{1}
\end{align}
For implementation of the learned circuit on the QPU, we use a $3$-qubit sublattice of the \computerfont{Aspen-8}. In an effort to increase hardware performance for this example, we further restrict the gateset allowed by VQC by explicitly enforce a linear entangling structure of $\mathsf{CZ}$ gates:
\begin{align} \label{eqn:mayers_1to2_state_dependent_cloning_gateset}
    \mathcal{G}_{\mathcal{P}_1^{1\rightarrow 2}} = \left\{\right. &\mathsf{R}^2_{z}(\theta), \mathsf{R}^3_{z}(\theta), \mathsf{R}^4_{z}(\theta),
    \mathsf{R}^2_{x}(\theta), \mathsf{R}^3_{x}(\theta), \mathsf{R}^4_{x}(\theta), \nonumber \\
    &\mathsf{R}^2_{y}(\theta), \mathsf{R}^3_{y}(\theta), \mathsf{R}^4_{y}(\theta),
    \CZ_{2, 3}, \CZ_{3, 4}\left.\right\}
\end{align}

The fidelities achieved by the VQC learned circuit can be seen in \figref{fig:mayers_1to2_cloning_fidelities_variational}. A deviation from the optimal fidelity is observed in the simulated case, partly due to tomographic errors in reconstructing the cloned states. We note that the corresponding circuit for \figref{fig:mayers_1to2_cloning_fidelities_variational} only actually used $2$ qubits (see \appref{app_sec:vqc_learned_circuits}). This is because while VQC was \emph{allowed} to use the ancilla, it \emph{chose not} in this case by applying only identity gates to it. This mimics the behaviour seen in the previous example of phase-covariant cloning. As such, we only use the two qubits shown in the inset (i) of the figure when running on the QPU to improve performance.
\begin{figure}
    \includegraphics[width=\columnwidth,height=0.27\textwidth]{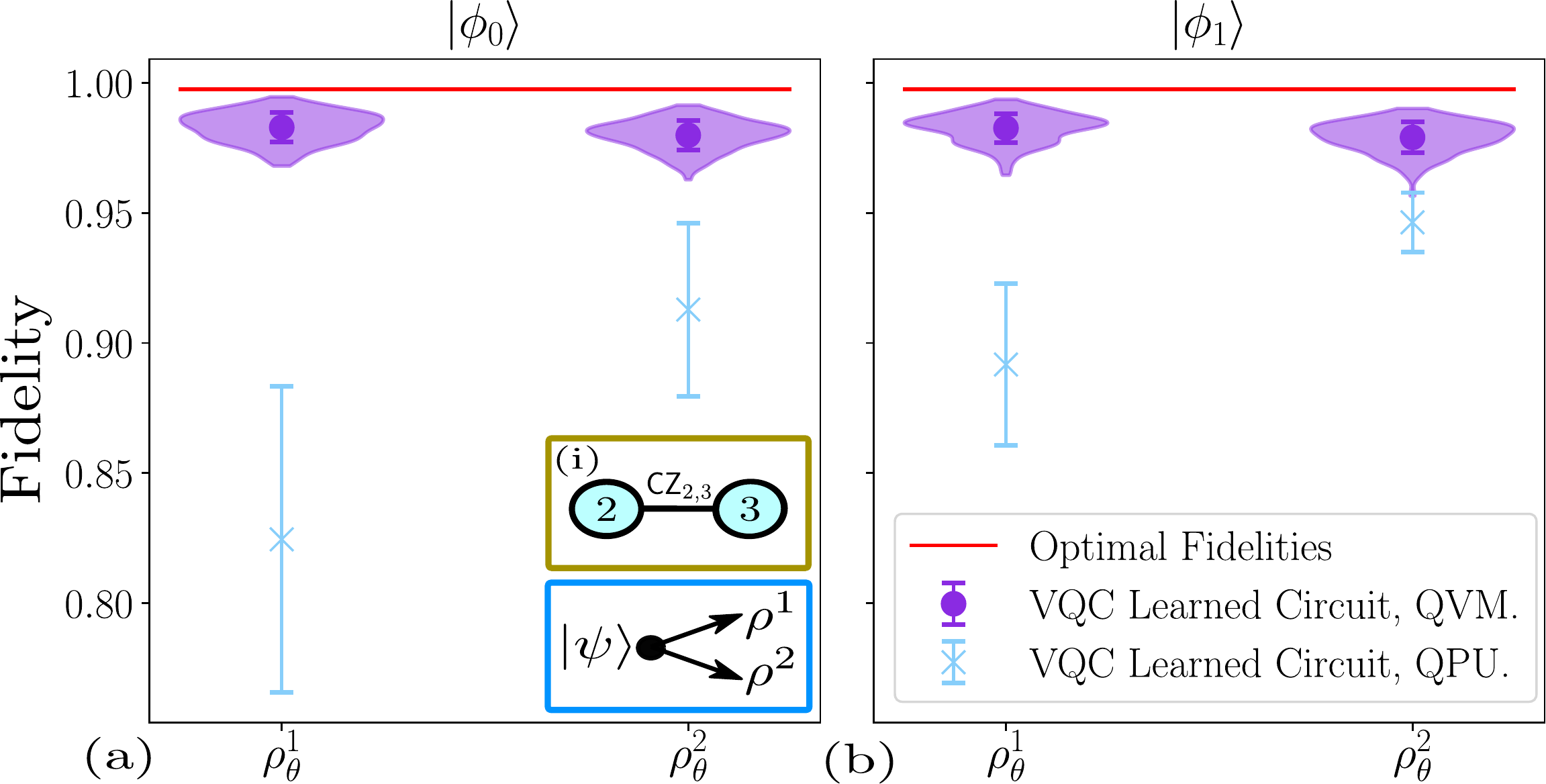}
    \caption{Fidelities of each output clone, $\rho^j_{\theta}$ achieved using VQC when ($1\rightarrow 2$) cloning the family of states used in the protocol, $\mathcal{P}_1$. In (a), $\ket{\phi_0}$ is inputted and in (b) $\ket{\phi_1}$ is the input to the QCM. Figure shows both simulated (QVM - purple circles) and on Rigetti hardware (QPU - blue crosses). For the QVM (QPU) results, 256 (5) samples of each state are used to generate statistics. Violin plots show complete distribution of outcomes and error bars show the means and standard deviations. Inset (i) shows the two qubits of the \computerfont{Aspen-8} chip which were used, with the allowed connectivity of a $\CZ$ between them. Note an ancilla was also allowed, but VQC chose not to use it in this example. The corresponding learned circuit is given in \appref{app_sec:vqc_learned_circuits}.}
    \label{fig:mayers_1to2_cloning_fidelities_variational}
\end{figure}

Now, returning to the attack on $\mathcal{P}_1$ above, we can compute the success probabilities using these fidelities. For illustration, let us return to the example in~\eqref{eqn:mayers_bob_pairs_discriminate}, where instead the cloned state is now produced from our VQC circuit, $\rho^0_c\rightarrow \rho^0_{\mathrm{VQC}}$.

\begin{theorem} \label{thm:vqc_bias_mayers_protocol}[VQC Attack Bias on $\mathcal{P}_1$]
    \\ Bob can achieve a bias of $\epsilon \approx 0.29$ using a state-dependent VQC attack on the protocol, $\mathcal{P}_1$, with a single copy of Alice's state.
\end{theorem}
\thmref{thm:vqc_bias_mayers_protocol} can be proven by computing the success probability as in \appref{app_ssec:attack_on_mayers_protocol}:
\begin{equation}\label{eqn:mayers_vqc_clone0}
\begin{split}
     P^{\mathrm{VQC}}_{\mathrm{succ}, \mathcal{P}_1} = \frac{1}{2}+\frac{1}{4}\Tr|\rho_1 - \ket{\phi_1}\bra{\phi_1}\otimes  \rho^0_{\mathrm{VQC}}| \approx 0.804
\end{split}
\end{equation}
The state $\rho_1$ is given in \eqref{eqn:mayers_bob_pairs_discriminate}. Here, we have a higher probability for Bob to correctly guess Alice's bit, $a$, but correspondingly the detection probability by Alice is higher than in the ideal case, due to a lower local fidelity of $F^{\mathrm{VQC}}_{\Lbs} = 0.985$. 

\subsubsection{Cloning \texorpdfstring{$\mathcal{P}_2 $}{} states.}
Next, we turn to the family of states used in the $4$ states protocol, which are:
\begin{equation}\label{eqn:aharonov_coinflip_states}
    \ket{\phi_{x, a}} = 
    \begin{cases}
    \ket{{\frac{\pi}{8}}_{x,0}} = \cos\left( \frac{\pi}{8} \right)\ket{0} + (-1)^x\sin\left( \frac{\pi}{8} \right)\ket{1} \\
    \ket{{\frac{\pi}{8}}_{x,1}} =  \sin\left( \frac{\pi}{8} \right)\ket{0} + (-1)^{x \oplus 1}\cos\left( \frac{\pi}{8} \right)\ket{1}
  \end{cases} 
\end{equation}
\\
\noindent\textbf{$1 \rightarrow 2$ Cloning.}\\
\noindent Firstly, we repeat the exercise from above with the same scenario, using the same gateset and subset of the \computerfont{Aspen-8} lattice $(\mathcal{G}_{\mathcal{P}_2^{1\rightarrow 2}} = \mathcal{G}_{\mathcal{P}_1^{1\rightarrow 2}})$. We use the local cost, \eqref{eqn:local_cost_full}, to train the model with, with a sequence length of $35$ gates. The results are seen in \figref{fig:aharonov_1to2_cloning_fidelities_variational} both on the QVM and the QPU. We note that the solution exhibits some small degree of asymmetry in the output states, due to the form of the local cost function. This asymmetry is especially pronounced as we scale the problem size and try to produce $N$ output clones, which we discuss in the next section.
\begin{figure}
    \centering
        \includegraphics[width=\columnwidth,height=0.38\textwidth]{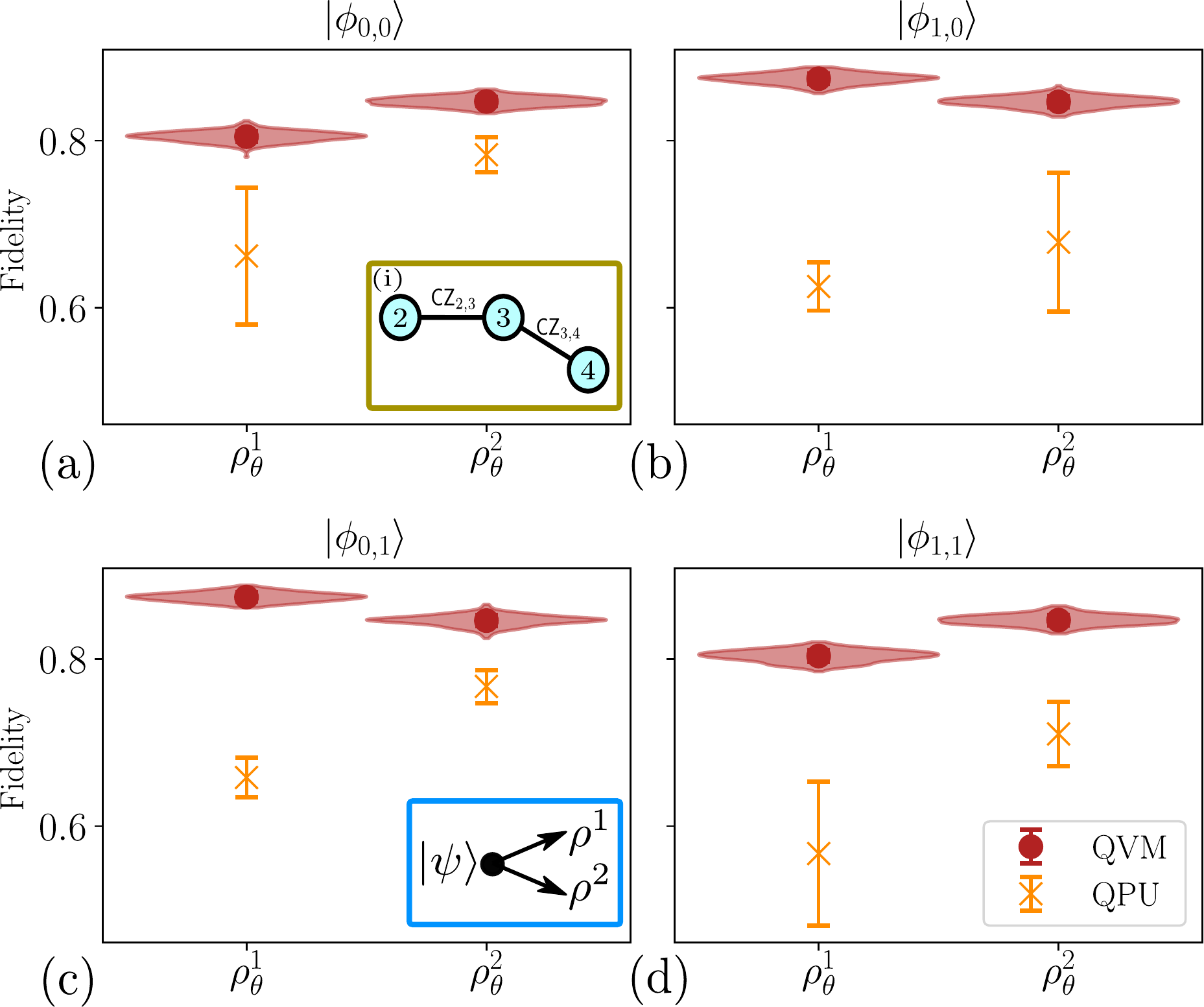}
    \caption{Fidelities achieved cloning the family of states, $\{\ket{\phi_{x, a}}\}$ used in the protocol, $\mathcal{P}_2$ using VQC both simulated (QVM - red circles) and on Rigetti hardware (QPU - orange crosses). We indicate the fidelities of the each clone received by Alice and Bob. For the QVM (QPU) results, 256 (3) samples of each state are used to generate statistics. Violin plots show complete distribution of outcomes and error bars show the means and standard deviations. Inset (i) shows the connectivity we allow in VQC for this example. The corresponding learned circuit is shown in \appref{app_sec:vqc_learned_circuits}.}
    \label{fig:aharonov_1to2_cloning_fidelities_variational}
\end{figure}

Now, we can relate the performance of the VQC cloner to the attacks discussed in \secref{ssec:aharonov_protocol_attack}. We do this by explicitly analysing the output states produced in the circuits of \figref{fig:aharonov_1to2_cloning_fidelities_variational} and following the derivation in \appref{app_sec:coin_flipping_cloning_attacks} for \thmref{thm:aharonov_attack_I_bias_probability_maintext_vqc} and \thmref{thm:aharonov_4state_attack_II_bias_probability_maintext_vqc}:
\begin{theorem}\label{thm:aharonov_attack_I_bias_probability_maintext_vqc}[VQC Cloning Attack (I) Bias on $\mathcal{P}_2$]
Using a cloning attack on the protocol, $\mathcal{P}_2$, (in attack model I) Bob can achieve a bias:
\begin{equation}
    \epsilon^{\mathrm{I}}_{\mathcal{P}_2, \mathrm{VQC}} \approx 0.345
\end{equation}
\end{theorem}
Similarly, we have the bias which can be achieved with attack II:
\begin{theorem}\label{thm:aharonov_4state_attack_II_bias_probability_maintext_vqc}[VQC Cloning Attack (II) Bias on $\mathcal{P}_2$]
Using a cloning attack on the protocol, $\mathcal{P}_2$, (in attack model II) Bob can achieve a bias:
\begin{equation}\label{eqn:attack_2_aharonov_success_probability_bound_maintext_real}
    \epsilon^{\mathrm{II}}_{\mathcal{P}_2, \mathrm{VQC}} = 0.241
\end{equation}
\end{theorem}
The discrepancy between these results and the ideal biases are primarily due to the small degree of asymmetry induced by the heuristics of VQC. However, we emphasise that these biases can now be achieved constructively, as a consequence of VQC.\\

\noindent \textbf{$1 \rightarrow 3$ and $2\rightarrow 4$ Cloning.}\\
\noindent Finally, we extend the above to the more general scenario of $M\rightarrow N$ cloning, taking $M=1, 2$ and $N=3, 4$. These examples are illustrative since they demonstrate strengths of the squared local cost function  (\eqref{eqn:squared_local_cost_mton}) over the local cost function (\eqref{eqn:local_cost_full}). In particular, we find the local cost function does not enforce symmetry strongly enough in the output clones, and using only the local cost function, suboptimal solutions are found. We particularly observed this in the example of $2\rightarrow 4$ cloning, where VQC tended to take a shortcut by allowing one of the input states to fly through the circuit (resulting in nearly $100\%$ fidelity for that clone), and then attempt to perform $1\rightarrow 3$ cloning with the remaining input state. By strongly enforcing symmetry in the output clones using the squared cost, this can be avoided as we demonstrate explicitly in \appref{app_sec:supplemental_numerical_results}.

We also test two connectivities in these examples, a fully connected (FC) and a nearest neighbour (NN) architecture as allowed by the following gatesets:
\begin{align} 
    \mathcal{G}^{\textnormal{NN}}_{\mathcal{P}_2^{1\rightarrow 3}} = 
    \left\{\right. &\mathsf{R}^i_{z}(\theta),  \mathsf{R}^i_{x}(\theta), \mathsf{R}^i_{y}(\theta), \nonumber \\
    &\CZ_{2, 3}, \CZ_{3, 4}, \CZ_{4, 5}\left.\right\} \quad
    \forall i \in \{2, 3, 4, 5\} \label{eqn:aharonov_1to3_state_dependent_cloning_gateset_NN}\\  \nonumber\\
    \mathcal{G}^{\textnormal{FC}}_{\mathcal{P}_2^{1\rightarrow 3}} =
    \left\{\right. &\mathsf{R}^i_{z}(\theta),  \mathsf{R}^i_{x}(\theta), \mathsf{R}^i_{y}(\theta), \CZ_{2, 3},  \CZ_{2, 4},  \CZ_{2, 5}, \nonumber\\
    &\CZ_{3, 4}, \CZ_{3, 5},  \CZ_{4, 5}\left.\right\} \quad
    \forall i \in \{2, 3, 4, 5\}  \label{eqn:aharonov_1to3_state_dependent_cloning_gateset_FC}
\end{align}
Note, that for $1\rightarrow 3$ ($2\rightarrow 4$) cloning, we actually use $4$ ($5$) qubits, with one being an ancilla.

\figref{fig:1to3_2to4_aharoanov_optimal_fidelities_plus_nn_vs_fc} shows the results for the optimal circuit found by VQC. In \figref{fig:1to3_2to4_aharoanov_optimal_fidelities_plus_nn_vs_fc}(a) we can achieve an average fidelity of $\approx 82\%$, using a NN connectivity, and for $2\rightarrow 4$ we can get an average fidelity of $\approx 84 \%$, using an FC connectivity, over all output clones.

\begin{figure}[ht]
    \begin{center}
        \includegraphics[width=\columnwidth, height=0.7\columnwidth]{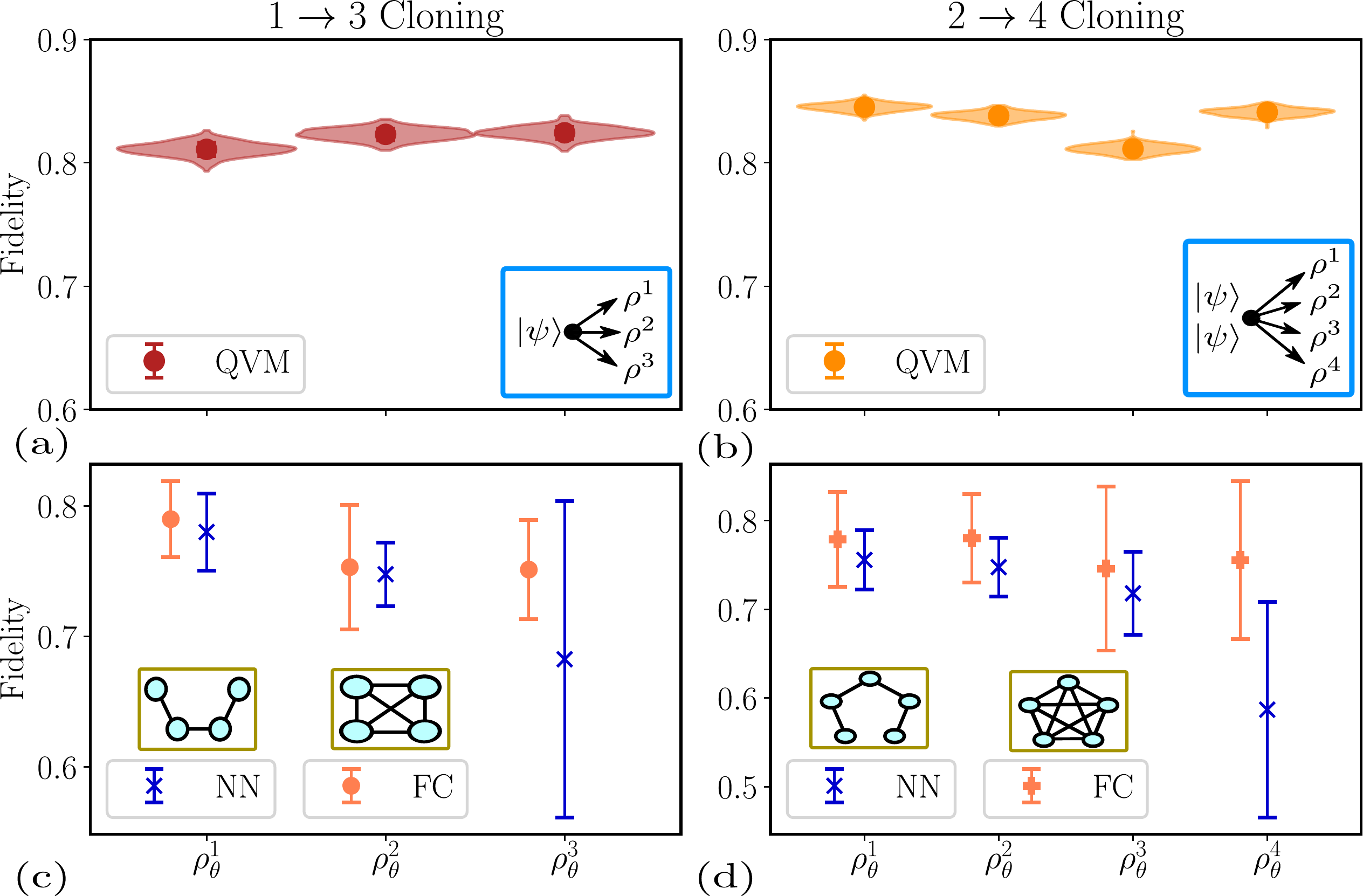}
    \caption{Clone fidelities for optimal circuits learned by VQC for (a) $1\rightarrow 3$ and (b) $2 \rightarrow 4$ cloning of \emph{all four} states in $\mathcal{P}_2$, \eqref{eqn:aharonov_coinflip_states}. Mean and standard deviations of $256$ samples are shown (violin plots show full distribution of fidelities), where the fidelities are computed using tomography only on the Rigetti QVM. (c-d) shows the mean and standard deviation of the optimal fidelities found by VQC over 15 independent runs ($15$ random initial structures, $\boldsymbol{g}$) for a nearest neighbour (NN - purple) versus (d) fully connected (FC - pink) entanglement connectivity allowed in the variable structure Ansatz for $1\rightarrow 3$ and $2 \rightarrow 4$ cloning of $\mathcal{P}_2$ states. Insets of (c-d) shown corresponding allowed $\CZ$ gates in each example. We use the following hyperparameters: 1) a sequence length of $l=35$ for $1\rightarrow 3$, and $l=40$ for $1\rightarrow 4$ with $50$ iterations over $\boldsymbol{g}$ in both cases,  2) the Adam optimiser with an initial learning rate of $\eta_{\text{init}} = 0.05$, 3) $50$ training samples. In all cases, we use the squared cost function, $\Cbs_{\sq}$ to train, and its gradients.
    }
    \label{fig:1to3_2to4_aharoanov_optimal_fidelities_plus_nn_vs_fc}
        \end{center}
\end{figure}
%


\section{Discussion}
\label{sec:discussion}

Quantum cloning is one of the most important ingredients not just as a tool in quantum cryptanalysis, but also with roots in foundational questions of quantum mechanics. However, given the amount of attention this field has received, a fundamental question remained elusive: how do we construct efficient, flexible, and noise tolerant circuits to actually perform approximate or probabilistic cloning? This question is especially pertinent in the current NISQ era where search for useful applications on small scale noisy quantum devices remains at the forefront. In this work, we attempt to answer this question by proposing variational quantum cloning (VQC), a cloning device that utilises the capability of short-depth quantum circuits and the power of classical computation to learn the ability to clone a state (or set of states) using the techniques of variational algorithms. This brings into view a whole new domain of performing realistic implementation of attacks on quantum cryptographic systems.   

We propose a family of cost functions for optimisation on a classical computer to suit the various needs of cloning scenarios. In particular, our proposed local and global cost functions are generic and display different desirable properties. We show that both our local and global cost functions provide an operational meaning and are faithful. Furthermore, we can prove the absence of barren plateaus for the local cost function for a hardware efficient $\Ansatz$.

Finally, to illustrate how VQC can be useful in quantum cryptography, we use quantum coin-flipping protocols as a specific example, deriving new attacks and demonstrate how VQC can be used to construct them. We concentrate specifically on this example, due to the connection with cloning fixed-overlap states, which has historically been a difficult subcase to tackle analytically.

We reinforce our theoretical proposal by providing numerical evidence of improved in cloning fidelities, when performed on actual quantum hardware, specifically the Rigetti hardware \computerfont{Aspen-8} QPU. For training, we use variable structure $\Ansatze$ with native Rigetti gates and demonstrate fidelity improvements up to 15$\%$, illustrating the effectiveness of our methodology.    

In conclusion, we remark that our work opens new frontiers of analysing quantum cryptographic schemes using quantum machine learning. In particular, this is applicable to secure communication schemes which are becoming increasingly relevant in quantum internet era.


\section*{Acknowledgments} \label{sec:acknowledgements}
We thank Atul Mantri for useful comments on the manuscript. This work was supported by the Engineering and Physical Sciences Research Council (grants EP/L01503X/1), EPSRC Centre for Doctoral Training in Pervasive Parallelism at the University of Edinburgh, School of Informatics,  Entrapping Machines, (grant FA9550-17-1-0055), and the H2020-FETOPEN Grant PHOQUSING (GA no.: 899544). We also thank Rigetti Computing for the use of their quantum compute resources, and views expressed in this paper are those of the authors and do not reflect the views or policies of Rigetti Computing.


\bibliographystyle{ieeetr}

\newpage

\onecolumngrid

\appendix


\section{Derivation of Analytic Gradients}\label{app_a:analytic_gradient}
Here we compute the analytic gradient, \eqref{eqn:analytic_squared_grad_mton} of the cost function, \eqref{eqn:squared_local_cost_mton}. The gradients of the local and global cost functions can also be computed analogously.
In this case, the cost is given by:
\begin{align} \label{eqn:squared_local_cost_mton_appendix_gradient}
    \Cbs_{\sq}^{M\rightarrow N}(\paramtheta) := \mathop{\mathbb{E}}_{\substack{\ket{\psi} \in \mathcal{S}}}\left[ \sum\limits_{i=1}^N (1-F^i_{\Lbs}(\paramtheta))^2\right. 
    \left.+ \sum\limits_{i<j}^N (F^i_{\Lbs}(\paramtheta)-F^j_{\Lbs}(\paramtheta))^2\right] 
\end{align}
where the expectation is taken over the uniform distribution. For example, in the phase covariant cloner of the states \eqref{eqn:x_y_plane_states}, the parameters, $\eta$ are sampled uniformly from the interval, $[0, 2\pi)$.

The local fidelity of the input state, $\ket{\psi}$ relative to the $j^{th}$ clone is:
\begin{align}\label{appendix_eqn:single_clone_fidelity}
    F^{j}_{\Lbs}(\paramtheta) = \bra{\psi}\rho_j(\boldsymbol{\theta})\ket{\psi}, \qquad
    \rho_j(\boldsymbol{\theta}) =  \Tr_{\bar{j}}\left(U(\boldsymbol{\theta})\rho_{\text{init}}U^{\dagger}(\boldsymbol{\theta}) \right)  ~ j \in [N]
\end{align}

Now, the derivative of \eqref{eqn:squared_local_cost_mton_appendix_gradient}, with respect to a single parameter, $\theta_l$, is given by:
\begin{equation*}
    \frac{\partial \Cbs_{\sq}(\boldsymbol{\theta})}{\partial \theta_l} =
    2\mathop{\mathbb{E}}_{\substack{\ket{\psi} \in \mathcal{S}}}\left[ \sum\limits_{i=1}^N (1-F^i_{\Lbs}(\paramtheta)) \left[-\frac{\partial  F^i_{\Lbs}(\boldsymbol{\theta})}{\partial \theta_l}\right]  + \sum\limits_{i<j}^N (F^i_{\Lbs}(\paramtheta)-F^j_{\Lbs}(\paramtheta)) \left[\frac{\partial  F^i_{\Lbs}(\boldsymbol{\theta})}{\partial \theta_l}  - \frac{\partial  F^j_{\Lbs}(\boldsymbol{\theta})}{\partial \theta_l} \right]\right] 
\end{equation*}
We can rewrite the expression for the fidelity of the $j^{th}$ clone as:
\begin{align}
    F^{j}_{\Lbs}(\boldsymbol{\theta}) = \bra{\psi}\rho_j(\boldsymbol{\theta})\ket{\psi} = \Tr\left[\ket{\psi}\bra{\psi}\rho_{j}\right] =  \Tr\left[\ket{\psi}\bra{\psi}\Tr_{\bar{j}}\left(U(\boldsymbol{\theta}) \rho_{\text{init}}U(\boldsymbol{\theta})^{\dagger} \right)\right] 
\end{align}

Using the linearity of the trace, the derivative of the fidelities with respect to the parameters, $\theta_l$, can be computed:
\begin{align}
    \frac{\partial   F^{j}_{\Lbs}(\boldsymbol{\theta})}{\partial \theta_l} =   \Tr\left[\ket{\psi}\bra{\psi}\Tr_{\bar{j}}\left(\frac{\partial  U(\boldsymbol{\theta})\rho_{\text{init}}U(\boldsymbol{\theta})^{\dagger}}{\partial \theta_l} \right)\right] 
\end{align}
Now, as mentioned in the main text, if we assume that each $U(\paramtheta) := U(\theta_d)U(\theta_{d - 1}) \dots U(\theta_{1})$ is composed of unitary gates of the form: $U(\theta_l) = \exp\left(-i\theta_l\Sigma_l\right)$, where $\Sigma_l^2= \mathds{1}$\footnote{From Ref.\cite{schuld_evaluating_2019}, we actually only need to assume that $\Sigma_l$ has at most two unique eigenvalues.} (for example, a tensor product of Pauli operators), then from Refs.~\cite{mitarai_quantum_2018, schuld_evaluating_2019}, we get:
\begin{align}
    \frac{\partial  U(\boldsymbol{\theta})\rho_{\text{init}}U(\boldsymbol{\theta})^{\dagger}}{\partial \theta_l} = U^{l+\frac{\pi}{2}}(\boldsymbol{\theta})\rho_{\text{init}}(U(\boldsymbol{\theta})^{l+\frac{\pi}{2}})^{\dagger} -  U^{l-\frac{\pi}{2}}(\boldsymbol{\theta})\rho_{\text{init}}(U(\boldsymbol{\theta})^{l-\frac{\pi}{2}})^{\dagger}
\end{align}
Where the notation, $U^{l\pm\frac{\pi}{2}}$, indicates the $l^{th}$ parameter has been shifted by $\pm \frac{\pi}{2}$, i.e.\@ $U^{l\pm\frac{\pi}{2}} := U(\theta_d)U(\theta_{d - 1}) \dots U(\theta_l\pm \pi / 2) \dots  U(\theta_{1})$. Now:
\begin{align*}
    \frac{\partial  F^{j}_{\Lbs}(\boldsymbol{\theta})}{\partial \theta_l} &=   \Tr\left[\ket{\psi}\bra{\psi}\Tr_{\bar{j}}\left(U^{l+\frac{\pi}{2}}(\boldsymbol{\theta})\rho_{\text{init}}(U(\boldsymbol{\theta})^{l+\frac{\pi}{2}})^{\dagger} \right)\right] 
    -   \Tr\left[\ket{\psi}\bra{\psi}\Tr_{\bar{j}}\left(U^{l-\frac{\pi}{2}}(\paramtheta)\rho_{\text{init}}(U(\boldsymbol{\theta})^{l-\frac{\pi}{2}})^{\dagger} \right)\right]\\
        \implies \frac{\partial   F^{ j}_{\Lbs}(\boldsymbol{\theta})}{\partial \theta_l} &=   \Tr\left[\ket{\psi}\bra{\psi}\rho_j^{l+\frac{\pi}{2}}(\boldsymbol{\theta})\right]
    -   \Tr\left[\ket{\psi}\bra{\psi}\rho_j^{l-\frac{\pi}{2}}(\boldsymbol{\theta})\right]
    = F^{(j, l+\frac{\pi}{2})}_{\Lbs}(\boldsymbol{\theta}) - F^{(j, l-\frac{\pi}{2})}_{\Lbs}(\boldsymbol{\theta})
\end{align*}
where we define $F^{(l\pm \frac{\pi}{2})}_j(\boldsymbol{\theta}) \coloneqq \bra{\psi}\rho_j^{l\pm\frac{\pi}{2}}(\boldsymbol{\theta})\ket{\psi}$ the fidelity of the $j^{th}$ clone, when prepared using a unitary whose $l^{th}$ parameter is shifted by $\pm \frac{\pi}{2}$, with respect to a target input state, $\ket{\psi}$. 

Plugging this into the above expression, we get:
\begin{equation}
    \frac{\partial \Cbs_{\sq}(\boldsymbol{\theta})}{\partial \theta_l} =
    \mathop{2\mathbb{E}}_{\substack{\ket{\psi} \in \mathcal{S}}}\left[
     \sum\limits_{i<j}^N (F^i_{\Lbs}-F^j_{\Lbs}) \left[F^{(i, l+\frac{\pi}{2})}_{\Lbs} - F^{(i, l-\frac{\pi}{2})}_{\Lbs} - F^{(j, l+\frac{\pi}{2})}_{\Lbs} + F^{(j, l-\frac{\pi}{2})}_{\Lbs} \right]\right.  - \left.\sum\limits_{i=1}^N (1-F^i_{\Lbs})\left[F^{(i, l+\frac{\pi}{2})}_{\Lbs} - F^{(i, l-\frac{\pi}{2})}_{\Lbs}\right] \right]
\end{equation}
which is the expression in the main text (suppressing the $\paramtheta$ dependence).


\section{Faithfulness: Relation between Local Cost Function and States} \label{app_sec:faithfulness}

In this section we explore the relationships between the local cost functions defined in \secref{sec:variational} which are used to train the parameterised circuit. In particular, we prove the various theorems presented in the main text.


\subsection{Symmetric Squared \& Local Cost Functions}

We look at the symmetric local cost functions. This is primarily of importance when the output clones are desired to have the same fidelities with respect to the input state.


\subsubsection{Squared Cost Function}

Here we examine the squared local cost function as defined in \eqref{eqn:squared_local_cost_mton} for $M \rightarrow N$ cloning and demonstrate that it  exhibits faithfullness arguments. The squared local cost function is given by:
\begin{equation} \label{eqn:squared_local_cost_mton_app}
\begin{split}
        \Cbs_{\sq}^{M\rightarrow N}(\paramtheta) &:= \mathop{\mathbb{E}}_{\substack{\ket{\psi} \in \mathcal{S}}}\left[ \sum\limits_{j=1}^N (1-F_i(\paramtheta))^2 + \sum\limits_{i<j}^N (F_i(\paramtheta)-F_j(\paramtheta))^2\right] = \frac{1}{\mathcal{N}}\int_{\mathcal{S}}\left[ \sum\limits_{j=1}^N (1-F_i(\paramtheta))^2 + \sum\limits_{i<j}^N (F_i(\paramtheta)-F_j(\paramtheta))^2\right]d \psi
\end{split}
\end{equation}
where the expectation of a fidelity $F_i$ over the states in distribution $\mathcal{S}$ is defined as $\mathbb{E}[F_i] = \frac{1}{\mathcal{N}}\int_{\mathcal{S}}F_i\cdot d\psi$, with the normalisation condition being $\mathcal{N} = \int_{\mathcal{S}}d\psi$. For qubit states, if the normalisation is over the entire Bloch sphere in $SU(2)$, then $\mathcal{N} = 4\pi$. For notation simplicity, we herein denote the $\Cbs_{\sq}^{M\rightarrow N}(\paramtheta)$ as $\Cbs_{\sq}(\paramtheta)$.

Before proving the results for weak faithfulness mentioned in the main text (\thmref{thm:squared_local_cost_squared_FS_weak_faithful} and \thmref{thm:squared_local_cost_squared_trace_weak_faithful}), we first discuss how the cost function is \emph{strongly} faithful.\\

\noindent \textbf{1. Strong faithfulness}:

\begin{theorem} \label{thm:squared_local_cost_squared_strong_faithful_local_appendix}
    The squared local cost function is locally strongly faithful, i.e.\@:
    \begin{equation}\label{eqn:squared_cost_function_locally_faithful_appendix}
        \Cbs_{\sq}(\paramtheta) = \Cbs_{\sq}^{\mathrm{opt}} \implies \rho_{\paramtheta}^{\psi, j} = \rho_{\opt}^{\psi, j} \qquad \forall \ket{\psi} \in \mathcal{S}, \forall j \in [N]
    \end{equation}
\end{theorem}

\begin{proof}
The cost function $\Cbs_{\textrm{sq}}(\paramtheta)$ achieves a minimum at the joint maximum of $\mathbb{E} [F_i(\paramtheta)]$ for all $i \in [N]$. In symmetric $M \rightarrow N$ cloning, the expectation value of all the $N$ output clone fidelities peak at $F_i = F_{\opt}$ for all input states $\ket{\psi}$ to be cloned. This corresponds to a unique optimal joint state $\rho_{\opt}^{\psi, j} = U_{\opt}\ket{\psi^{\otimes M}, 0^{\otimes N-M}}\bra{\psi^{\otimes M}, 0^{\otimes N-M}}U_{\opt}^{\dagger}$ for each $\ket{\psi} \in S$, where $U_{\opt}$ is the unitary producing the the optimal state. Since the joint optimal state and the corresponding fidelities are unique for all input states in the distribution, we conclude that the cost function achieves a minimum under precisely the unique condition i.e.\@ $\mathbb{E}[F_j(\paramtheta)] = F_{\opt}$ for all $j \in [N]$. This condition implies that, 
\begin{equation}\label{Eq:strong_faithful_condition}
    \rho^{\psi,j}_{\paramtheta} = \rho^{\psi,j}_{\opt}, \qquad \forall \ket{\psi} \in S, \forall j \in [N]
\end{equation}
where $\rho^{\psi,j}_{\paramtheta}$ is the reduced output clone state of the $j$-th cloner corresponding to the joint state $\rho^{\psi}_{\paramtheta}$ produced by the learnt unitary, and
$\rho^{\psi,j}_{\opt}$ is the optimal reduced states for the $i$-th cloner corresponding to the input state $\ket{\psi}$. We note that since $F_{\opt}$ is the same for all the reduced states $i \in [N]$, this implies that the optimal reduced states are all the same for a given $\ket\psi \in \mathcal{S}$.  Thus \eqref{Eq:strong_faithful_condition} provides the necessary guarantee that minimising the cost function of the parameterised circuit results in the corresponding circuit output state being equal to the optimal clone state for all the input states in the distribution.  
\end{proof}

\noindent \textbf{2. Weak faithfulness:}\\

\noindent Computing the exact fidelities of the output states requires an infinite number of copies. In reality, we run the iteration only a finite number of times and thus our cost function can only reach the optimal cost up to some precision. This is also relevant when running the circuit on devices in the NISQ era which would inherently introduce noise in the system. Thus, we can only hope to minimise the the cost function up to within some precision of the optimal cost. 

Formally, we state this as the following lemma:

\begin{lemma} \label{lemma:trace_bound_squared_cost}
Suppose the cost function is $\epsilon$-close to the optimal cost in symmetric cloning
%
\begin{equation}
    \Cbs_{\mathrm{sq}}(\paramtheta) - \Cbs^{\opt}_{\mathrm{sq}} \leqslant \epsilon
    \label{eq:cost_to_epsilon_squared_appendix}
\end{equation}
then we obtain that,
\begin{equation}\label{eq:tracecloseness_squared_local_appendix}
        \Tr\left[(\rho_{\opt}^{\psi, j} - \rho^{\psi, j}_{\paramtheta})\ket{\psi}\bra{\psi}\right] \leqslant \frac{\mathcal{N}\epsilon}{2(1 - F_{\opt})}, \qquad \forall \ket{\psi} \in \mathcal{S}, \forall j \in [N]
\end{equation}
\end{lemma}

\begin{proof}
In $M \rightarrow N$ symmetric cloning, the optimal cost function value is obtained when all the $N$ output clones are produced with the same optimal fidelity $F_{\opt}$.  Thus, using  Eq~\ref{eqn:squared_local_cost_mton_app}, the optimal cost function value is given by,
\begin{equation}
    \Cbs^{\opt}_{\sq} = N \cdot (1 - F_{\opt})^2
\end{equation}
Since the optimal cost function corresponds to all the output clones producing states with same fidelity, hence as the parameterised circuit starts to minimise the cost function $\Cbs_{\mathrm{sq}}(\paramtheta)$ to reach the optimal value, all the output clones start to produce states with approximately same fidelity. In the instance $\epsilon \rightarrow 0$, the second terms of Eq~\ref{eqn:squared_local_cost_mton_app} starts to vanish. Thus, the cost function explicitly enforces the symmetry property. Using the instance of $\epsilon \rightarrow 0$, we now consider the difference in the cost function value of the learned parameterised state and the optimal clone state:
\begin{equation}     \label{eqn:costtotrace_squared_appendix}
    \begin{split}
        \Cbs_{\textrm{sq}}(\paramtheta) - \Cbs^{\opt}_{\textrm{sq}} &=
        \frac{1}{\mathcal{N}}\int_{\mathcal{S}}\left[ \sum\limits_{i}^N (1-F_i(\paramtheta))^2 + \sum\limits_{i<j}^N  (F_i(\paramtheta)-F_j(\paramtheta))^2\right]d \psi  - N\cdot (1 - F_{\opt})^2\\
        &\approx \frac{1}{\mathcal{N}}\int_{\mathcal{S}}\left[ \sum\limits_{j}^N (1-F_j(\paramtheta))^2  - N\cdot (1 - F_{\opt})^2 \right]d \psi  \\
    &\approx \frac{1}{\mathcal{N}}\int_{\mathcal{S}}\left[ \sum\limits_{j}^N (F_{\opt} - F_j(\paramtheta))(2 - F_{\opt} - F_j(\paramtheta)) \right]d \psi  \\
     &\geqslant \frac{2(1 - F_{\opt})}{\mathcal{N}}\int_{\mathcal{S}}\left[ \sum\limits_{j}^N (F_{\opt} - F_j(\paramtheta)) \right]d \psi  \\
        &= \frac{2(1 - F_{\opt})}{\mathcal{N}}\left[ \sum\limits_{j}^N \int_{\mathcal{S}}\text{Tr}[(\rho_{\opt}^{\psi, j} - \rho^{\psi, j}_{\paramtheta})\ket{\psi}\bra{\psi}]d\psi \right]\\
    \end{split}
\end{equation}
Here $F_{\opt} = \frac{1}{\mathcal{N}}\int_{\mathcal{S}}\text{Tr}\left[\rho_{\opt}^{\psi, j}\ket{\psi}\bra{\psi}\right]d\psi$ is the same across all the input states $\ket{\psi} \in \mathcal{S}$, with the normalisation $\mathcal{N} = \int_{\mathcal{S}}d\psi$.  Utilising the inequality in \eqref{eq:cost_to_epsilon_squared_appendix} and \eqref{eqn:costtotrace_squared_appendix}, we obtain,
\begin{equation}
\begin{split}
       &  \sum\limits_{j}^N \int_{\mathcal{S}}\text{Tr}\left[(\rho_{\opt}^{\psi, j} - \rho^{\psi, j}_{\paramtheta})\ket{\psi}\bra{\psi}\right]d\psi  \leqslant \frac{\mathcal{N}\epsilon}{2(1 - F_{\opt})} \\
       & \implies \text{Tr}\left[(\rho_{\opt}^{\psi, j} - \rho^{\psi, j}_{\paramtheta})\ket{\psi}\bra{\psi}\right] \leqslant \frac{\mathcal{N}\epsilon}{2(1 - F_{\opt})}, \hspace{3mm} \forall \ket{\psi} \in \mathcal{S}, \forall j \in [N]
\label{eq:tracecloseness_squared_local}
\end{split}
\end{equation}
\end{proof}
The above inequality allows us to quantify the closeness of the state produced by the parameterised unitary and the unique optimal clone state for any $\ket{\psi} \in \mathcal{S}$. We quantify this closeness of the states in the two popular distance measures in quantum information, the Fubini-Study distance\cite{nielsen_quantum_2010} and the Trace distance between the two quantum states. \\

Now, we are in a position to prove \thmref{thm:squared_local_cost_squared_FS_weak_faithful} and \thmref{thm:squared_local_cost_squared_trace_weak_faithful}, which we repeat here for clarity:

\begin{theorem}\label{thm:squared_local_cost_squared_FS_weak_faithful_appendix}
The squared cost function as defined \eqref{eqn:squared_local_cost_mton}, is $\epsilon$-weakly faithful with respect to the Fubini-distanc measure $\textrm{D}_{\mathrm{FS}}$.
In other words, if the squared cost function, \eqref{eqn:squared_local_cost_mton}, is $\epsilon$-close to its minimum, i.e.\@:
\begin{equation}\label{eq:squared_cost_to_epsilon_appendix}
    \Cbs_{\mathrm{sq}}(\paramtheta) - \Cbs^{\mathrm{opt}}_{\mathrm{sq}} \leqslant \epsilon
\end{equation}
where $\Cbs^{\opt}_{\mathrm{sq}} := \underset{\paramtheta}{\textrm{min}}\sum\limits_{i}^N (1-F_i(\paramtheta))^2 + \sum\limits_{i<j}^N  (F_i(\paramtheta)-F_j(\paramtheta))^2 = N(1-F_{\mathrm{opt}})^2$ is the optimal theoretical cost using fidelities produced by the ideal \emph{symmetric} cloning machine, then the following fact holds:
\begin{equation}     \label{eqn:fubini_study_bound_squared_appendix}
    \textrm{D}_{\mathrm{FS}}(\rho^{\psi, j}_{\paramtheta}, \rho_{\opt}^{\psi, j}) \leqslant \frac{\mathcal{N}}{2(1 - F_{\mathrm{opt}})\sin(F_{\mathrm{opt}})}\cdot \epsilon := f_1(\epsilon),   \qquad \forall \ket{\psi} \in \mathcal{S}, \forall j \in [N]
\end{equation}
\end{theorem}

\begin{proof}
To prove \thmref{thm:squared_local_cost_squared_FS_weak_faithful_appendix}, we revisit and rewrite the Fubini-Study distance as \cite{nielsen_quantum_2010}:
\begin{equation} \label{eqn:fubini_study_defn_appendix}
    \textrm{D}_{\textrm{FS}}(\rho,\sigma) = \arccos{\sqrt{F(\rho, \sigma)}} = \text{arccos}\hspace{1mm} \bra{\phi}\tau\rangle
\end{equation}
where $\ket{\phi}$ and $\ket{\tau}$ are the purifications of $\rho$ and $\sigma$ respectively which maximise the overlap. We note that $\textrm{D}_{\textrm{FS}}(\rho,\sigma)$ lies between $[0,\pi/2]$, with the value $\pi/2$ corresponding to the unique solution of $\rho = \sigma$. Since this distance is a metric,  it follows the triangle's inequality, i.e., for any three states $\rho, \sigma$ and $\delta$,
\begin{equation}
     \textrm{D}_{\textrm{FS}}(\rho,\sigma) \leqslant  \textrm{D}_{\textrm{FS}}(\rho,\delta) +  \textrm{D}_{\textrm{FS}}(\sigma,\delta)
     \label{eq:triangle-inequality}
\end{equation}

Rewriting the result of \lemref{lemma:trace_bound_squared_cost} in terms of fidelity for each $\ket{\psi} \in \mathcal{S}$ and correspondingly in terms of Fubini-Study distance using \eqref{eqn:fubini_study_defn_appendix} is, 
\begin{equation}
    F(\rho_{\opt}^{\psi, j}, \ket{\psi}) - F(\rho^{\psi,j}_{\paramtheta}, \ket{\psi}) \leqslant \epsilon' \implies  \cos^2(\textrm{D}_{\textrm{FS}}(\rho_{\opt}^{\psi, j}, \ket{\psi})) -  \cos^2(\textrm{D}_{\textrm{FS}}(\rho^{\psi,j}_{\paramtheta}, \ket{\psi})) \leqslant \epsilon' 
    \label{eq:fidelityFS_squared_local}
\end{equation}
where $\epsilon' = \mathcal{N}\epsilon/2(1 - F_{\opt})$. 
Let us denote $D_{\pm}^{\psi} = \textrm{D}_{\textrm{FS}}(\rho_{\opt}^{\psi, j}, \ket{\psi}) \pm \textrm{D}_{\textrm{FS}}(\rho^{\psi,j}_{\paramtheta}, \ket{\psi})$
This inequality in \eqref{eq:fidelityFS_squared_local} can be further rewritten as,
\begin{equation}
\begin{split}
    \cos(\textrm{D}_{\textrm{FS}}(\rho_{\opt}^{\psi, j}, \ket{\psi})) -  \cos(\textrm{D}_{\textrm{FS}}(\rho^{\psi,j}_{\paramtheta}, \ket{\psi})) &\leqslant \frac{\epsilon'}{\cos(\textrm{D}_{\textrm{FS}}(\rho_{\opt}^{\psi, j}, \ket{\psi})) +  \cos(\textrm{D}_{\textrm{FS}}(\rho^{\psi,j}_{\paramtheta}, \ket{\psi}))} \\
    \cos(\textrm{D}_{\textrm{FS}}(\rho_{\opt}^{\psi, j}, \ket{\psi})) -  \cos(\textrm{D}_{\textrm{FS}}(\rho^{\psi,j}_{\paramtheta}, \ket{\psi})) &\lessapprox \frac{\epsilon'}{2\cos(\textrm{D}_{\textrm{FS}}(\rho_{\opt}^{\psi, j}, \ket{\psi}))} \\
    2\sin\left(\frac{D^{\psi}_{+}}{2}\right)\sin\left(\frac{\textrm{D}^{\psi}_{-}}{2}\right) &\leqslant \frac{\epsilon'}{2\cos(\textrm{D}_{\textrm{FS}}(\rho_{\opt}^{\psi, j}, \ket{\psi}))} \\
    \implies \textrm{D}^{\psi}_{-} &\leqslant \frac{\epsilon'}{\sin(\textrm{D}_{\textrm{FS}}(\rho_{\opt}^{\psi, j}, \ket{\psi}))}  = \frac{\mathcal{N}\epsilon}{2(1 - F_{\opt})\sin(F_{\opt})}
\end{split}
\label{eq:cosinerelation}
\end{equation}
where we have used the approximations that in the limit $\epsilon \rightarrow 0$, $\textrm{D}_{\textrm{FS}}(\rho_{\opt}^{\psi, j}, \ket{\psi}) \approx \textrm{D}_{\textrm{FS}}(\rho^{\psi,j}_{\paramtheta}, \ket{\psi})$ and the trigonometric identities $\cos\left( x - y \right) = 2\sin \left(\frac{x+y}{2}\right)\sin \left(\frac{x-y}{2}\right)$, and $\sin 2x = 2\sin x \cos x$. 

Further, using the Fubini-Study metric triangle's inequality on the states $\{\rho_{\opt}^{\psi, j}, \rho^{\psi,j}_{\paramtheta}, \ket{\psi}\}$
results in,
\begin{equation}
 \textrm{D}_{\textrm{FS}}(\rho^{\psi,j}_{\paramtheta}, \ket{\psi}) \leqslant \textrm{D}_{\textrm{FS}}(\rho_{\opt}^{\psi, j}, \ket{\psi}) + \textrm{D}_{\textrm{FS}}(\rho^{\psi,j}_{\paramtheta}, \rho_{\opt}^{\psi, j}) 
\end{equation}
Combining the above inequality and \eqref{eq:cosinerelation} results in,
\begin{equation} \label{eqn:squared_cost_fubini_study_bound_appendix}
    \textrm{D}_{\textrm{FS}}(\rho^{\psi,j}_{\paramtheta}, \rho_{\opt}^{\psi, j}) \leqslant \frac{\mathcal{N}}{2(1 - F_{\opt})\sin(F_{\opt})}\cdot \epsilon, \hspace{3mm} \forall \ket{\psi} \in \mathcal{S}
\end{equation}
This bounds the closeness of the trained output state and the optimal output state as a function of $\epsilon$. \\
\end{proof}

As our second result, we prove \thmref{thm:squared_local_cost_squared_trace_weak_faithful} which provides a closeness argument in terms of the trace distance of the output state of the parameterised circuit and the optimal clone state $\ket\psi \in \mathcal{S}$. Contrary to the Fubini-study distance, this distance only holds true when the input states are qubits. The trace distance is a desirable bound to have since it is a strong notion of distance between quantum states, generalising the total variation distance between classical probability distributions\cite{nielsen_quantum_2010}.

\begin{theorem} \label{thm:squared_local_cost_squared_trace_weak_faithful_appendix}
The squared cost function, \eqref{eqn:squared_local_cost_mton}, is $\epsilon$-weakly faithful with respect to the trace distance $\textrm{D}_{\Tr}$.
\begin{equation}     \label{eqn:trace_distance_bound_squared_appendix}
        \textrm{D}_{\Tr}(\rho_{\opt}^{\psi, j},  \rho^{\psi, j}_{\paramtheta})  \leqslant g_1(\epsilon), \qquad \forall j \in [N]
\end{equation}
where:
\begin{equation} \label{eqn:trace_distance_squared_cost_bound_function_appendix}
    g_1(\epsilon) \approx \frac{1}{2}\sqrt{4F_{\mathrm{opt}}(1 - F_{\mathrm{opt}}) + \epsilon\frac{\mathcal{N}(1 - 2F_{\mathrm{opt}})}{2(1 - F_{\mathrm{opt}})}}
\end{equation}
\end{theorem}

\begin{proof}

Firstly, we note that $F_{\opt} = \bra{\psi}\rho_{\opt}^{\psi, j}\ket{\psi}$ is the same value for all input states $\ket\psi \in \mathcal{S}$. We apply the change of basis from $\ket{\psi} \rightarrow \ket{0}$ by applying the unitary $V\ket{\psi} = \ket{0}$. Then the effective change on the state $\rho_{\opt}^{\psi, j}$ to have a fidelity $F_{\opt}$ with the state $\ket{0}$ is, $\rho_{\opt}^{\psi, j} \rightarrow V\rho_{\opt}^{\psi, j}V^{\dagger}$. 

We can write the state $V\rho_{\opt}^{\psi, j}V^{\dagger}$ as,
\begin{equation}
    V\rho_{\opt}^{\psi, j}V^{\dagger} = \begin{pmatrix}
F_{\opt} & a^*\\ 
a & 1- F_{\opt}
\end{pmatrix}
\end{equation}
where we use the usual properties of a density matrix and $a \in \mathbb{C}$. The upper bound condition in \eqref{eq:tracecloseness} states that $\bra{\psi}\rho_{\opt}^{\psi, j} - \rho^{\psi, j}_{\paramtheta}\ket{\psi} = \bra{0}V(\rho_{\opt}^{\psi, j} - \rho^{\psi, j}_{\paramtheta})V^{\dagger}\ket{0}
= \mathcal{N}\epsilon/2(1 - F_{\opt}) = \epsilon'$ then becomes,
\begin{equation}
    V\rho^{\psi, j}_{\paramtheta}V^{\dagger} = \begin{pmatrix}
F_{\opt} + \epsilon' & b^*\\ 
b & 1- (F_{\opt} + \epsilon')
\end{pmatrix}
\end{equation}
for some $b \in \mathbb{C}$. The condition that $V\rho_{\opt}^{\psi, j}V^{\dagger}, V\rho^{\psi, j}_{\paramtheta}V^{\dagger} \geqslant 0$ i.e. they are positive, implies that,
\begin{equation}
 |a|^2 \leqslant F_{\opt}(1 - F_{\opt}) \equiv r_{F_{\textrm{oct}}}^2, \hspace{3mm} |b|^2 \leqslant (F_{\opt} + \epsilon')(1 - (F_{\opt} + \epsilon')) \equiv r_{F_{\opt} + \epsilon'}
\end{equation}
The trace distance between two general qubit states is related to the positive eigenvalue of the difference of the two qubit states. Consider the eigenvalues of $V\rho_{\opt}^{\psi, j}V^{\dagger} -  V\rho^{\psi, j}_{\paramtheta}V^{\dagger}$. The two eigenvalues of this matrix is $\lambda_{\pm} = \pm \sqrt{\epsilon'^2 + |a - b|^2}$. From this, the trace distance between the two states is,
\begin{equation}
    \mathrm{D}_{\mathrm{Tr}}(V\rho_{\opt}^{\psi, j}V^{\dagger}, V\rho^{\psi, j}_{\paramtheta}V^{\dagger}) = \frac{1}{2}\left|\left|V\rho_{\opt}^{\psi, j}V^{\dagger} - V\rho^{\psi, j}_{\paramtheta}V^{\dagger}\right|\right| = \frac{1}{2}|\lambda_{+}| = \frac{1}{2}\sqrt{\epsilon'^2 + |a - b|^2}
\end{equation}
We note that the trace distance is unitary invariant. Thus, 
\begin{equation}
    \begin{split}
        \textrm{D}_{\textrm{Tr}}(\rho_{\opt}^{\psi, j},  \rho^{\psi, j}_{\paramtheta})  &= \textrm{D}_{\textrm{Tr}}(V\rho_{\opt}^{\psi, j}V^{\dagger}, V\rho^{\psi, j}_{\paramtheta}V^{\dagger}) \\
        &= \frac{1}{2}\sqrt{\epsilon'^2 + |a - b|^2} \\
        &\leqslant \frac{1}{2}\sqrt{\epsilon'^2 + (r_{F_{\textrm{oct}}} + r_{F_{\opt} + \epsilon'})^2} \\
        &\approx \frac{1}{2}\sqrt{4F_{\opt}(1 - F_{\opt}) + \epsilon'(1 - 2F_{\opt})} \\
        &=  \frac{1}{2}\sqrt{4F_{\opt}(1 - F_{\opt}) + \epsilon\frac{\mathcal{N}(1 - 2F_{\opt})}{2(1 - F_{\opt})}}
    \end{split}
    \label{eq:trace-distance-closeness}
\end{equation}
where we have used the inequality $|a - b|^2 \leqslant ||a| + |b||^2$ for all $a,b \in \mathbb{C}$.

\end{proof}

The two distance measures convey the the closeness in the value of parameterised cost function w.r.t. the optimal cost value, i.e. they guarantee that the output state of the learnt circuit is also close to the optimal unique cloned output state with appropriate minimum distance values.


\subsubsection{Local Cost Function} \label{sssec:local_cost_function_appendix}
Next, returning the local cost function defined for $M \rightarrow N$ cloning to include the distribution $\mathcal{S}$ over the input states is,
\begin{align} \label{eqn:local_cost_full_appendix}  
    \Cbs_{\Lbs}(\paramtheta) &:= \mathbb{E}\left[1 - \frac{1}{N}\left(\sum\limits_{j=1}^{N} F_j({\paramtheta})\right)\right] = 1 - \frac{1}{N\mathcal{N}}\int_{\mathcal{S}} \sum\limits_{j=1}^{N} F_j({\paramtheta}) d\psi
\end{align}
where $\mathcal{N} = \int_{\mathcal{S}}d\psi$ is the normalisation condition. 
As above, we can show this cost function also exhibits strong faithfulness:\\

\noindent \textbf{1. Strong faithfulness:}

\begin{theorem} \label{thm:squared_local_cost_local_FS_strong_faithful_appendix}
    The squared local cost function is locally strongly faithful:
    \begin{equation}\label{eqn:strongly_local_faithful_local_cost_defn_appendix}
        \Cbs_{\Lbs}(\paramtheta) = \Cbs_{\Lbs}^{\mathrm{opt}} \implies \rho_{\paramtheta}^{\psi, j} = \rho_{\opt}^{\psi, j} \qquad \forall \ket{\psi} \in \mathcal{S}, \forall j \in [N]
    \end{equation}
\end{theorem}

\begin{proof}
Similar the faithfulness arguments of the squared cost function, one can immediately see that the cost function $\Cbs_{\Lbs}(\paramtheta)$ achieves a unique minimum at the joint maximum of $\mathbb{E}[F_j(\paramtheta)]$ for all $j \in [N]$. Thus, the minimum of $\Cbs_{\Lbs}(\paramtheta)$ corresponds to the unique optimal joint state with its unique local reduced states $\rho_{\opt}^{\psi, j}$ for each $j \in [N]$ for each input state $\ket{\psi} \in \mathcal{S}$. 
Thus the cost function achieves a minimum under precisely the unique condition i.e. the output state is equal to the optimal clone state. \\ 
\end{proof}

\noindent \textbf{2. Weak faithfulness:} \\

Now, we can also prove analogous versions of weak faithfulness for the local cost function, \eqref{eqn:local_cost_full}. Many of the steps in the proof follow similarly to the squared cost derivations above, so we omit them for brevity where possible. As above, we first have the following lemma:
\begin{lemma} \label{lemma:trace_bound_local_cost}
Suppose the cost function is $\epsilon$-close to the optimal cost in symmetric cloning
%
\begin{equation}
    \Cbs_{\Lbs}(\paramtheta) - \Cbs^{\opt}_{\Lbs} \leqslant \epsilon
    \label{eqn:costtoepsilon_local_appendix}
\end{equation}
where we assume $ \lim_{\epsilon \rightarrow 0} |\mathbb{E}[F_i(\paramtheta)] - \mathbb{E}[F_j(\paramtheta)]| \rightarrow 0, \forall i, j$, and therefore $\Cbs_{\opt} := 1-F_{\opt}$. Then,
\begin{equation}\label{eq:tracecloseness_local_cost_appendix}
        \Tr[(\rho_{\opt}^{\psi, j} - \rho^{\psi, j}_{\paramtheta})\ket{\psi}\bra{\psi}] \leqslant \mathcal{N}\epsilon, \qquad \forall \ket{\psi} \in \mathcal{S}, \forall j \in [N]
\end{equation}
\end{lemma}
The proof of \lemref{lemma:trace_bound_local_cost} follows identically to \lemref{lemma:trace_bound_squared_cost}, but with the exception that we can write $\Cbs_{\Lbs}(\paramtheta) - \Cbs^{\opt}_{\Lbs} = \mathbb{E}\left(F_{\opt} - F(\paramtheta)\right)$ in the symmetric case, assuming $F_i(\paramtheta) \approx F_j(\paramtheta)$,  $\forall i \neq j \in [N]$.

Now, to prove \thmref{thm:local_cost_FS_weak_faithful_appendix}:

\begin{theorem} \label{thm:local_cost_FS_weak_faithful_appendix}
The local cost function, \eqref{eqn:local_cost_full}, is $\epsilon$-weakly faithful with respect to $\mathrm{D}_{\mathrm{FS}}$
\begin{equation}  \label{eqn:cost_to_epsilon_local_appendix}
    \Cbs_{\Lbs}(\paramtheta) - \Cbs_{\Lbs}^{\opt} \leqslant \epsilon
\end{equation}
where $\Cbs_{\Lbs}^{\opt} := 1 - F_{\mathrm{opt}}$ then the following fact holds:
\begin{equation}     \label{eqn:fubini_study_bound_local_appendix}
    \textrm{D}_{\mathrm{FS}}(\rho^{\psi,j}_{\paramtheta}, \rho^{\psi, j}_{\opt}) \leqslant \frac{\mathcal{N}\epsilon}{\sin(F_{\opt})} =: f_2(\epsilon), \hspace{3mm} \forall \ket{\psi} \in \mathcal{S}, \forall j \in [N]
\end{equation}
\end{theorem}

\begin{proof}
 We rewrite the \eqref{eq:tracecloseness_local_cost_appendix} in terms of the Fubini-Study distance,
\begin{equation}    \label{eqn:fidelityFS_local_appendix}
    F(\rho_{\opt}^{\psi, j}, \ket{\psi}) - F(\rho^{\psi,j}_{\paramtheta}, \ket{\psi}) \leqslant \mathcal{N}\epsilon \implies  \cos^2(\textrm{D}_{\textrm{FS}}(\rho_{\opt}^{\psi, j}, \ket{\psi})) -  \cos^2(\textrm{D}_{\textrm{FS}}(\rho^{\psi,j}_{\paramtheta}, \ket{\psi})) \leqslant \mathcal{N}\epsilon 
\end{equation}
Following the derivation in the squared cost function section, we obtain the Fubini-Study closeness as,
\begin{equation}
    \textrm{D}_{\textrm{FS}}(\rho^{\psi,j}_{\paramtheta}, \rho_{\opt}^{\psi, j}) \leqslant \frac{\mathcal{N}\epsilon}{\sin(F_{\opt})}, \qquad \forall \ket{\psi} \in \mathcal{S}, \forall j \in [N]
    \label{eq:FSbound-standard-local}
\end{equation}
\end{proof}

Finally, we have \thmref{thm:local_cost_trace_weak_faithful_appendix} relating to the trace distance. The proof follows identically to \thmref{thm:squared_local_cost_squared_trace_weak_faithful_appendix}
so we just state the result:

\begin{theorem} \label{thm:local_cost_trace_weak_faithful_appendix}
The local cost function, \eqref{eqn:local_cost_full}, is $\epsilon$-weakly faithful with respect to $\textrm{D}_{\Tr}$ on qubits.
\begin{equation}     \label{eqn:trace_distance_bound_local_appendix}
    \textrm{D}_{\Tr}(\rho_{\opt}^{\psi, j},  \rho^{\psi, j}_{\paramtheta})  \leqslant \frac{1}{2}\sqrt{4F_{\opt}(1 - F_{\opt}) + \mathcal{N}\epsilon(1 - 2F_{\opt})} =: g_2(\epsilon), \qquad \forall j \in [N]
\end{equation}

\end{theorem}
%


\subsection{Asymmetric Cloning} \label{app_ssec:asymmetric_cloning}
 
As discussed in the main text, for certain applications, we require asymmetric cloning in the output states i.e.\@ in the $1 \rightarrow 2$ cloning,  the optimal reduced states of Bob and Eve do not necessarily have the same fidelities with respect to the input states. We note that the cost functions proposed for symmetric cloning does not work for the asymmetric case because the symmetric cost functions are a monotonic function of Bob's and Eve's output state fidelities with respect to the input states, thus they always converge to the optimal fidelity values which are same for for Bob and Eve. This section provides a construction for asymmetric cost function with a desired output fidelity in one of the clones.  

\subsubsection{Optimal asymmetric fidelities} \label{app_sssec:optimal_asymmetric_fidelities}

From Ref.\cite{scarani_quantum_2005}, any universal $1 \rightarrow 2$ cloning circuit producing outputs clones for Bob and Eve must satisfy the no-cloning inequality:
\begin{equation} \label{eqn:no_cloning_inequality_appendix}
 \sqrt{(1 - F^{p, B}_{\Lbs})(1 - F^{q, E}_{\Lbs})} \geqslant \frac{1}{2} - (1 - F^{p, B}_{\Lbs}) - (1 -F^{q, E}_{\Lbs})   
\end{equation}
where the output clones of Bob and Eve are denote by $F^{p, B}_{\Lbs}$ and $F^{q, E}_{\Lbs}$ for the desired parameterisations $p$ and $q$.

It can be easily verified that the fidelities that saturate the above inequality are,
\begin{equation}     \label{eqn:asymmetric_clones_appendix}
    F^{p, B}_{\Lbs} = 1 - \frac{p^2}{2}, \hspace{3mm} F^{q, E}_{\Lbs} = 1 - \frac{q^2}{2}, \qquad p, q \in [0,1] , 
\end{equation}
with  $p, q$ satisfy $p^2 + q^2 + pq = 1$. This implies that Eve is free to choose a desired fidelity for either clone, by varying the parameter, $p$. For example, suppose Eve wish to send a clone to Bob with a particular fidelity $F^p_B = 1 - p^2/2$, then from \eqref{eqn:no_cloning_inequality_appendix} her clone would have a corresponding fidelity: 
\begin{equation} \label{eqn:asym_fids_eve_fixed_q}
    F^p_E = 1 - \frac{1}{4}(2 - p^2 - p\sqrt{4 - 3p^2})
\end{equation}
We will use this to build the asymmetric cost function, \eqref{eqn:asymmetric_cost_function_maintext}. in the next section.

\subsection{Asymmetric Cost Functions} \label{app_ssec:asymmetric_cost_function_and_guarantees}

Here we define the asymmetric cost function for $1\rightarrow 2$ cloning, but remark that it can be generalised to arbitrary $M \rightarrow N$ cloning. This cost function for a particular input state family, $\mathcal{S}$ is then:
\begin{equation} \label{eqn:asymmetric_cost_full_appendix}
      \Cbs_{\Lbs, \textrm{asym}}(\paramtheta) := \mathbb{E}\left[F_{\Lbs}^{p, E} - F_{\Lbs}^{E}(\paramtheta)\right]^2 + \mathbb{E}\left[F_{\Lbs}^{p, E} -F_{\Lbs}^{E}(\paramtheta)\right]^2 
      = \frac{1}{\mathcal{N}}\int_{\mathcal{S}}\left((F_{\Lbs}^{p, B} - F_{\Lbs}^{B}(\paramtheta))^2 + (F_{\Lbs}^{p, E} - F_{\Lbs}^{E}(\paramtheta))^2\right)d\psi
\end{equation}
with $F_{\Lbs}^{p, j}, j\in \{B, E\}$ defined according to the conditions in \eqref{eqn:asym_fids_eve_fixed_q}. We note Eve could also choose a specific fidelity for her clone, parameterised by $q$, $F_{\Lbs}^{q, E} = 1 - q^2/2$, which would in turn determine $F_{\Lbs}^{q, B}$ as above.

\subsubsection{Asymmetric Faithfulness} \label{ssec:asymmetric_faithfulness}

\noindent \textbf{1. Strong faithfulness:} 

\begin{theorem} \label{eqn:asymmetric_local_cost_squared_FS_strong_faithful}
    The asymmetric $1 \rightarrow 2$ local cost function is  strongly faithful:
    \begin{equation}\label{eqn:asymmetric_strong_local_faithfulness}
         \Cbs_{\Lbs, \mathrm{asym}}(\paramtheta) = \Cbs^{ \opt}_{\Lbs, \mathrm{asym}}(\paramtheta)  \implies \rho_{\paramtheta}^{\psi, i} = \rho_{\opt}^{\psi, i} \qquad \forall \ket{\psi} \in \mathcal{S}, \forall i \in \{B, E\}
    \end{equation}
\end{theorem}

\begin{proof}
The cost function $\Cbs_{\Lbs, \mathrm{asym}}(\paramtheta)$ achieves the minimum value of zero, uniquely when $F_{\Lbs}^{B}(\paramtheta) = F_{\Lbs}^{p, B}$ and $F_{\Lbs}^{E}(\paramtheta) = F_{\Lbs}^{p, E}$ for all input states $\ket{\psi} \in \mathcal{S}$. This corresponds to the unique reduced states $\rho^{\psi,B}_{\mathrm{opt}}$ and $\rho^{\psi, E}_{\mathrm{opt}}$ for Bob and Eve. Thus the cost function, achieves a unique minimum of zero precisely when the output reduced state for Bob and Eve is equal to the optimal clones for all inputs in $\mathcal{S}$. \\
\end{proof}

\noindent \textbf{2. Weak faithfulness:}

Returning again to $\epsilon$-weak faithfulness, we get similar results as in the symmetric case above:

\begin{theorem} \label{thm:asymmetric_cost_FS_weak_faithful_appendix}
The asymmetric cost function, \eqref{eqn:asymmetric_cost_full_appendix}, is $\epsilon$-weakly faithful with respect to $\mathrm{D}_{\mathrm{FS}}$
\begin{equation}  \label{eqn:asymmetric_cost_function_epsilon_guarantee}
 \Cbs_{\Lbs, \mathrm{asym}}(\paramtheta) - \Cbs^{\opt}_{\Lbs, \mathrm{asym}} \leqslant \epsilon   
\end{equation}
where $\Cbs^{\opt}_{\Lbs, \mathrm{asym}} = 0$. then the following fact holds for Bob and Eve's reduced states::
\begin{equation} \label{eqn:asymmetric_cost_function_FS_bound_appendix}
     \textrm{D}_{\textrm{FS}}(\rho^{\psi,B}_{\paramtheta}, \rho^{\psi,B}_{\opt}) \leqslant \frac{\sqrt{\mathcal{N}\epsilon}}{\sin(1 - p^2/2)}, \hspace{3mm} \textrm{D}_{\textrm{FS}}(\rho^{\psi,E}_{\paramtheta}, \rho^{\psi,E}_{\opt}) \leqslant \frac{\sqrt{\mathcal{N} \epsilon}}{\sin(1 - q^2/2)}
\end{equation}
Furthermore, we also have the following trace distance bounds:
\begin{equation}    \label{eqn:asymmmetric_trace_distance_bound_Bob_Eve}
   \textrm{D}_{\Tr}(\rho^{\psi, B},  \rho^{\psi, B}_{\paramtheta})  
        \leqslant \frac{1}{2}\sqrt{p^2(2 - p^2) - \sqrt{\mathcal{N}\epsilon}(1 - p^2)}, 
        \hspace{10mm} \textrm{D}_{\Tr}(\rho^{\psi, E},  \rho^{\psi, E}_{\paramtheta})  
        \leqslant \frac{1}{2}\sqrt{q^2(2 - q^2) - \sqrt{\mathcal{N}\epsilon}(1 - q^2)} 
\end{equation}

\end{theorem}

\begin{proof}
Firstly, we derive a similar result to \lemref{lemma:trace_bound_local_cost} and \lemref{lemma:trace_bound_squared_cost}. By expanding the term $|\Cbs_{\Lbs, \textrm{asym}}(\paramtheta) - \Cbs_{\Lbs, \opt}|$ in terms of the corresponding output states, we obtain,
\begin{equation} \label{eqn:asymmetric_optimal_cost_expanded}
\begin{split}
     |\Cbs_{\Lbs,\textrm{asym}}(\paramtheta) - \Cbs_{\Lbs, \opt}| &= 
     \left|\frac{1}{\mathcal{N}}\int_{\mathcal{S}}\left((F_{\Lbs}^{p, B}- F_{\Lbs}^{B}(\paramtheta))^2 + (F_{\Lbs}^{p, E} - F_{\Lbs}^{ E}(\paramtheta))^2\right)d\psi \right| \\
     &= \frac{1}{\mathcal{N}}\int_{\mathcal{S}}\left[\left(\Tr[(\rho^{\psi, B} - \rho^{\psi, B}_{\paramtheta})\ket{\psi}\bra{\psi}] \right)^2 + \left(\Tr[(\rho^{\psi, E} - \rho^{\psi, E}_{\paramtheta})\ket{\psi}\bra{\psi}] \right)^2 \right]d\psi
\end{split}
\end{equation}
Using the inequalities \eqref{eqn:asymmetric_cost_function_epsilon_guarantee} and \eqref{eqn:asymmetric_optimal_cost_expanded}, we get,
\begin{equation}\label{eqn:aymmmetric_cost_local_trace_bound}
     \frac{1}{\mathcal{N}}\int_{\mathcal{S}}\left(\text{Tr}[(\rho^{\psi, j}_{\opt} - \rho^{\psi, j}_{\paramtheta})\ket{\psi}\bra{\psi}] \right)^2 d\psi \leqslant \epsilon 
     \implies \text{Tr}[(\rho^{\psi, j}_{\opt} - \rho^{j}_{\paramtheta})\ket{\psi}\bra{\psi}]\leqslant \sqrt{\mathcal{N}\epsilon}
\end{equation}
where $j \in \{B,E\}$. Thus the above inequality holds true for the output clone states corresponding to both Bob and Eve.

Next, to derive \eqref{eqn:asymmetric_cost_function_FS_bound_appendix} we rewrite the \eqref{eqn:aymmmetric_cost_local_trace_bound} in terms of the Fubini-Study metric,
\begin{equation}
    F(\rho^{\psi, j}_{\opt}, \ket{\psi}) - F(\rho^{\psi, j}_{\paramtheta}, \ket{\psi}) \leqslant \sqrt{\mathcal{N}\epsilon} \implies  \cos^2(\textrm{D}_{\textrm{FS}}(\rho^{\psi, j}_{\opt}, \ket{\psi})) -  \cos^2(\textrm{D}_{\textrm{FS}}(\rho^{\psi, j}_{\paramtheta}, \ket{\psi})) \leqslant \sqrt{\mathcal{N}\epsilon} 
    \label{eq:fidelityFS_asymmetric_appendix}
\end{equation}
Following the derivation in the squared symmetric  cost function section, we obtain the Fubini-Study closeness as,
\begin{equation}
    \textrm{D}_{\textrm{FS}}(\rho^{\psi,j}_{\paramtheta}, \rho^{\psi, j}_{\opt}) \leqslant \frac{\sqrt{\mathcal{N}\epsilon}}{\sin(F_{\Lbs}^{r, j})}, \hspace{3mm} \forall \ket{\psi} \in \mathcal{S}
    \label{eq:FSbound-asym-standard-local}
\end{equation}
where $F_{\Lbs}^{r, j}$ is the optimal cloning fidelity corresponding to $j \in \{B,E\}$ with $r \in \{p, q\}$. Finally, plugging in the optimal asymmetric fidelities, $F_{\Lbs}^{p, B} = 1-p^2/2$, and similarly for $F_{\Lbs}^{q, E}$ we arrive at:
\begin{equation}
     \textrm{D}_{\textrm{FS}}(\rho^{\psi,B}_{\paramtheta}, \rho^{\psi,B}_{\opt}) \leqslant \frac{\sqrt{\mathcal{N} \epsilon}}{\sin(1 - p^2/2)}, \qquad \textrm{D}_{\textrm{FS}}(\rho^{\psi,E}_{\paramtheta}, \rho^{\psi,E}_{\opt}) \leqslant \frac{\sqrt{\mathcal{N}\epsilon}}{\sin(1 - q^2/2)}
\end{equation}
Finally, to prove \eqref{eqn:asymmmetric_trace_distance_bound_Bob_Eve}, we follow the trace distance derivation bounds as in previous sections and obtain:
\begin{equation}     \label{eq:asym-trace-distance-closeness}
        \textrm{D}_{\Tr}(\rho^{\psi, j},  \rho^{\psi,j}_{\paramtheta})
        \leqslant \frac{1}{2}\sqrt{4F_{\Lbs}^{r, j}(1 - F_{\Lbs}^{r, j}) + \sqrt{\mathcal{N}\epsilon}(1 - 2F_{\Lbs}^{r, j} ) }
\end{equation}
Again, plugging in the optimal fidelities for Bob and Eve completes the proof.

\end{proof}


\section{Faithfulness of Global Optimisation and its Relationship with Local Optimisation} \label{app:local_vs_global_circuit_fidelities}

This section explores the relationship between local and global cost function optimisation for different cloners (universal, phase-covariant, etc.). In particular, we address the question of whether optimising a cloner with a local or a global cost function results in the generation of same unique output clone state, and thus the same optimal local/global fidelity. This relationship particularly aides in choosing the correct cost function for a particular application. We note that this relationship only manifests in \emph{symmetric} cloning, since there is no possibility to enforce asymmetry in the global cost function. As seen in \eqref{eqn:asymmetric_cost_full_appendix}, the only way to enforce asymmetry is by constructing a cost function which optimises with respect to the local asymmetric optimal fidelites.  


\subsection{Global Faithfulness} \label{app_ssec:global_faithfulness}

As a first step, we show in the next theorems that the global cost function exhibits the notions of strong and weak faithfulness by proving statements that the closeness of global cost function with the optimal cost value implies the closeness of global output clone state with the optimal global clone state. 
\begin{theorem} \label{thm:global_cost_strong_faithful_appendix}
    The standard global cost function is globally strongly faithful, i.e.\@:
    \begin{equation}\label{eqn:strongly_faithful_cost_defn_appendix}
        \Cbs_{\Gbs}(\paramtheta) = \Cbs^{\mathrm{opt}}_{\Gbs} \implies \rho_{\paramtheta}^\psi = \rho_{\opt}^{\psi} \qquad \forall \ket{\psi} \in \mathcal{S}
    \end{equation}
\end{theorem}
\begin{proof}
 The global cost function $\Cbs_{\Gbs}(\paramtheta)$ achieves the minimum value $\Cbs^{\opt}_{\Gbs}$ at a unique point corresponding to $\mathbb{E}[F_{\Gbs}(\paramtheta) = F_{\Gbs}^{\opt}$, where $F_{\Gbs}^{\opt}$ corresponds to the fidelity term for $\Cbs^{\mathrm{opt}}_{\Gbs}$. This corresponds to the unique global clone state $\rho^{\psi}_{\opt}$. Thus the cost function, achieves a unique minimum under precisely the unique condition i.e. the output global state is equal to the optimal clone state for all inputs in the distribution.\end{proof}
Now we provide statements of weak faithfulness, something that is much more relevant in the practical implementation of the cloning scheme using global optimisation.
\begin{lemma} \label{lemma:global-state-closeness}
Suppose the cost function is $\epsilon$-close to the optimal cost in symmetric cloning
%
\begin{equation}
    C_{\Gbs}(\paramtheta) - C^{\opt}_{\Gbs} \leq \epsilon
    \label{eq:costtoepsilon_global_trace}
\end{equation}
where $\Cbs^{\opt}_{\Gbs} := 1 - F_{\Gbs}^{\opt}$. Then,
\begin{equation}\label{eq:tracecloseness}
        \Tr\left[(\rho^{\psi}_{\opt} - \rho^{\psi}_{\paramtheta})\ket{\psi}^{\otimes 2}\bra{\psi}^{\otimes 2}\right] \leqslant \mathcal{N}\epsilon, \qquad \forall \ket{\psi} \in \mathcal{S}
\end{equation}
\end{lemma}
\begin{proof}
The proof of lemma~\ref{lemma:global-state-closeness} follows identically to lemma~\ref{lemma:trace_bound_local_cost} but with the exception that $C_{\Gbs}(\paramtheta) - C^{\opt}_{\Gbs} = \mathbb{E}[ F_{\Gbs}^{\opt} - F_{\Gbs}(\paramtheta)]$.
\end{proof}
For the above lemma we are now in the position to prove the theorem~\ref{thm:trace_FS_bound_global}.
\begin{theorem} \label{thm:trace_FS_bound_global}
Suppose the cost function is $\epsilon$-close to the optimal cost in symmetric cloning
%
\begin{equation}
    C_{\Gbs}(\paramtheta) - C^{\opt}_{\Gbs} \leq \epsilon
    \label{eq:costtoepsilon_global_FS}
\end{equation}
where $\Cbs^{\opt}_{\Gbs} := 1 - F_{\Gbs}^{\opt}$. Then,
\begin{equation}     \label{eqn:global_fubini_study_bound_squared_appendix}
    \textrm{D}_{\mathrm{FS}}(\rho^{\psi}_{\paramtheta}, \rho^{\psi}_{\opt}) \leqslant \frac{\mathcal{N}\epsilon}{\sin(F_{\Gbs}^{\opt})} =: f_4(\epsilon),   \qquad \forall \ket{\psi} \in \mathcal{S}
\end{equation}
and,
\begin{equation}\label{eq:globaltracecloseness}
        \textrm{D}_{\mathrm{Tr}}(\rho^{\psi}_{\opt}, \rho^{\psi}_{\paramtheta}) \leqslant  \frac{1}{2}\sqrt{4F^{\opt}_{\Gbs}(1 - F^{\opt}_{\Gbs}) + \mathcal{N}\epsilon(1 - 2F^{\opt}_{\Gbs})}  =: g_4(\epsilon) , \qquad \forall \ket{\psi} \in \mathcal{S}
\end{equation}
\end{theorem}
\begin{proof}
The proof of the theorem~\ref{thm:trace_FS_bound_global} follows along the same lines as the proof of closeness of the Fubini-study distance for standard local cost function as provided in theorem~\ref{thm:local_cost_FS_weak_faithful_appendix} and the closeness of trace distance as provided in the theorem~\ref{thm:local_cost_trace_weak_faithful_appendix}.
\end{proof}
\subsection{Relationship between Local and Global Cost Functions} \label{app_ssec:global_vs_local_cost_relationship}

An interesting relation that we derive about the global cost function $\Cbs_{\Gbs}(\paramtheta)$ and local cost function $\Cbs_{\Lbs}(\paramtheta)$ is in the following theorem~\ref{thm:relationship-local-global-appendix}. We note that such an inequality was also proven in the work of \cite{bravo-prieto_variational_2019}.
\begin{theorem}\label{thm:relationship-local-global-appendix}
For the general case of $M \rightarrow N$ cloning, the global cost function $\Cbs_{\Gbs}(\paramtheta)$ and the local cost function $\Cbs_{\Lbs}(\paramtheta)$ satisfy the inequality,
\begin{equation}
 \Cbs_{\Lbs}(\paramtheta) \leqslant \Cbs_{\Gbs}(\paramtheta) \leqslant N\cdot\Cbs_{\Lbs}(\paramtheta)  
\end{equation}
\end{theorem}
\begin{proof}
We first prove the first part of the inequality,
\begin{equation}
\begin{split}
    \Cbs_{\Gbs}(\paramtheta) - \Cbs_{\Lbs}(\paramtheta) &= \frac{1}{\mathcal{N}}\int_{\mathcal{S}}\textrm{Tr}((\mathsf{O}_{\Gbs}^{\psi} - \mathsf{O}_{\Lbs}^{\psi})\rho_{\paramtheta}^{\psi})d\psi \\
    &= \frac{1}{\mathcal{N}N}\int_{\mathcal{S}}\textrm{Tr}\left(\left(
    \sum\limits_{j=1}^N\left(\ketbra{\psi}{\psi}_j \otimes \mathds{1}_{\Bar{j}} - \ketbra{\psi}{\psi}_1 \otimes \cdots \ketbra{\psi}{\psi}_N \right)\right)\rho^{\psi}_{\paramtheta}\right) \geq 0\\
    &\qquad\implies \Cbs_{\Gbs}(\paramtheta) \geq \Cbs_{\Lbs}(\paramtheta)
\end{split}   
\label{eq:global-local}
\end{equation}
where $\mathsf{O}_{\Lbs}^{\psi}$ is defined in \eqref{eqn:local_cost_full}, and the inequality in the second line holds because
\begin{equation}
 \sum\limits_{j=1}^N\left(\ketbra{\psi}{\psi}_j \otimes \mathds{1}_{\Bar{j}} - \ketbra{\psi}{\psi}_1 \otimes \cdots \ketbra{\psi}{\psi}_N \right) =  \sum\limits_{j=1}^N\ketbra{\psi}{\psi}_j \otimes (\mathds{1}_{\Bar{j}} - \ketbra{\psi}{\psi}_{\Bar{j}}) \geq 0, \qquad \forall \ket{\psi} \in \mathcal{S}   
\end{equation}

For the second part of the inequality, we consider the operator $N \mathsf{O}_{\Lbs}^{\psi} - \mathsf{O}_{\Gbs}^{\psi}$,
\begin{equation}
     \begin{split}
         N \mathsf{O}_{\Lbs}^{\psi} - \mathsf{O}_{\Gbs}^{\psi} &= (N - 1)\mathds{1} - \sum_{j=1}^{N}\left(\ketbra{\psi}{\psi}_j \otimes \mathds{1}_{\Bar{j}}\right) + \ketbra{\psi}{\psi}_1 \otimes \cdots \ketbra{\psi}{\psi}_N \\
         &= \sum_{j = 1}^{N-1}\left(\mathds{1}_j \otimes \mathds{1}_{\Bar{j}} - \ketbra{\psi}{\psi}_j \otimes \mathds{1}_{\Bar{j}}\right) - \ketbra{\psi}{\psi}_{N} \otimes \mathds{1}_{\Bar{N}} + \ketbra{\psi}{\psi}_1 \otimes \cdots \ketbra{\psi}{\psi}_N \\
         &= \sum_{j = 1}^{N-1}\left((\mathds{1} - \ketbra{\psi}{\psi})_{j}) \otimes \mathds{1}_{\Bar{j}}\right) - \bigotimes_{j=1}^{N-1}(\mathds{1} - \ketbra{\psi}{\psi})_j) \otimes \ketbra{\psi}{\psi}_N \\
         &= (\mathds{1} - \ketbra{\psi}{\psi})_{1}\otimes \left(\mathds{1}_{\Bar{1}} - \bigotimes_{j=2}^{N-1}(\mathds{1} -  \ketbra{\psi}{\psi}_{j}) \otimes \ketbra{\psi}{\psi}_N\right) + \sum_{j = 2}^{N-1}\left((\mathds{1} - \ketbra{\psi}{\psi})_{j}) \otimes \mathds{1}_{\Bar{j}}\right) \\
         &\geq 0
     \end{split}
\end{equation}   
where the second last line is positive because each individual operator is positive for all $\ket{\psi} \in \mathcal{S}$.
\end{proof}
We however note that the inequality proven in theorem~\ref{thm:relationship-local-global-appendix} does not allow us make statements about the similarity of individual clones from the closeness of the global cost function and vice versa.  This can be seen as follows:
\begin{equation}
\begin{split}
    \Cbs_{\Gbs}(\paramtheta) - \Cbs_{\Gbs}^{\opt} \leqslant \epsilon &\implies  \Cbs_{\Lbs}(\paramtheta) - \Cbs_{\Lbs}^{\opt} \leqslant \epsilon - (\Cbs_{\Gbs}(\paramtheta) - \Cbs_{\Lbs}(\paramtheta)) + (\Cbs_{\Lbs}^{\opt} - \Cbs_{\Gbs}^{\opt})\\ 
    &\implies \Cbs_{\Lbs}(\paramtheta) - \Cbs_{\Lbs}^{\opt} \leqslant \epsilon + (\Cbs_{\Lbs}^{\opt} - \Cbs_{\Gbs}^{\opt}) \\ 
    & \qquad \nRightarrow \Cbs_{\Lbs}(\paramtheta) - \Cbs_{\Lbs}^{\opt} \leqslant \epsilon
\end{split}    
\end{equation}
Here we have used the result of \thmref{thm:relationship-local-global-appendix} that $\Cbs_{\Gbs}(\paramtheta) \geq \Cbs_{\Lbs}(\paramtheta)$ and we note that $\Cbs_{\Lbs}^{\opt} - \Cbs_{\Gbs}^{\opt} \neq 0$ for all the $M \rightarrow N$ cloning. In particular, for $1 \rightarrow 2$ cloning, $\Cbs_{\Lbs}^{\opt} = 5/6$, while $\Cbs_{\Gbs}^{\opt} = 2/3$. This is due to the non-vanishing property of these cost functions, even at the theoretical optimal.

However, this does not imply that there is no method to prove the closeness of the individual clones from the closeness statement of cost function $\Cbs_{\Gbs}^{\paramtheta}$. In the next section, we prove this closeness for universal and phase-covariant cloner in terms of strong faithfulness by leveraging the property that the optimal clone state is the same when optimised via the global cost function or the local cost function.   

\subsection{Local Faithfulness from Global Optimisation} \label{app_ssec:local_faithfulness_from_global_optimisation}
 
Here, we revisit the discussion in the conclusion of the previous section, regards to the relationship between the local and global cost functions. We establish this strong and weak faithfulness guarantees for the \emph{special cases} of universal and phase-covariant cloning and prove \thmref{thm:local_clones_from_global_cost_universal_maintext_universal} and \thmref{thm:local_clones_from_global_cost_universal_maintext_phase_covariant} from the main text:
\begin{theorem}[\thmref{thm:local_clones_from_global_cost_universal_maintext_universal} in main text] \label{thm:local_clones_from_global_cost_universal_appendix_universal}
The global cost function is locally strongly faithful for a universal symmetric cloner, i.e,:
\begin{equation}\label{eqn:strongly_faithful_global_cost_local_clone_strong_universal_appendix}
    \Cbs_{\Gbs}(\paramtheta) = \Cbs^{\mathrm{opt}}_{\Gbs} \iff \rho_{\paramtheta}^{\psi, j} = \rho_{\opt}^{\psi, j} \qquad \forall \ket{\psi} \in \mathcal{H}, \forall j \in \{1,\dots,N\}
\end{equation}
\end{theorem}

\begin{proof}
In the symmetric universal case, $\Cbs^{\opt}_{\Lbs}$ has a unique minimum when, each local fidelity saturates:
\begin{align} \label{eqn:univ-localfid}
    F_{\Lbs}^{\opt} = \frac{M(N+2) + N - M}{N(M+2)}
\end{align}
achieved by local reduced states, $\{\rho_{\opt}^{\psi, j}\}_{j=1}^N$.

Now, it has been shown that the optimal global fidelity $F_{\Gbs}$ that can be reached\cite{buzek_universal_1998, scarani_quantum_2005} is,
\begin{align} \label{eqn:univ-globalfid}
    F_{\Gbs}^{\opt} = \frac{N!(M+1)!}{M!(N+1)!}
\end{align}
which also is the corresponding unique minimum value for $\Cbs^{\opt}_{\Gbs}$, achieved by some global state $\rho_\opt^\psi$.

Finally, it was proven in Refs.\cite{werner_optimal_1998, keyl_optimal_1999} that the cloner which achieves one of these bounds is unique and also saturates the other, and therefore must also achieve the unique minimum of both global and local cost functions, $\Cbs^{\opt}_{\Gbs}$ and $\Cbs^{\opt}_{\Lbs}$. Hence, the local states which optimise $\Cbs^{\opt}_{\Lbs}$  must be the reduced density matrices of the global state which optimises $\Cbs^{\opt}_{\Gbs}$ and so:
\begin{equation}    \label{Eq:local-global-state-closeness}
    \rho_{\opt}^{\psi, j} := \Tr_{\Bar{j}}(\rho_{\opt}^{\psi}), \ \ \forall j
\end{equation}
Thus for a universal cloner, the cost function with respect to both local and global fidelities will converge to the same minimum.
\end{proof}

Now, before proving \thmref{thm:local_clones_from_global_cost_universal_maintext_phase_covariant}, we first need the following lemma (we return to the notation of $B, E$ and $E^*$ for clarity):

\begin{lemma}\label{lemma:phase-covariant-clone-uniqueness}
For any $1\rightarrow 2$ phase-covariant cloning machine which takes states $\ket{0}_{B}\otimes\ket{\psi}_{E}$ and an ancillary qubit $\ket{A}_{E^*}$ as input, where $\ket{\psi} := \frac{1}{\sqrt{2}}(\ket{0} + e^{i\theta}\ket{1})$, and outputs a 3-qubit state $\ket{\Psi_{BEE^*}}$ in the following form:
\begin{align} \label{eqn:phasecov-tripartite-output}
    \ket{\Psi_{BEE^*}} = \frac{1}{2}\left[\ket{0,0} + e^{i\phi}(\sin\eta\ket{0,1} + \cos\eta\ket{1,0}))\ket{0}_{E^*} + e^{i\phi}\ket{1,1} + (\cos\eta\ket{0,1} + \sin\eta\ket{1,0})\ket{1}_{E^*}\right]
\end{align}
 the global and local fidelities are simultaneously maximised at $\eta = \frac{\pi}{4}$ where $0\leq \eta \leq \frac{\pi}{2}$ is the `shrinking factor'. 
\end{lemma}
\begin{proof}
To prove this, we follow the formalism that was adopted by Cerf et. al. \cite{cerf_cloning_2002}. This uses the fact that a symmetric phase covariant cloner induces a mapping of the following form~\cite{scarani_quantum_2005}:
\begin{align} \label{eqn:phasecov-map}
    \begin{split}
        & \ket{0}\ket{0}\ket{0} \rightarrow \ket{0}\ket{0}\ket{0} \\
        & \ket{1}\ket{0}\ket{0} \rightarrow (\sin\eta\ket{0}\ket{1} + \cos\eta\ket{1}\ket{0})\ket{0} \\
        & \ket{0}\ket{1}\ket{1} \rightarrow (\cos\eta\ket{0}\ket{1} + \sin\eta\ket{1}\ket{0})\ket{1} \\
        & \ket{1}\ket{1}\ket{1} \rightarrow \ket{1}\ket{1}\ket{1}
    \end{split}
\end{align}
Next, we calculate the global state by tracing out the ancillary state to get $\rho_{\Gbs}^{\opt}$:
\begin{align} \label{eqn:phasecov-global-density}
    \rho_{\Gbs}^{\opt} = \textrm{Tr}_{E^*}(\ket{\Psi_{BEE^*}}\bra{\Psi_{BEE^*}}) = \ket{\Phi_1}\bra{\Phi_1} + \ket{\Phi_2}\bra{\Phi_2}
\end{align}
where $\ket{\Phi_1} := \frac{1}{2}\left[\ket{0,0} + e^{i\phi}(\sin\eta\ket{0,1} + \cos\eta\ket{1,0})\right]$ and $\ket{\Phi_2} := \frac{1}{2}\left[e^{i\phi}\ket{1,1} + (\cos\eta\ket{0,1} + \sin\eta\ket{1,0})\right]$. Hence the global fidelity can be found as
\begin{align} \label{eqn:phasecov_global_fid_appendix}
    F_{\Gbs}^{\opt} = \Tr(\ket{\psi}\bra{\psi}^{\otimes 2}\rho_{\Gbs}^{\opt}) = |\braket{\psi^{\otimes 2}}{\Phi_1}|^2 +|\braket{\psi^{\otimes 2}}{\Phi_2}|^2 = \frac{1}{8}(1 + \sin\eta + \cos\eta)^2
\end{align}
Now, optimising $F_{\Gbs}^{\opt}$ with respect $\eta$, we see that $F_{\Gbs}^{\opt}$ has only one extremum between $[0,\frac{\pi}{2}]$ specifically at $\eta=\frac{\pi}{4}$. We can also see that the local fidelity is also achieved for the same $\eta$ and is equal to:
\begin{align} \label{eqn:phasecov_local_fid_appendix}
    F_{\Lbs}^{\opt} = \frac{1}{2}\left(1 + \frac{\sqrt{2}}{2}\right)
\end{align}
which is the upper bound for local fidelity of the phase-covariant cloner.
\end{proof}

\begin{theorem} [\thmref{thm:local_clones_from_global_cost_universal_maintext_phase_covariant} in main text] \label{thm:local_clones_from_global_cost_phase-covariant_appendix}
The global cost function is locally strongly faithful for phase-covariant symmetric cloner, i.e.:
\begin{equation}\label{eqn:strongly_faithful_global_cost_local_clone_strong_phase_covariant_appendix}
    \Cbs_{\Gbs}(\paramtheta) = \Cbs^{\mathrm{opt}}_{\Gbs} \iff \rho_{\paramtheta}^{\psi, j} = \rho_{\opt}^{\psi, j} \qquad \forall \ket{\psi} \in \mathcal{S}, \forall j \in \{B, E\}
\end{equation}
where $\mathcal{S}$ is the distribution corresponding to phase-covariant cloning.  
\end{theorem}
\begin{proof}

Now, we have in \lemref{lemma:phase-covariant-clone-uniqueness} that the global and local fidelities of a phase covariant cloner are both achieved with a cloning transformation of the form in \eqref{eqn:phasecov-map}. Applying this transformation unitary to $\ket{\psi}\ket{\Phi^+}_{BE}$ (where $\ket{\Phi^+}_{BE}$ is a Bell state) leads to Cerf's formalism for cloning. Furthermore, we can observe that due to the symmetry of the problem, this transformation is unique (up to global phases) and so any optimal cloner must achieve it.

Furthermore, one can check that the ideal circuit in \figref{fig:qubit_cloning_ideal_circ} does indeed produce an output in the form of \eqref{eqn:phasecov-tripartite-output} once the preparation angles have been set for phase-covariant cloning. By a similar argument to the above, we can see that a variational cloning machine which ahieves an optimal cost function value, i.e. $\Cbs_{\Gbs}(\paramtheta) = \Cbs^{\mathrm{opt}}_{\Gbs}$ much also saturate the optimal cloning fidelities. Furthermore, by the uniqueness of the above transformation (\eqref{eqn:phasecov-map}) we also have that the local states of the variational cloner are the same as the optimal transformation, which completes the proof.

\end{proof}



\section{Sample Complexity of the Algorithm} \label{app_ssec:sample_complexity_appendix}

VQC requires classical minimisation of one of the cost functions $\Cbs(\paramtheta) :=\{\Cbs_{\textrm{sq}}(\paramtheta), \Cbs_{\Lbs}(\paramtheta)\}, \Cbs_{\textrm{asym}}(\paramtheta), \Cbs_{\Gbs}(\paramtheta) \}$ to achieve the optimal cost value. Such an iteration to find the minimum would, in practice, involve the following steps,
\begin{enumerate}
    \item Prepare an initial circuit with randomised $\paramtheta$ parameter. Input a state $\ket{\psi} \in \mathcal{S}$ into the circuit and compute the cost function value $\Cbs^{\psi}(\paramtheta)$ using the $\SWAP$ test\cite{buhrman_quantum_2001} to estimate overlap between the output clone and the input state.  Let $L$ denote the number of states $\ket{\psi}$ used in order to estimate $\Cbs^{\psi}(\paramtheta)$.
    \item Pick $K$ different states $\ket{\psi}$ picked uniformly randomly from the distribution $\mathcal{S}$ to estimate the `averaged' cost function $\Cbs(\paramtheta) = \mathop{\mathbb{E}}_{\substack{\ket{\psi} \in \mathcal{S}}}[\Cbs^{\psi}(\paramtheta)]$.
    \item Use the classical optimisation to iterate over $\paramtheta$ to converge to the optimal cost function value $\Cbs^{\opt}$.
\end{enumerate}

Estimating the overlap in the above steps is sufficient for our purposes, since the two coincide when at least one of the states is a pure state:
\begin{equation}\label{eqn:fidelity_equals_state_overlap_one_pure_state}
    F(\ketbra{\psi}{\psi}, \rho) = \bra{\psi}\rho\ket{\psi} = \Tr(\ketbra{\psi}{\psi} \rho)
\end{equation}
We note that the parameterised quantum approach to converge to the optimal cost function is a heuristic algorithm. Thus there are no provable guarantees on the number of $\paramtheta$ iterations to converge to $\Cbs^{\opt}$.  However, one can provide guarantees on the number of samples $L \times K$ required to estimate $\Cbs(\paramtheta)$ for a particular instance of $\paramtheta$ up to additive error $\epsilon'$.

\begin{theorem}[\thmref{thm:sample-complexity_main_text} in main text]\label{thm:sample-complexity-appendix}
The number of samples $L \times K$ required to estimate the cost function $\Cbs(\paramtheta)$ up to $\epsilon'$-additive closeness with a success probability $\delta$ is,
\begin{equation}\label{eqn:number-of-samples-appendix}
    L \times K =  \mathcal{O}\left(\frac{1}{\epsilon'^2}\log \frac{2}{\delta}\right)
\end{equation}
where $K$ is the number of distinct states $\ket{\psi}$ sampled uniformly at random from the distribution $\mathcal{S}$, and $L$ is the number of copies of each input state. 
\end{theorem}
\begin{proof}

We provide the proof for the cost function $\Cbs_{\Gbs}(\paramtheta)$. However, this proof extends in a straightforward manner to other cost functions. As a reminder, the global cost function is defined as,
\begin{equation}
    \Cbs_{\Gbs}^{\psi}(\paramtheta) = 1 - \bra{\phi}\rho_{\paramtheta}^{\psi}\ket{\phi} \qquad \implies \qquad \Cbs_{\Gbs}(\paramtheta) = 1 - \frac{1}{\mathcal{N}}\int_{\mathcal{S}}\bra{\phi}\rho^{\psi}_{\paramtheta}\ket{\phi}d\psi
\end{equation}
The estimation of $\Cbs_{\Gbs}^{\psi}(\paramtheta)$ requires the computation of the overlap $\bra{\phi}\rho_{\paramtheta}^{\psi}\ket{\phi}$. We denote $\ket{\phi} := \ket{\psi}^{\otimes N}$ to be the $N$-fold tensor product of the input state to be cloned. The $\SWAP$ test proposed by Buhrman et. al.\cite{buhrman_quantum_2001} is an algorithm to compute this overlap, and the circuit is given in \figref{ciruit:swaptest_global_and_local}

\begin{figure}
    \begin{center}
        \includegraphics[width=0.8\columnwidth, height=0.25\columnwidth]{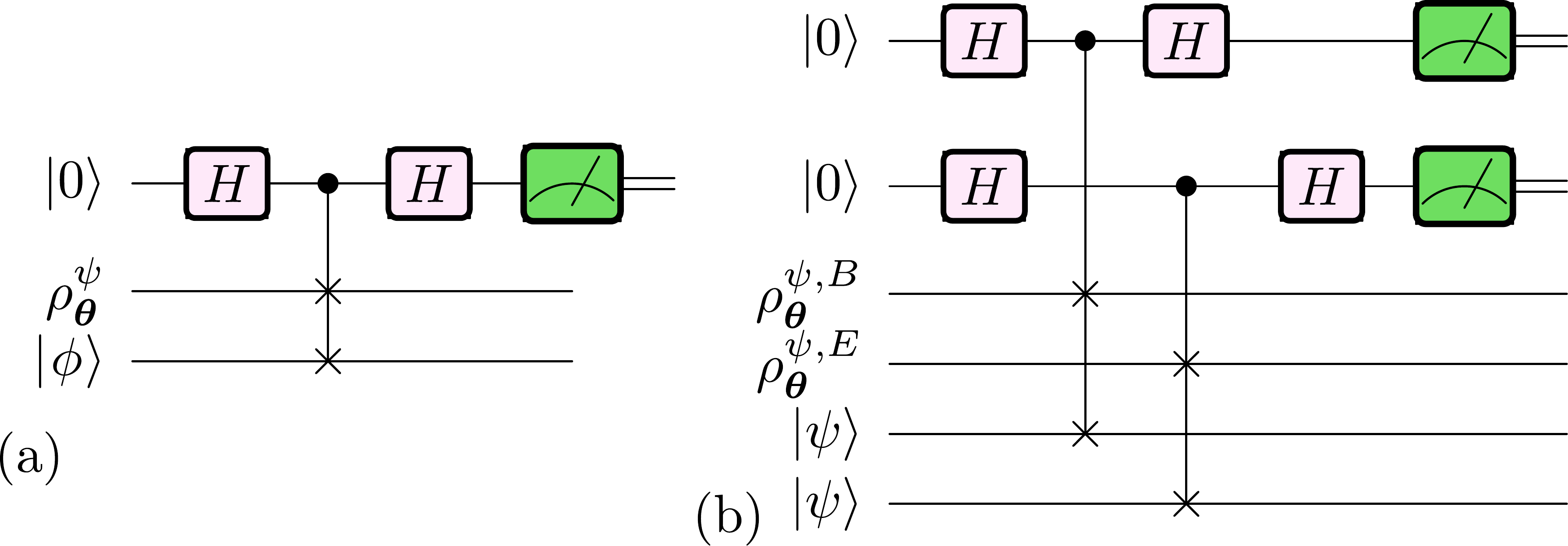}
    \caption{$\SWAP$ test circuit illustrated for $1\rightarrow 2$ cloning. In (a) for example, we compare the global state $\rho^{\psi}_{\paramtheta}$, with the state $\ket{\phi}$, where $\ket{\phi} := \ket{\psi}\otimes \ket{\psi}$ is the product state of two copies of $\psi$. (b) Local $\SWAP$ test with the reduced state of Bob and Eve separately. One ancilla is required for each fidelity to be computed. The generalisation for $N$ input states in $M \rightarrow N$ cloning is straightforward.}
    \label{ciruit:swaptest_global_and_local}
        \end{center}
\end{figure}

%
This test inputs the states $\ket{\phi}$ and $\rho^{\psi}_{\paramtheta}$ with an additional ancilla qubit $\ket{0}$, and measures the ancilla in the end in the computational basis. The probability of obtaining an outcome `1' in the measurement is,
\begin{equation}
 \mathbb{P}[\textrm{`1'}] = p^{\psi} = \frac{1}{2}(1 -  \bra{\phi}\rho_{\paramtheta}^{\psi}\ket{\phi})   
\end{equation}
From this, we see that $\Cbs^{\psi}_{\Gbs}(\paramtheta) = 2p^{\psi}$. To estimate this cost function value, we run the SWAP test $L$ times and estimate the averaged number of `1' outcomes. Let us the define the estimator for cost function with $L$ samples to be,
\begin{equation}
 \widehat{\Cbs}^{\psi}_{\Gbs, \textrm{avg}}(\paramtheta) = \frac{1}{L}\sum_{i=1}^{L} \widehat{\Cbs}^{\psi}_{\Gbs, i}(\paramtheta)    
\end{equation}
where $\widehat{\Cbs}^{\psi}_{\Gbs, i}(\paramtheta)$ is equal to 2 if the $\SWAP$ test outcome at the $i$-th run is `1' and is 0 otherwise. From this we can see that the expected value of $ \widehat{\Cbs}^{\psi}_{\Gbs, \textrm{avg}}(\paramtheta)$ is,
\begin{equation}
    \mathbb{E}\left[ \widehat{\Cbs}^{\psi}_{\Gbs, \textrm{avg}}(\paramtheta)\right] = \Cbs^{\psi}_{\Gbs}(\paramtheta) = 2p^{\psi}
\end{equation}
Now consider $K$ different states $\{\ket{\psi_1},\cdots\ket{\psi_K}\}$ are chosen uniformly at random from the distribution $\mathcal{S}$. The average value of the cost over $K$ is:
\begin{equation}
 \widehat{\Cbs}_{\Gbs, \textrm{avg}}(\paramtheta) = \frac{1}{K}\sum_{j=1}^{K}\widehat{\Cbs}^{\psi_j}_{\Gbs, \textrm{avg}}(\paramtheta) = \frac{1}{LK}\sum_{j=1}^{K}\sum_{i=1}^{L} \widehat{\Cbs}^{\psi_j}_{\Gbs, i}(\paramtheta)    
\end{equation}
with,
\begin{equation}
    \mathbb{E}\left[\widehat{\Cbs}_{\Gbs, \textrm{avg}}(\paramtheta)\right] = \frac{1}{K}\sum_{j=1}^{K}\frac{1}{\mathcal{N}}\int_{\mathcal{S}} \widehat{\Cbs}^{\psi}_{\Gbs, \textrm{avg}}(\paramtheta) d\psi = \Cbs_{\Gbs}(\paramtheta)
\end{equation}
Using H\"{o}effding's inequality~\cite{hoeffding_probability_1963}, one can obtain a probabilistic bound on $| \widehat{\Cbs}_{\Gbs, \textrm{avg}}(\paramtheta) - \Cbs_{\Gbs}(\paramtheta)|$,
\begin{equation}\label{eqn:Hoeffeding1-appendix}
    \mathbb{P}\left[| \widehat{\Cbs}_{\Gbs, \textrm{avg}}(\paramtheta) - \Cbs_{\Gbs}(\paramtheta)| \geq \epsilon'\right] \leqslant 2e^{-2KL\epsilon'^2}
\end{equation}
Now, setting $2e^{-2KL\epsilon'^2} = \delta$ and solving for $L\times K$ gives:
\begin{equation}
    L \times K = \mathcal{O}\left(\frac{1}{\epsilon'^2}\log \frac{2}{\delta}\right)
\end{equation}
\end{proof}

Finally, we note that the $\SWAP$ test, in practice, is somewhat challenging to implement on NISQ devices, predominately due to the compilation overhead of compiling the $3$-local controlled $\SWAP$ into the native gateset of a particular quantum hardware. Furthermore, in this case, we have a strict need for the copy of the input state $\ket{\psi}$ to be kept coherent while implementing the $\SWAP$ test, due to the equivalence between fidelity and overlap if one state is pure. This is due to the fact that for \emph{mixed} quantum states, there is no known efficient method to compute the fidelity exactly\cite{watrous_quantum_2002} and one must resort to using bounds on it, perhaps also discovered variationally\cite{chen_alternative_2002, mendonca_alternative_2008, puchala_bound_2009, cerezo_variational_2020}. In light of this, one could use the shorter depth circuits to compute the overlap found using a variational approach similar to that implemented here\cite{cincio_learning_2018}.



\section{Coin Flipping Protocols \& Cloning Attacks} \label{app_sec:coin_flipping_cloning_attacks}

Here we give more detail about the specific coin flipping protocols discussed in the main text, and calculations relating to the cloning attacks on them.

\subsection{The General Protocol, \texorpdfstring{$\mathcal{P}_1$}{} with \texorpdfstring{$k$}{} rounds} \label{app_ssec:mayers_protocol_and_attack}

For completeness, here we describe the general $k$ round protocol of Ref.~\cite{mayers_unconditionally_1999}, $\mathcal{P}_1$. In the general case, Alice and Bob now choose $k$ random bits, $\{a_1, \dots, a_k\}$ and $\{b_1, \dots, b_k\}$ respectively. The final bit is now be the \computerfont{XOR} of input bits over all $k$ rounds i.e.,
\begin{equation}
x = \bigoplus_j a_j  \oplus \bigoplus_j b_j   
\end{equation}
In each round $j = 1,\dots, k$ of the protocol, and for every step $i=1, \dots, n$ within each round, Alice uniformly picks a random bit $c_{i,j}$ and sends the state $\ket{\phi_c^{i,j}} := \ket{\phi_{c_{i,j}}} \otimes\ket{\phi_{\overline{c_{i,j}}})}$ to Bob. Likewise, Bob uniformly picks a random bit $d_{i,j}$ and sends the state $\ket{\phi_d^{i,j}} := \ket{\phi_{d_{i,j}}} \otimes\ket{\phi_{\overline{d_{i,j}}}}$ to Alice. Hence, each party sends multiple copies of either $\ket{\phi_0}\otimes\ket{\phi_1}$ or $\ket{\phi_1}\otimes\ket{\phi_0}$\footnote{Note that if $c_{i,j}$ and $d_{i,j}$ are chosen independently of $a_j$ and $b_j$, no information about the primary bits has been transferred.}.

In the next step, for each $j$ and $i$, Alice announces the value $a_j\oplus c_{i,j}$. If $a_j\oplus c_{i,j} = 0$, Bob returns the second state of the pair $(i,j)$ back to Alice, and sends the first state otherwise. Similarly Bob announces $b_j\oplus d_{i,j}$ and Alice returns one of the states back to Bob accordingly. Now, we come to why it is sufficient to only consider a single round in the protocol from the point of view of a cloning attack.  This is because a dishonest Bob can bias the protocol if he learns about Alice's bit $a_j$, which he can do by guessing $c_{i,j}$ with a probability better than $1/2$. With this knowledge, Bob only needs to announce a single false $b_j \oplus d_{i,j}$ in order to cheat, and so this strategy can be deferred to the final round~\cite{mayers_unconditionally_1999}. Hence a single round of the protocol is sufficient for analysis, and we herein drop the $j$ index. 

In the last phase of the protocol, after $a$ and $b$ are announced by both sides (so $x$ can be computed by both sides), Alice measures the remaining states with the projectors, $(E_{b}, E^{\perp}_{b})$ and the returned states by Bob with $(E_{\overline{a}}, E^{\perp}_{\overline{a}})$ (\eqref{eqn:mayers_povm_appendix}). She aborts the protocol if she gets the measurement result corresponding to $\perp$, and declares Bob as being dishonest. In this sense, the use of quantum states in this protocol is purely for the purpose of cheat-detection.

\begin{align}\label{eqn:mayers_povm_appendix}
 E_l &= \ket{\phi_l}\bra{\phi_l}^{\otimes n} \\
     E_l^\perp &= \mathds{1} - \ket{\phi_l}\bra{\phi_l}^{\otimes n}, \qquad l \in\{0, 1\}
\end{align}

\subsection{A Cloning Attack on \texorpdfstring{$\mathcal{P}_1$}{}} \label{app_ssec:attack_on_mayers_protocol}
Here we present the explicit attack and calculation that can be implemented by Bob on $\mathcal{P}_1^1$. WLOG, we assume that Bob wishes to bias the bit towards $x = 0$. First, we give the attack for when Alice only sends one copy of the state ($n=1$), and then discuss the general case:

\begin{attack}{Cloning Attack on $\mathcal{P}_1$ with one round.d} \label{attack:mayers}
\textit{Inputs.} Random bit for Alice ($a \leftarrow_R \{0, 1\}$) and Bob ($b\leftarrow_R  \{0, 1\}$). Bob receives a state $\ket{\phi_c^i}$.
\sbline
\textit{Goal.} A biased bit towards $0$, i.e.\@ $p(x=0) > 1/2$.
\sbline
\textit{The attack:}
\begin{enumerate}
\item for $i = 1, \dots, n$: 
   \begin{enumerate}
    \item \textbf{Step 1:} Alice announces $a \oplus c_i$. If $a \oplus c_i = 0$, Bob sends the second qubit of $\ket{\phi_c^i}$ to Alice, otherwise he sends the first qubit. 
    \item \textbf{Step 2:} Bob runs a $1\rightarrow 2$ state dependent cloner on the qubit he has to return to Alice, producing $2$ approximate clones. He sends her one clone and keeps the other.
    \item \textbf{Step 3:} Bob runs an optimal state discrimination on the remaining qubit and any other output of the cloner, and finds $c_1$ with a maximum success probability $P^{\opt}_{\mathrm{disc}, \mathcal{P}_1}$. He then guesses a bit $a'$ such that $P_{\mathrm{succ}, \mathcal{P}_1}(a' = a) := P^{\opt}_{\mathrm{disc}, \mathcal{P}_1}$.
     \item \textbf{Step 4:} If $a'\oplus b = 0$ he continues the protocol honestly and announces $b \oplus d_1$, otherwise he announces $a' \oplus d_1$. The remaining qubit on Alice's side is $\ket{\phi^i_{a}}$.
  \end{enumerate}
\end{enumerate}
\end{attack}

Now, we revisit the following theorem from the main text:

\begin{theorem}\label{thm:mayers_attack_bias_probability_appendix}[\thmref{thm:mayers_attack_bias_probability} in main text.]
Bob can achieve a bias of $\epsilon \approx 0.27$ using a state-dependent cloning attack on the protocol, $\mathcal{P}_1$ with a single copy of Alice's state.
\end{theorem}

\begin{proof}

As mentioned in the previous section, the final measurements performed by Alice on her remaining $n$ states, plus the $n$ states returned to her by Bob allow her to detect his nefarious behaviour. If he performed a cloning attack, the $\perp$ outcomes would be detected by Alice sometimes. We must compute both the probability that he is able to guess the value of Alice's bit $a$ (by guessing the value of the bit $c_1$), and the probability that he is detected by Alice. This would provide us with Bob's final success probability in cheating, and hence the bias probability.

At the start of the attack, Bob has a product state of either $\ket{\phi_0}\otimes\ket{\phi_1}$ or $\ket{\phi_1}\otimes\ket{\phi_0}$ (but he does not know which). In Step 2, depending on Alice's announced bit, Bob proceeds to clone one of the qubits, sends one copy to Alice and keeps the other to himself. As mentioned in the main text we can assume, WLOG, that Alice's announced bit is $0$. In this case, at this point in the attack, he has one of the following pairs: $\ket{\phi_0}\bra{\phi_0}\otimes\rho^1_c$ or $\ket{\phi_1}\bra{\phi_1}\otimes\rho^0_c$, where $\rho^1_c$ and $\rho^0_c$ are leftover clones for $\ket{\phi_1}$ and $\ket{\phi_0}$ respectively. 

Bob must now discriminate between the following density matrices:
\begin{align}\label{eqn:mayers_bob_pairs_discriminate_appendix}
   \rho_1 &= \ket{\phi_0}\bra{\phi_0}\otimes\ket{\phi_1}\bra{\phi_1}\\ \textrm{ and }  \qquad \rho_2 &= \ket{\phi_1}\bra{\phi_1}\otimes\rho^0_c
\end{align}

Alternatively, if Alice announced $a\oplus c_i = 1$, he would have:
\begin{align}\label{eqn:mayers_bob_pairs_discriminate_alternate_appendix}
    \rho_1 &= \ket{\phi_1}\bra{\phi_1}\otimes\ket{\phi_0}\bra{\phi_0}, \\ \textrm{ and }  \qquad  \rho_2 &= \ket{\phi_0}\bra{\phi_0}\otimes\rho^1_c
\end{align}
In either case, we have that the minimum discrimination error for two density matrices is given by the Holevo-Helstrom\cite{holevo_statistical_1973, helstrom_quantum_1969} bound as follows\footnote{This also is because the we assume a symmetric cloning machine for both $\ket{\phi_0}$ and $\ket{\phi_1}$. If this is not the case, the guessing probability is instead the average of the discrimination probabilities of both cases.}:
\begin{equation}\label{eqn:helstrom}
    P^{\opt}_{\mathrm{disc}} = \frac{1}{2} + \frac{1}{4}\norm{\rho_1 - \rho_2}_{\Tr} = \frac{1}{2} + \frac{1}{2}D_{\Tr}(\rho_1, \rho_2)
\end{equation}
The ideal symmetric cloning machine for these states will have an output of the form:
\begin{equation}\label{eqn:reduced_local_state_mayers_attack_appendix}
      \rho_c = \alpha\ket{\phi_0}\bra{\phi_0} + \beta \ket{\phi_1}\bra{\phi_1}
     + \gamma (\ket{\phi_0}\bra{\phi_1} + \ket{\phi_1}\bra{\phi_0})  
\end{equation} 
Where $\alpha, \beta$ and $\gamma$ are functions of the overlap $s = \braket{\phi_0}{\phi_1} = \cos{\frac{\pi}{9}}$. Now, using \eqref{eqn:mayers_bob_pairs_discriminate_appendix}, $\rho_2$ can be written as follows:
\begin{equation}\label{eqn:reduced_local_state_rho2_mayers_attack_appendix}
      \rho_2 = \alpha\ket{\phi_1}\bra{\phi_1}\otimes\ket{\phi_0}\bra{\phi_0} + \beta  \ket{\phi_1}\bra{\phi_1}\otimes\ket{\phi_1}\bra{\phi_1}
     + \gamma ( \ket{\phi_1}\bra{\phi_1}\otimes\ket{\phi_0}\bra{\phi_1} +  \ket{\phi_1}\bra{\phi_1}\otimes\ket{\phi_1}\bra{\phi_0})  
\end{equation} 
Finally by plugging in the values of the coefficients in \eqref{eqn:reduced_local_state_mayers_attack_appendix} for the optimal local cloning machine\cite{brus_optimal_1998} and finding the eigenvalues of $\sigma := (\rho_1 - \rho_2)$, we can calculate the corresponding value for \eqref{eqn:helstrom}, and recover the following minimum error probability:
\begin{equation}\label{eqn:mayers-cloning-min-error}
P_{\mathrm{fail}, \mathcal{P}_1} = P^{\mathrm{er}}_{\mathrm{disc}, \mathcal{P}_1} = 1 - P^{\opt}_{\mathrm{disc}, \mathcal{P}_1} \approx 0.214
\end{equation}

\noindent This means that Bob can successfully guess $c_1$ with ${P}^{1}_{\textrm{succ}, \mathcal{P}_1} = 78.5\%$ probability.

Now we look at the probability of a cheating Bob being detected by Alice. We note that whenever Bob guesses $a$ successfully, the measurements $(E_{b}, E^{\perp}_{b})$ will be passed with probability $1$, hence we use $(E_{\overline{a}}, E^{\perp}_{\overline{a}})$ where the states sent by Bob will be measured. Using \eqref{eqn:local_optimal_non_ortho_fidelity_1to2} with the value of overlap $s = \cos\left(\pi/9\right)$, the optimal fidelity is $F_{\mathsf{L}} \approx 0.997$ and so the probability of Bob getting caught is at most $1\%$. Putting this together with Bob's guessing probability for $a$ gives his overall success probability of $77.5\%$.

This implies that Bob is able to successfully create a bias of $\epsilon \approx 0.775 - 0.5 = 0.275$.

\end{proof}

We also have the following corollary, for a general $n$ number of states exchanged:

\begin{corollary}
    The probability of Bob successfully guessing $a$ over all $n$ copies has the property:
    \begin{equation} \label{eqn:bob_guess_prob_n_rounds_mayers_appendix}
        \lim\limits_{n\rightarrow \infty}{P}^{n}_{\mathrm{succ}, \mathcal{P}_1} = 1
    \end{equation}
\end{corollary}

\begin{proof}
If Bob repeats the above attack, \attref{attack:mayers}, over all $n$ copies, he will guess $n$ different bits $\{a'_{i}\}_{i=1}^n$. He can then take a majority vote and announce $b$ such that $a^* \oplus b = 0$, where we denote $a^*$ as the bit he guesses in at least $\frac{n}{2} + 1$ of the rounds.

If $n$ is even, he may have guessed $a'$ to be $0$ and $1$ an equal number of times. In this case, the attack becomes indecisive and Bob is forced to guess at random. Hence we separate the success probability for even and odd $n$ as follows: 
\begin{equation}\label{eqn:mayers_attack_guess_probability_appendix}
    P^n_{\textrm{succ}, \mathcal{P}_1} = 
    \begin{cases}
    \sum\limits_{k=\frac{n+1}{2}}^{n} \binom{n}{k}(1 - P_{\mathrm{fail}})^k P_{\mathrm{fail}}^{n-k} &n  \text{ odd,} \\
     \sum\limits_{k=\frac{n}{2} + 1}^{n} \binom{n}{k}(1 - P_{\mathrm{fail}})^k P_{\mathrm{fail}}^{n-k} + \frac{1}{2}\binom{n}{n/2}(1 - P_{\mathrm{fail}})^{\frac{n}{2}} P_{\mathrm{fail}}^{\frac{n}{2}}  & n \text{ even} \\
  \end{cases} 
\end{equation}
By substituting the value of $P_{\mathrm{fail}}$ one can see that the function is uniformly increasing with $n$ so $\lim\limits_{n\rightarrow \infty}{P}^{n}_{\textrm{succ}, \mathcal{P}_1} = 1$.
\end{proof}

Although as Bob's success probability in guessing correctly increases with $n$, the probability of his cheating strategy getting detected by Alice will also increase. We also note that this strategy is independent of $k$, the number of different bits used during the protocol.

\subsection{Attack I on \texorpdfstring{$\mathcal{P}_2$}{}} \label{app_sssec:aharonov_attack_I_computation}

Next, we analyse the Attack I on the protocol $\mathcal{P}_2$ and compute the success probability:
\begin{theorem}\label{thm:aharonov_attack_I_bias_probability_appendix}[\thmref{thm:aharonov_attack_I_bias_probability_maintext} in main text.]
Using a cloning attack on the protocol, $\mathcal{P}_2$, (in attack model I) Bob can achieve a bias:
\begin{equation}
    \epsilon^{\mathrm{I}}_{\mathcal{P}_2, \mathrm{ideal}} \approx 0.35
\end{equation}
\end{theorem}
We note first that this attack model (i.e.\@ using cloning) can be considered a constructive way of implementing the optimal discrimination strategy of the states Alice is to send. In order to bias the bit, Bob needs to discriminate between the four pure states in \eqref{eq:coinflip-st} or equivalently between the ensembles encoding $a=\{0,1\}$, where the optimal discrimination is done via a set of POVM measurements.

However, by implementing a cloning based attack, we can simplify the discrimination.
This is because the symmetric state-dependent cloner (which is a unitary) has the interesting feature that for either case ($a=0$ or $a=1$), the cloner's output is a pure state in the $2$-qubit Hilbert space. As such, the states (after going through the QCM) can be optimally discriminated via a set of projective measurements $\{P_v, P_{v^{\perp}}\}$, rather than general POVMs (as would be the case if the QCM was not used). So, using VQC to obtain optimal cloning strategies also is a means to potentially reduce resources for quantum state discrimination also. Now, we give the proof of \thmref{thm:aharonov_attack_I_bias_probability_appendix}:

\begin{proof}
The attack involves the global output state of the cloning machine. For this attack we can use the fixed overlap $1 \rightarrow 2$ cloner with the global fidelity given by \eqref{eqn:optimal_global_non_ortho_state_fidelity}:
\begin{equation}\label{eqn:optimal_fidelity_g_1_2}
    F^{{\mathrm{FO}, \textrm{opt}}}_{\mathsf{G}}(1,2) = \frac{1}{2}\left( 1 + s^{3} + \sqrt{1-s^{2}}\sqrt{1-s^{4}} \right) \approx 0.983
\end{equation}
where $s=\sin(2\phi) = \cos(\frac{\pi}{4})$ for $\mathcal{P}_2$. Also alternatively we can use the 4-state cloner which clones the two states with a fixed overlap plus their orthogonal set. For both of these cloners we are interested in the global state of the cloner which we denote as $\ket{\psi_{x, a}^{1\rightarrow 2}}$ for an input state $\ket{\phi_{x, a}}$.

In order for Bob to guess $a$ he must discriminate between $\ket{\phi_{0, 0}}$ (encoding $a=0$) and $\ket{\phi_{1, 1}}$ (encoding $a=1$) or alternatively the pair $\{\ket{\phi_{0, 1}}, \ket{\phi_{1, 0}}\}$\footnote{This is since the pairs $\{\ket{\phi_{0, 0}}, \ket{\phi_{0, 1}}\}$ are orthogonal and $\{\ket{\phi_{0, 0}}, \ket{\phi_{1, 0}}\}$ both encode $a=0$, so the only choice is to discriminate between $\ket{\phi_{0, 0}}$ and $\ket{\phi_{1, 1}}$.}. Due to the symmetry and without an ancilla, the cloner preserves the overlap between each pairs i.e.\@ $\braket{\psi_{0, 0}^{1\rightarrow 2}}{\psi_{1, 1}^{1\rightarrow 2}} = \braket{\phi_{0,0}}{\phi_{1,1}} = s$. (For the other two states of $\mathcal{P}_2$ we also have $\braket{\psi_{0, 1}^{1\rightarrow 2}}{\psi_{1, 0}^{1\rightarrow 2}} = s$).

\begin{figure}
    \centering
    \includegraphics[width=0.5\columnwidth, height = 0.3\columnwidth]{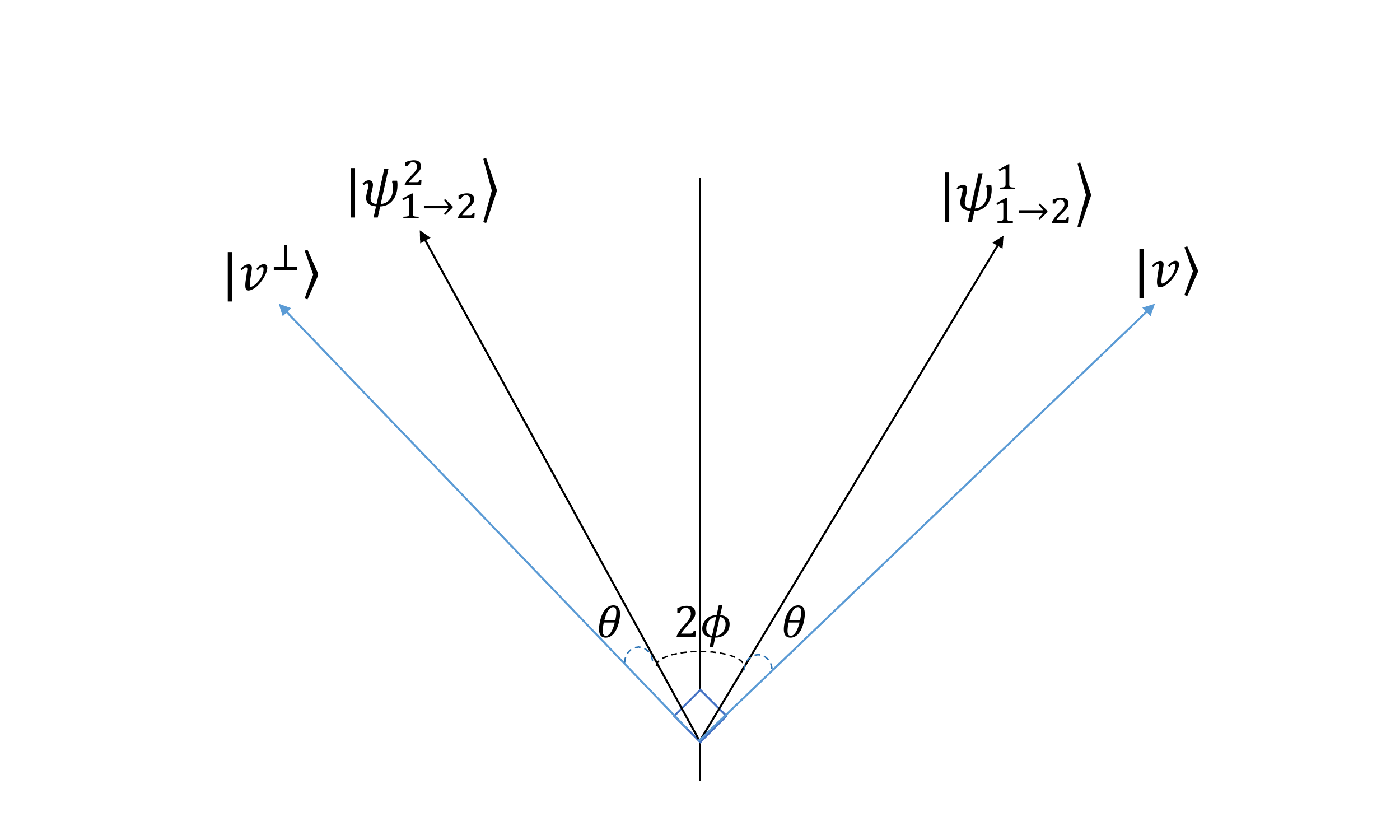}
    \caption{Discrimination of the output states of the state-dependant cloner with projective measurements.}
    \label{fig:pure-disc-a}
\end{figure}

\noindent Now we select the projective measurements $P_v=\ket{v}\bra{v}$ and $P_{v^{\perp}}=\ket{v^{\perp}}\bra{v^{\perp}}$ such that $\braket{v}{v^{\perp}} = 0$. One can show that the discrimination probability is optimal when $\ket{v}$ and $\ket{v^{\perp}}$ are symmetric with respect to the target states, illustrated in \figref{fig:pure-disc-a} according to the Neumark theorem. From the figure, we have that $\braket{v}{v^{\perp}} = 0$ so $2\theta + 2\phi = \frac{\pi}{2} \Rightarrow \theta = \frac{\pi}{4} - \phi$. Finally, writing the cloner's states for $\{\ket{\psi_{0, 0}^{1\rightarrow 2}}, \ket{\psi_{1,1}^{1\rightarrow 2}}\}$ in the basis $\{\ket{v}, \ket{v^{\perp}}\}$ gives:
\begin{equation}\label{eqn:coin_flip_output_cloner_states_in_v}
    \begin{split}
    & \ket{\psi_{0, 0}^{1\rightarrow 2}} = \cos\left(
    \frac{\pi}{4} - \phi\right)\ket{v} + \sin\left(
    \frac{\pi}{4} - \phi\right)\ket{v^{\perp}},\\
    & \ket{\psi_{1, 1}^{1\rightarrow 2}} = \cos\left(\frac{\pi}{4} - \phi\right)\ket{v} - \sin\left(\frac{\pi}{4} - \phi\right)\ket{v^{\perp}}
    \end{split}
\end{equation}
where it can be checked that $\braket{\psi_{0, 0}^{1\rightarrow 2}}{\psi_{1, 1}^{1\rightarrow 2}} = \cos\left(\frac{\pi}{2} - 2\phi\right) = \sin\left(2\phi\right) = s$. Hence $\ket{v}$ and $\ket{v^{\perp}}$ can be explicitly derived. Note that these bases are also symmetric with respect to the other pair i.e $\{\ket{\psi_{0, 1}^{1\rightarrow 2}}, \ket{\psi_{1,0}^{1\rightarrow 2}}\}$. Finally, the success probability of this measurement is then given by:
\begin{equation}\label{eqn:coin_flip_global_disc_probability_appendix}
P^{\opt, \mathrm{I}}_{\mathrm{disc}, \mathcal{P}_2} = \frac{1}{2} + \frac{1}{2}\braket{\psi_{0, 0}^{1\rightarrow 2}}{\psi_{1, 1}^{1\rightarrow 2}} =   \frac{1}{2} + \frac{1}{2}\sin{2\phi} = 0.853
\end{equation}
which is the maximum cheating probability for Bob. From this, we derive the bias as:
\begin{equation}
   \epsilon^{\mathrm{I}}_{\mathcal{P}_2, \mathrm{ideal}} =  P^{\opt, \mathrm{I}}_{\mathrm{disc}, \mathcal{P}_2} -   \frac{1}{2} = 0.353
\end{equation}
which completes the proof.
\end{proof}

\subsection{Attack II on \texorpdfstring{$\mathcal{P}_2$}{}} \label{app_sssec:aharonov_attack_II_computation}
Finally, we address the success probability of the attack II on the protocol, $\mathcal{P}_2$.
Here, we consider two scenarios:
\begin{enumerate}
    \item A cloning machine which is able to clone \emph{all} $4$ states $\ket{\phi_{0, 0}}, \ket{\phi_{1, 1}}$ \emph{and} $\ket{\phi_{0, 1}}, \ket{\phi_{1, 0}}$,
    \item A cloning machine tailored only the two states, $\ket{\phi_{0, 0}}$ and $\ket{\phi_{1, 1}}$ (which Bob needs to discriminate between).
\end{enumerate}

We focus on the former scenario, since it connects more cleanly with the VQC clone fidelities, but scenario 2 facilitates a more optimal attack (in the ideal scenario).

\noindent \textbf{Scenario 1:} In this case, we can compute an exact discrimination probability, but it will result in a less optimal attack.
\begin{theorem}\label{thm:aharonov_4state_attack_II_bias_probability_appendix}[\thmref{thm:aharonov_4state_attack_II_bias_probability_maintext} in main text.]
Using a cloning attack on the protocol, $\mathcal{P}_2$, (in attack model II) Bob can achieve a bias:
\begin{equation}\label{eqn:attack_4_state_aharonov_success_probability_exact_appendix}
    \epsilon^{\mathrm{II}}_{\mathcal{P}_2, \mathrm{ideal}} = 0.25
\end{equation}
\end{theorem}
\begin{proof}
Considering the 4 states to be in the $\mathsf{X} - \mathsf{Z}$ plane of the Bloch sphere, the density matrices of each state can be represented as:
\begin{equation} \label{eqn:4state-density-bloch}
\rho_{ij} = \frac{1}{2}(\mathds{1} + m_{ij}^x\sigma_x + m_{ij}^z\sigma_z)
\end{equation}
where $\sigma_x$ and $\sigma_z$ are Pauli matrices and $m_{ij}$ is a $3$ dimensional vector given by:
\begin{equation} \label{eqn:4state-cloner-vectors-bloch}
\begin{split}
    & m_{00} := [\sin(2\phi), 0, \cos(2\phi)]\\
    & m_{01} := [-\sin(2\phi), 0, -\cos(2\phi)]\\
    & m_{10} := [-\sin(2\phi), 0, \cos(2\phi)]\\
    & m_{11} := [\sin(2\phi), 0, -\cos(2\phi)]
\end{split}
\end{equation}
After the cloning (in the ideal case), the density matrix of each clone will become:
\begin{equation} \label{eqn:4state-cloner-density-bloch}
\rho_{ij}^c = \frac{1}{2}(\mathds{1} + \eta_x m_{ij}^x\sigma_x + \eta_z m_{ij}^z\sigma_z)
\end{equation}
where $\eta_x$ and $\eta_z$ are the shrinking factors in each direction given as follows:
\begin{equation} \label{eqn:4state-cloner-shrinking-factors}
\eta_x = \sin^2(2\phi)\sqrt{\frac{1}{\sin^4(2\phi) + \cos^4(2\phi)}}, \quad \quad 
\eta_z = \cos^2(2\phi)\sqrt{\frac{1}{\sin^4(2\phi) + \cos^4(2\phi)}}
\end{equation}
For the states used in $\mathcal{P}_2$, we have $\phi = \frac{\pi}{8}$ and hence $\eta_x = \eta_z := \eta =  \frac{1}{\sqrt{2}}$. Again, we can return to the discrimination probability between the two ensembles encoding $a=0$ and $a=1$ in \eqref{eqn:ensembles}. Here, we have:
\begin{align*} \label{eqn:4state-cloner-final-probability}
     P^{\opt, \mathrm{II}}_{\mathrm{disc}, \mathcal{P}_2} &= \frac{1}{2} + \frac{1}{4}\left|\left|\rho_{(a=0)} - \rho_{(a=1)}\right|\right|_{\Tr}\\
     &= \frac{1}{2} + \frac{1}{4}\left|\left|\frac{1}{2}\left[(\rho_{00}^c - \rho_{11}^c) + (\rho_{10}^c - \rho_{01}^c\right]\right|\right|_{\Tr} \\
     &= \frac{1}{2} + \frac{1}{4}\left|\left|\frac{\eta}{4}((m_{00}^x - m_{11}^x + m_{10}^x - m_{01}^x)\sigma_x + (m_{00}^z - m_{11}^z + m_{10}^z - m_{01}^z)\sigma_z \right|\right|_{\Tr} \\
       &= \frac{1}{2} + \frac{\eta\cos(2\phi)}{4}\left|\left|\sigma_z \right|\right|_{\Tr} \\
     &= \frac{1}{2} + \frac{\eta\cos(2\phi)}{2} = \frac{3}{4}
\end{align*}
Computing the bias in the same way as above completes the proof.
\end{proof}

\textbf{Scenario 2:} Here, we give a bound on the success probabilities of Bob in terms of the local fidelities of the QCM where the cloning machine is only tailored to clone two fixed-overlap states. Here we rely on the fact that Bob can discriminate between the two ensembles of states (for $a=0$, $a=1$) with equal probabilities. 

\begin{theorem}\label{thm:aharonov_2state_attack_II_bias_probability_appendix}
The optimal discrimination probability for a $2$-state cloning attack on the protocol, $A$, (in attack model II) is:
\begin{equation} \label{eqn:attack_2_state_aharonov_success_probability_bound_appendix}
     0.619 \leq P^{\opt, \mathrm{II}}_{\mathrm{disc}, \mathcal{P}_2} \leq 0.823
\end{equation}
\end{theorem}

\begin{proof}

For each of the input states, $\ket{\phi_{i,j}}$, in \eqref{eqn:aharonov_coinflip_states}, we denote  $\rho_{ij}^c$ to be a clone outputted from the QCM. Due to symmetry, we only need to consider one of the two output clones. We can now write the effective states for each encoding ($a=0, a=1$) as:
\begin{equation}\label{eqn:ensembles}
\rho_{(a=0)} := \frac{1}{2}(\rho_{00}^c + \rho_{10}^c), \qquad \qquad \rho_{(a=1)} := \frac{1}{2}(\rho_{01}^c + \rho_{11}^c)
\end{equation}
Dealing with these two states is sufficient since it can be shown that discriminating between these two density matrices, is equivalent to discriminating between the entire set of $4$ states in \eqref{eq:coinflip-st}.

Again, we use the discrimination probability from the Holevo-Helstrom bound:
\begin{equation}\label{eqn:opt-helstrom-clone}
 P^{\opt, II}_{\mathrm{disc}, \mathcal{P}_2} := P^{\opt}_{\mathrm{disc}}(\rho_{(a=0)},\rho_{(a=1)}) := \frac{1}{2} + \frac{1}{2}D_{\Tr}(\rho_{(a=0)},\rho_{(a=1)})
\end{equation}
Now, we have:
\begin{equation}\label{eqn:opt-trace-distance}
\begin{split}
   D_{\Tr}(\rho_{(a=0)},\rho_{(a=1)}) & = \frac{1}{2}\left|\left|\rho_{(a=0)} - \rho_{(a=1)}\right|\right|_{\Tr} = \frac{1}{2}\left|\left|\frac{1}{2}(\rho_{00}^c - \rho_{11}^c) + \frac{1}{2}(\rho_{10}^c - \rho_{01}^c)\right|\right|_{\Tr} \\
   & \leq \frac{1}{4}\left|\left|(\rho_{00}^c - \rho_{11}^c)\right|\right|_{\Tr} + \left|\left|(\rho_{10}^c - \rho_{01}^c)\right|\right|_{\Tr} \\
   & \leq \frac{1}{2} \left[D_{\Tr}(\rho_{00}^c, \rho_{11}^c) + D_{\Tr}(\rho_{01}^c,\rho_{10}^c)\right]\\
   \implies    P^{\opt}_{\mathrm{disc}}(\rho_{(a=0)},\rho_{(a=1)}) &\leq  \frac{1}{2} (P^{\opt}_{\mathrm{disc}}(\rho_{00}^c,\rho_{11}^c) + P^{\opt}_{\mathrm{disc}}(\rho_{01}^c,\rho_{10}^c))  \\
   &= P^{\opt}_{\mathrm{disc}}(\rho_{00}^c,\rho_{11}^c) 
\end{split}
\end{equation}
The last equality follows since for both ensembles, $\{\ket{\phi_{0,0}},\ket{\phi_{1,1}}\}$ and $\{\ket{\phi_{0,1}},\ket{\phi_{1,0}}\}$, we have that their output clones have equal discrimination probability:
\begin{equation}
    P^{\opt}_{\mathrm{disc}}(\rho_{00}^c,\rho_{11}^c)  = P^{\opt}_{\mathrm{disc}}(\rho_{01}^c,\rho_{10}^c) 
\end{equation}
This is because the QCM is symmetric, and depends only on the overlap of the states (we have in both cases $\braket{\phi_{00}}{\phi_{11}} = \braket{\phi_{01}}{\phi_{10}} = \sin(2\phi)$).

Furthermore, since the cloning machine can only lower the discrimination probability between two states, we have:
\begin{equation*}
   P^{\opt}_{\mathrm{disc}}(\rho_{00}^c,\rho_{11}^c) \leq P^{\opt}_{\mathrm{disc}}(\rho_{00}^c,\ket{\phi_{1,1}}\bra{\phi_{1,1}}) =: \overline{P^{\opt}_{\mathrm{disc}}}
\end{equation*}
Now, using the relationship between fidelity and the trace distance, we have the bounds:
\begin{equation}\label{eqn:opt_disc_probability_helstrom}
\frac{1}{2}+\frac{1}{2}\left(1 - \sqrt{\bra{\phi_{1,1}}\rho^c_{00}\ket{\phi_{1,1}}}\right) \le \overline{P^{\opt}_{\mathrm{disc}}} \le \frac{1}{2} + \frac{1}{2}\sqrt{1 - \bra{\phi_{1,1}}\rho^c_{00}\ket{\phi_{1,1}}}
\end{equation}
By plugging in the observed density matrix for the output clone, we can find this discrimination probability.
As in the previous section, the output density matrix from the QCM for an output clone can be written as \eqref{eqn:reduced_local_state_mayers_attack_appendix}:
\begin{equation}
   \rho^c_{00} = \alpha\ket{\phi_{0,0}}\bra{\phi_{0,0}} + \beta \ket{\phi_{1,1}}\bra{\phi_{1,1}} +
\gamma (\ket{\phi_{0,0}}\bra{\phi_{1,1}}+ \ket{\phi_{1,1}}\bra{\phi_{0,0}}) 
\end{equation}

Which has a local fidelity, $F_{\Lbs} = \bra{\phi_{0,0}}\rho^c_{00}\ket{\phi_{0,0}} = \alpha + s^2\beta + s\gamma$. On the other hand, we have $F(\rho^c_{00}, \ketbra{\phi_{1, 1}}{\phi_{1, 1}}) = \bra{\phi_{1, 1}}\rho^c_{00}\ket{\phi_{1, 1}} = s^2\alpha + \beta + s\gamma$.

Combining these two, we then have:
\begin{equation}
F(\rho^c_{00}, \ketbra{\phi_{1, 1}}{\phi_{1, 1}}) = F_{\Lbs} + (s^2 - 1)(\alpha - \beta)
\end{equation}
Plugging in $F_{\Lbs}$ from \eqref{eqn:local_optimal_non_ortho_fidelity_1to2}, and $\alpha - \beta = \sqrt{\frac{1-s^2}{1-s^4}}$ (for an optimal state-dependent cloner), we get:
\begin{equation} \label{eqn:attack_2_state_aharonov_success_probability_bound_not_filled_appendix}
    \frac{1}{2} + \frac{1}{2}\left[1 - \sqrt{F_{\Lbs} + (s^2 - 1) \sqrt{\frac{1-s^2}{1-s^4}}}\right] \leq P^{\opt, \textrm{II}}_{\mathrm{disc}, \mathcal{P}_2} \leq \frac{1}{2} + \frac{1}{2}\sqrt{1 - F_{\Lbs} - (s^2 - 1) \sqrt{\frac{1-s^2}{1-s^4}}}
\end{equation}
To complete the proof, we use $F_{\Lbs} \approx 0.989$ and $s = 1/\sqrt{2}$ which gives the numerical discrimination probabilities above.

\end{proof}


\section{Quantum Circuit Structure Learning}\label{app_b:structure_learning}
Here we give further details about the structure learning approach we implement, inspired by Refs.\cite{cincio_learning_2018, cerezo_variational_2020}. This approach fixes the length, $l$, of the circuit sequence to be used, and as mentioned in the main text contains parameterised single qubit gates, and un-parameterised entangling gates, which we chose to be $\CZ$ for simplicity. For example, with a three qubit chip, we have:
\begin{equation} \label{eqn:phase_covariant_cloning_gateset_appendix}
    \mathcal{G} = \left\{\right. \mathsf{R}^0_{z}(\theta), \mathsf{R}^1_{z}(\theta), \mathsf{R}^2_{z}(\theta),
    \mathsf{R}^0_{x}(\theta), \mathsf{R}^1_{x}(\theta), \mathsf{R}^2_{x}(\theta), 
    \mathsf{R}^0_{y}(\theta), \mathsf{R}^1_{y}(\theta), \mathsf{R}^2_{y}(\theta), 
    \CZ_{0, 1}, \CZ_{1, 2}, \CZ_{0, 2}\left.\right\}
\end{equation}
We use the $\CZ$ gate as the entangler for two reasons. The first is that $\CZ$ is a native entangling gate on the Rigetti hardware. The second is that it simplifies our problem slightly, since it is symmetric on the control and target qubit, we do not need to worry about the ordering of the qubits: $\CZ_{i, j} = \CZ_{j, i}$. The fixed angle $\mathsf{R}_x(\pm \pi/2)$ and continuous angle $\mathsf{R}_z(\theta)$ gates are also native on the Rigetti hardware and we add the $\mathsf{R}_y$ gate for completeness, which can be compiled into the above as follows, $\mathsf{R}_{y}(\theta) = \mathsf{R}_x(\pi/2)\mathsf{R}_z(\theta)\mathsf{R}_x(-\pi/2)$.  The unitary to be learned is given by: 
\begin{align}\label{eqn:structure_learning_unitary}
U_{\boldsymbol{g}}(\paramtheta) = U_{g_1}(\theta_1)U_{g_2}(\theta_2) \dots U_{g_l}(\theta_l)
\end{align}
where each gate is from the above set $\mathcal{G}$. The sequence, $\boldsymbol{g} := [g_1, \dots, g_l]$, in \eqref{eqn:variable_structure_ansatz_optimisation_problem} in the main text and \eqref{eqn:structure_learning_unitary} above, corresponds to the indices of the gates in an ordered version of $\mathcal{G}$. So using $\mathcal{G}$ in \eqref{eqn:phase_covariant_cloning_gateset_appendix} as an example, $\boldsymbol{g} = [0, 6, 3, 2, 10]$ would give the unitary:
\begin{align}\label{eqn:structure_learning_unitary_example}
U_{\boldsymbol{g}}(\paramtheta) =  \mathsf{R}^0_{z}(\theta_1)\mathsf{R}^1_{y}(\theta_2)\mathsf{R}^0_{x}(\theta_3) \mathsf{R}^2_{z}(\theta_4) \CZ_{0, 1}
\end{align}
and $\paramtheta := [\theta_1, \theta_2, \theta_3, \theta_4, 0]$
 The procedure of Refs.\cite{cincio_learning_2018, cincio_machine_2020, du_quantum_2020, li_quantum_2020} is intentionally flexible, and the gateset above \eqref{eqn:phase_covariant_cloning_gateset} can be swapped with any native gateset to fit on a particular quantum hardware.
 
 At the beginning of the procedure, the gate sequence is chosen randomly (a random sequence, $\boldsymbol{g}$), and also the parameters ($\paramtheta$) therein\footnote{If some information is known about the problem beforehand, this could be used to initialise the sequence to improve performance.}.

The optimisation procedure proceeds over a number of \computerfont{epochs} and \computerfont{iterations}. In each \computerfont{iteration}, $\boldsymbol{g}$ is perturbed by altering $d$ gates, $\boldsymbol{g}^{\computerfont{\text{iter}}} \rightarrow \boldsymbol{g}^{\computerfont{\text{iter}} +1}$. The probability of changing $d$ gates is given by $1/2^d$, and the probability of doing nothing (i.e.\@ $\boldsymbol{g}^{\computerfont{\text{iter}}} = \boldsymbol{g}^{\computerfont{\text{iter}} + 1} $) is:
\begin{equation}
    \Pr(d=0) = 1 - \sum_{d=1}^l\frac{1}{2^d} = 2 - \frac{1 - \frac{1}{2^l}}{ 1- \frac{1}{2} } - \frac{1}{2^l}
\end{equation}
The \computerfont{epochs} corresponding to optimisation of the parameters $\paramtheta$ using gradient descent with the Adam optimiser, as throughout the main text. We typically set the maximum number of \computerfont{epochs} to be $100$ and \computerfont{iterations} to be $50$ in all this work. After each \computerfont{iteration}, the best cost, $\Cbs^{\textrm{best}}_t$ for a chosen cost - either the local, \eqref{eqn:local_cost_full} ($t = \Lbs$), the global, \eqref{eqn:global_cost_full} ($t = \Gbs$), squared, \eqref{eqn:squared_local_cost_mton} ($t = \sq$) or some other choice, is updated, if that \computerfont{iteration} has found a circuit with a lower cost. As in Ref.\cite{cincio_learning_2018}, we repeatedly compress the sequence by removing redundant gates (e.g.\@ combining $U_{g_i}(\theta_i)$ and $U_{g_{i+1}}(\theta_{i+1})$ if $g_i = g_i+1$), and adding random gates to keep the sequence length fixed at $g_l$.

\figref{fig:all_runs_structure_learning} illustrates some results from this protocol. We find that with an increasing sequence length, the procedure is more likely to find circuits which achieve the minimum cost, and is able to first do so with a circuit with between $25$-$30$ gates from the above gateset in \eqref{eqn:phase_covariant_cloning_gateset}. We also plot the results achieved in a particular run of the protocol in \figref{fig:all_runs_structure_learning}(b). As the circuit learns, it is able to subsequently lower $\Cbs^{\textrm{best}}_t$, until it eventually finds a circuit capable of achieving the optimal cost for the problem.
\begin{figure}
    \centering
        \includegraphics[width=0.8\columnwidth,height=0.3\textwidth]{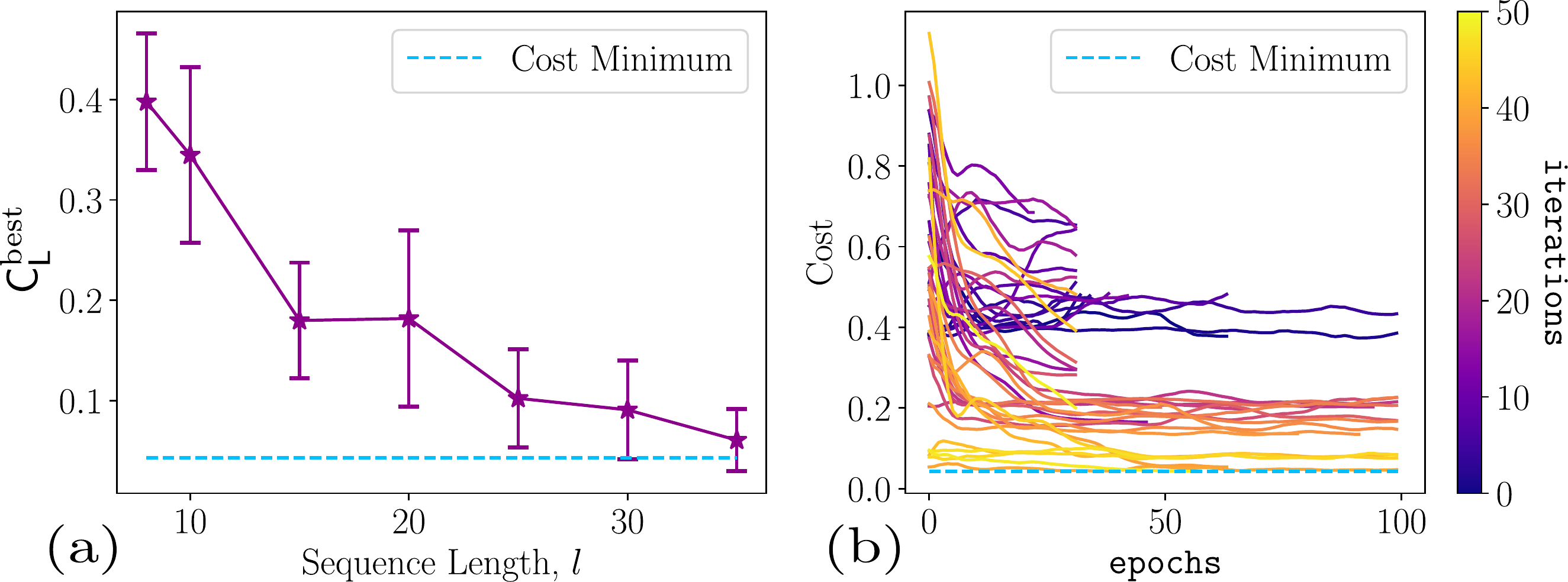}
    \caption{(a) $\Cbs^{\textrm{best}}_{\Lbs}$ as a function of sequence length, $l$ in achieving the same task as \figref{fig:qubit_cloning_ideal_circ}, where the Bob and Eve's clones appear in qubits $1$ and $2$. As $l$ increases, the number of runs which successfully approach the theoretical minimum increases. Error bars shown mean and standard deviations for the minimum costs achieved over $20$ independent runs with each sequence length. (b) Cost achieved for 50 iterations of the structure learning protocols, using a sequence length of $l=35$. Each line corresponds to a slightly different circuit structure, $\boldsymbol{g}$. Early iterations (darker lines) are not able to find the minimum, but eventually, a circuit is found which has this capacity. For each $\boldsymbol{g}$, $\paramtheta$ is trained for $100$ epochs of gradient descent, using the Adam optimiser. If an \computerfont{iteration} has not converged close enough to $\Cbs^{\textrm{best}}_{\Lbs}$ by $30$ epochs, the \computerfont{iteration} is ended.}
    \label{fig:all_runs_structure_learning}
\end{figure}
%


\section{Supplemental Numerical Results} \label{app_sec:supplemental_numerical_results}

\subsection{Phase Covariant Cloning with a Fixed Ideal Ansatz} \label{app_ssec:pc_cloning_fixed_ansatz}

As discussed in the main text, the ideal circuit for performing phase-covariant cloning is given by \figref{fig:qubit_cloning_ideal_circ}. Here, we learn the parameters of this fixed circuit. This gives us the opportunity to illustrate the effect of measurement noise in using the $\SWAP$ test to compute the fidelity. The results of this can be seen in \figref{fig:fixed_structure_phase_cov_cloning_swap_vs_no_swap}. We compare the $\SWAP$ test in \figref{fig:fixed_structure_phase_cov_cloning_swap_vs_no_swap}(a) to direct simulation of the qubit density matrices in \figref{fig:fixed_structure_phase_cov_cloning_swap_vs_no_swap} (b) to compute the fidelity. The effect of measurement noise can be clearly seen in the latter case. 

We note in the main text that we do not use the $\SWAP$ test when running the experiments on the \computerfont{Aspen} QPU. This is because the test fails to output the fidelity since both states to compare will be mixed due to device noise.
\begin{figure}
    \centering
        \includegraphics[width=0.8\columnwidth,height=0.3\textwidth]{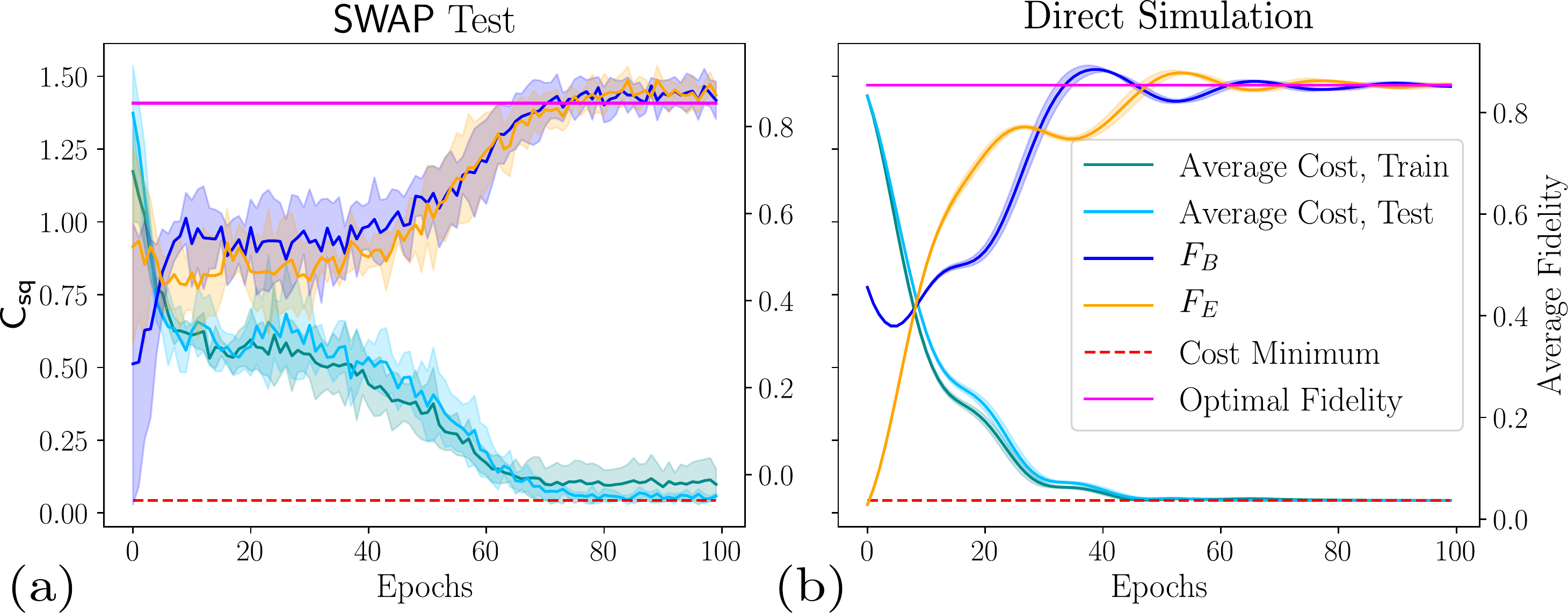}
    \caption{Learning the parameters of the fixed circuit in \figref{fig:qubit_cloning_ideal_circ}. We use $30$ random samples with an $80/20\%$ train/test split. To train, we use the analytic gradient, \eqref{eqn:analytic_squared_grad_mton}, and the Adam optimiser with a batch size of $10$ states and an initial learning rate of $0.05$. In all cases, the error bars show mean and standard deviation over $5$ independent training runs. Figure shows the results when the fidelity is computed using (a) the $\SWAP$ test (with $50$ measurement shots) and (b) using direct density matrix simulation. In both cases, we plot the average (squared) cost (\eqref{eqn:squared_local_cost_mton}) on the train and test set, and also the average fidelities of the output states of Bob, $F_B$, and Eve, $F_E$, corresponding to this cost function value. Also plotted are the theoretical optimal fidelities (magenta solid line) for this family of states, and the corresponding cost minimum (red dash line).}
    \label{fig:fixed_structure_phase_cov_cloning_swap_vs_no_swap}
\end{figure}
However, this essentially reproduces the findings of Ref.\cite{jasek_experimental_2019} in a slightly different scenario. Furthermore, this was only possible because we had prior knowledge of an optimal circuit to implement the cloning transformation from Ref.~\cite{buzek_quantum_1997, fan_quantum_2014}. Of course, in generality this information is not available, and so we favour the variable structure Ansatz discussed above.

\subsection{Phase Covariant Cloning with a Fixed Hardware Efficient Ansatz} \label{app_ssec:pc_cloning_fixed_hardware_efficient_ansatz}

We also test a hardware efficient fixed structure $\Ansatz$ for the sample problem as in the previous section. Here, we introduce a number of layers in the $\Ansatz$, $K$, in which each layer has a fixed structure. For simplicity, we choose each layer to have parameterised single qubit rotations, $\mathsf{R}_y(\theta)$, and nearest neighbour $\CZ$ gates. We deal again with $1\rightarrow 2$ cloning, so we use $3$ qubits and therefore we have $2$ $\CZ$ gates per layer. We show the results for $K=1$ layer to $K=6$ layers in \figref{fig:hardware_efficient_ansatz_phase_cov_cloning}.

\begin{figure}
    \centering
        \includegraphics[width=0.5\columnwidth,height=0.35\textwidth]{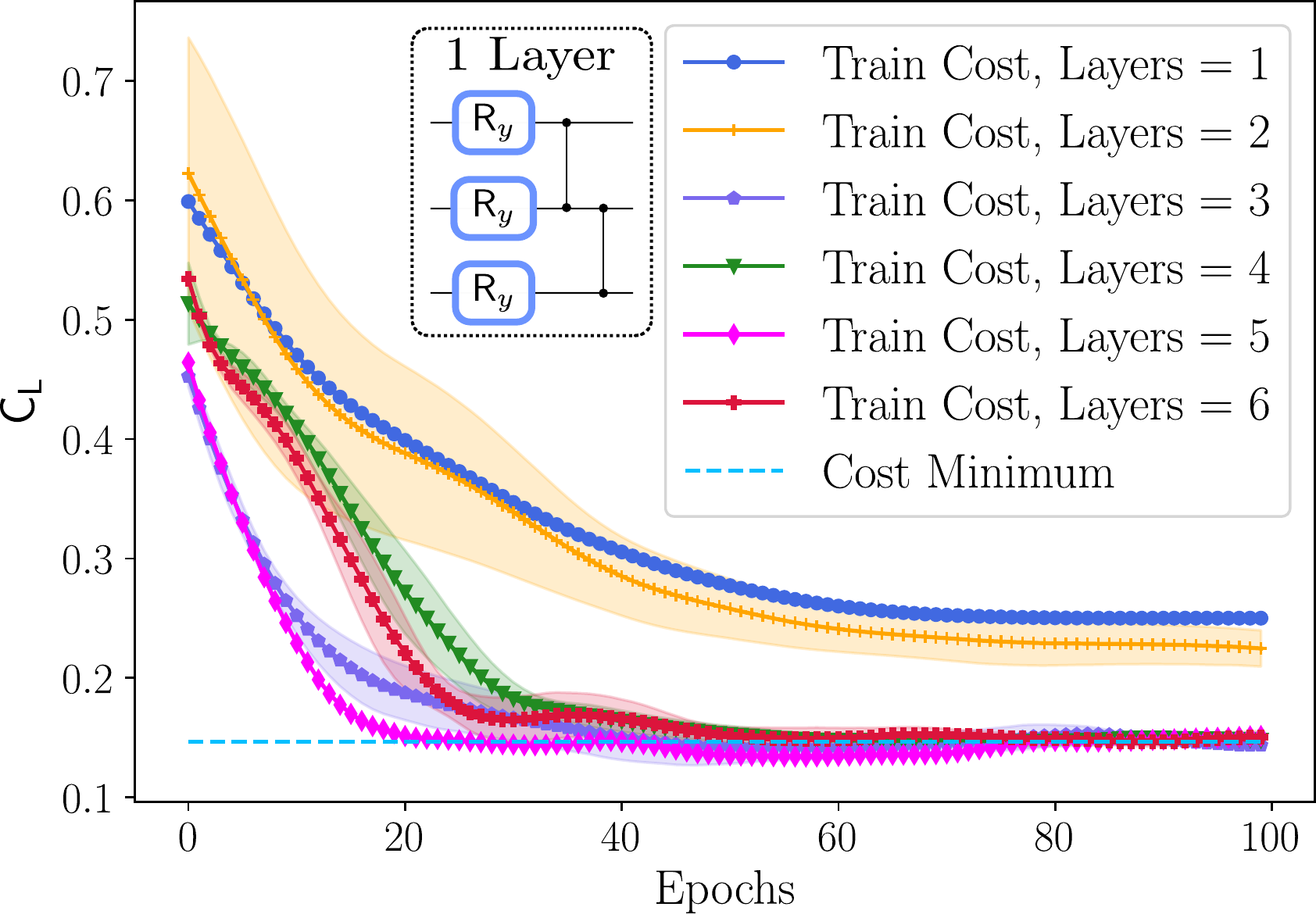}
    \caption{Local cost, $\Cbs_{\Lbs}$ minimised on a training set of $24$ random phase-covariant states. We plot layers $L\in [1, \dots 6]$ of the hardware-efficient Ansatz shown in the inset. Fastest convergence is observed for $L=5$ but $L=3$ is sufficient to achieve a minimal cost value, which is the same number of entangling gates as in \figref{fig:qubit_cloning_ideal_circ}. Error bars shown mean and standard deviation over $5$ independent training runs.}
    \label{fig:hardware_efficient_ansatz_phase_cov_cloning}
\end{figure}
Not surprisingly, we observe convergence to the minimum as the number of layers increases, saturating at $K=3$.
Furthermore, we can examine VQC for the existence of barren plateaus in this scenario. We do this specifically for a local cost, given by:
\begin{align}
    \Cbs_{\Lbs} &= \kappa\Tr\left[\mathsf{O}_{\Lbs}U(\paramtheta)\rho U(\paramtheta)^{\dagger}\right],\\
    \mathsf{O}_{\Lbs} &= c_0\mathds{1} + \sum_{j}c_j\mathsf{O}_j
\end{align}
Note that taking $c_0 = 1, c_j = -1/N~\forall j$ and $\kappa = 1$ recovers the specific form of our cost, \eqref{eqn:local_cost_full}. We will prove that this cost does not exhibit barren plateaus for a sufficiently shallow alternating layered Ansatz, i.e. $U(\paramtheta)$ contains blocks, $W$, acting on alternating pairs of qubits\cite{cerezo_cost-function-dependent_2020}. To do so, we first recall the following Theorem from Ref.\cite{cerezo_cost-function-dependent_2020}:
\begin{theorem}[Adapted from Theorem $2$ in Ref.\cite{cerezo_cost-function-dependent_2020}]
Consider a trainable parameter, $\theta^l$ in a block, $W$ of an alternating layered Ansatz (denoted $U(\paramtheta)$). Let $\text{Var}[\partial_l \Cbs]$ be the variance of an $m$-local cost function, $\Cbs$ with respect to $\theta^l$. If each block in $U(\paramtheta)$ forms a local $2$-design, then $\text{Var}[\partial_l \Cbs]$ is lower bounded by:
\begin{align}
    G_N(K, k) &\leq \text{Var}[\partial_l \Cbs]   \\
     G_N(K, k) &= \frac{2^{m(k+1) - 1}}{(2^{2m} - 1)^2(2^m+1)^{K+k}} \sum\limits_{j \in j_{\mathcal{L}}} \sum\limits_{\substack{(p, p')\in p_{\mathcal{L}_B} \\ p' \geq p}}c_j^2 D_{\text{HS}}(\rho_{p, p'}, \Tr(\rho_{p, p'})\mathds{1}/d_{\rho_{(p, p')}})D_{\text{HS}}(\mathsf{O}_j, \Tr(\mathsf{O}_j)\mathds{1}/d_{\mathsf{O}_j})
\end{align}
$j_\mathcal{L}$ are the set of $j$ indices in the forward light-cone $\mathcal{L}_B$ of the block $W$ and $\rho_{p, p'}$ is the partial trace of the input state, $\rho$, down to the subsystems $S_p, S_{p+1}, \dots, S_{p'}$. $d_M$ denotes the dimension of a matrix $M$.
\end{theorem}
$S_p$ in the above represents the qubit subsystem in which $W$ acts. Firstly, the operators $\mathsf{O}_j$ are all single qubit projectors ($m=1$ local), $\ketbra{\psi}{\psi}$, so we have:
\begin{align}
    D_{\text{HS}}\left(\mathsf{O}_j, \Tr(\mathsf{O}_j)\frac{\mathds{1}}{d}\right) &=  D_{\text{HS}}\left(\ketbra{\psi}{\psi}, \Tr\left[\ketbra{\psi}{\psi}\right]\frac{\mathds{1}}{2}\right) =
     \sqrt{\Tr\left[\left(\ketbra{\psi}{\psi} - \frac{\mathds{1}}{2} \right)\left(\ketbra{\psi}{\psi} - \frac{\mathds{1}}{2}\right)^{\dagger}\right]}\\
    &= \sqrt{\Tr\left[\ketbra{\psi}{\psi} - \frac{\ketbra{\psi}{\psi}}{2} - \frac{\ketbra{\psi}{\psi}}{2} + \frac{\mathds{1}}{4}\right]} = \sqrt{\Tr\left(\frac{\mathds{1}}{4}\right)} = \frac{1}{\sqrt{2}}
\end{align}
So $G(K, k)$ simplifies to:
\begin{equation}
     G_N(K, k) = \frac{2^{k}}{3^{K+k+2}\sqrt{2}N^2} \sum\limits_{j \in j_{\mathcal{L}}} \sum\limits_{\substack{(p, p')\in p_{\mathcal{L}_B} \\ p' \geq p}} D_{\text{HS}}(\rho_{p, p'}, \mathds{1}/d_{\rho_{(p, p')}})
\end{equation}
If we now define $S_{P}$ to be the subsystems from $p$ to $p'$, the reduced state of $\rho$ in $S_P$ will be one of either $\ketbra{\psi}{\psi}^{\otimes |P|}$, $\ketbra{0}{0}^{\otimes |P|}$ or $\ketbra{\psi}{\psi}^{\otimes q}\ketbra{0}{0}^{\otimes |P|-q}$ for some $q<|P|$ where we denote $|P|$ to be the number of qubits in the reduced subsystem $S_P$. Since these are all pure states, we can compute $D_{\text{HS}}(\rho_{p, p'}, \mathds{1}/d_{\rho_{(p, p')}}) = \sqrt{1-1/d_{\rho_{(p, p')}}}$. Lower bounding the sum over $j$ by $1$ and $\sqrt{1-1/d_{\rho_{(p, p')}}}$ by $1/\sqrt{2}$ ($d_{\rho_{(p, p')}}$ is at least $2$) gives:
\begin{equation}
     \frac{2^{k}}{3^{K+k+2}2N^2} \leq G_N(K, k)
\end{equation}
Finally, by choosing $K \in \mathcal{O}\left(\log(N)\right)$, we have that $k, K+k \in \mathcal{O}(\log(N))$ and so $G_n(K, k) \in \Omega\left(1/\text{poly}(N)\right)$.
Since we have that if $G(K, k)$ vanishes no faster than $\Omega\left(1/\text{poly}(N)\right)$, then so does the variance of the gradient and so will not require exponential resources to estimate. As a result, we can formalise the following corollary:
\begin{corollary}\label{corr:barren_plateau_vqc_local_cost}[Absence of Barren Plateau in Local Cost]
 Given the local VQC cost function, $\Cbs_{\Lbs}$ (\eqref{eqn:local_cost_full}) in $M\rightarrow N$ cloning, and a hardware efficient fixed structure $\Ansatz$, $U(\paramtheta)$, made up of alternating blocks, $W$, with a depth $\mathcal{O}(\log(N))$, where each block forms a local 2-design. Then the variance of the gradient of $\Cbs_{\Lbs}$ with respect to a parameter, $\theta_l$ can be lower bounded as:
 \begin{equation}
     G_N := \min(G_N(K, k)) \leq \text{Var}[\partial_l \Cbs],\qquad G_N(K, k) \in \Omega(1/\text{poly}(N))
 \end{equation}
\end{corollary}
One final thing to note, is that the $\Ansatz$ we choose in \figref{fig:hardware_efficient_ansatz_phase_cov_cloning}, does not form an exact local 2-design, but the same $\Ansatz$ is used in Ref.~\cite{cerezo_cost-function-dependent_2020}) and is sufficient to exhibit a cost function dependent barren plateau.

\subsection{Training Sample Complexity} \label{app_ssec:training_sample_complexity}

Here we study the sample complexity of the training procedure by retraining the continuous parameters of the learned circuit (\figref{fig:learned_vs_ideal_circ_on_hw}(c)) starting from a random initialisation of the parameters, $\paramtheta$. As expected, as the number of training samples increases (i.e.\@ the number of random choices of the phase parameter, $\eta$, in \eqref{eqn:x_y_plane_states}), the generalisation error (difference between training and test error) approaches zero. This is not surprising, since the training set will eventually cover all states on the equator of the Bloch sphere.

\begin{figure}[ht]
    \begin{center}
        \includegraphics[width=0.7\columnwidth, height=0.4\columnwidth]{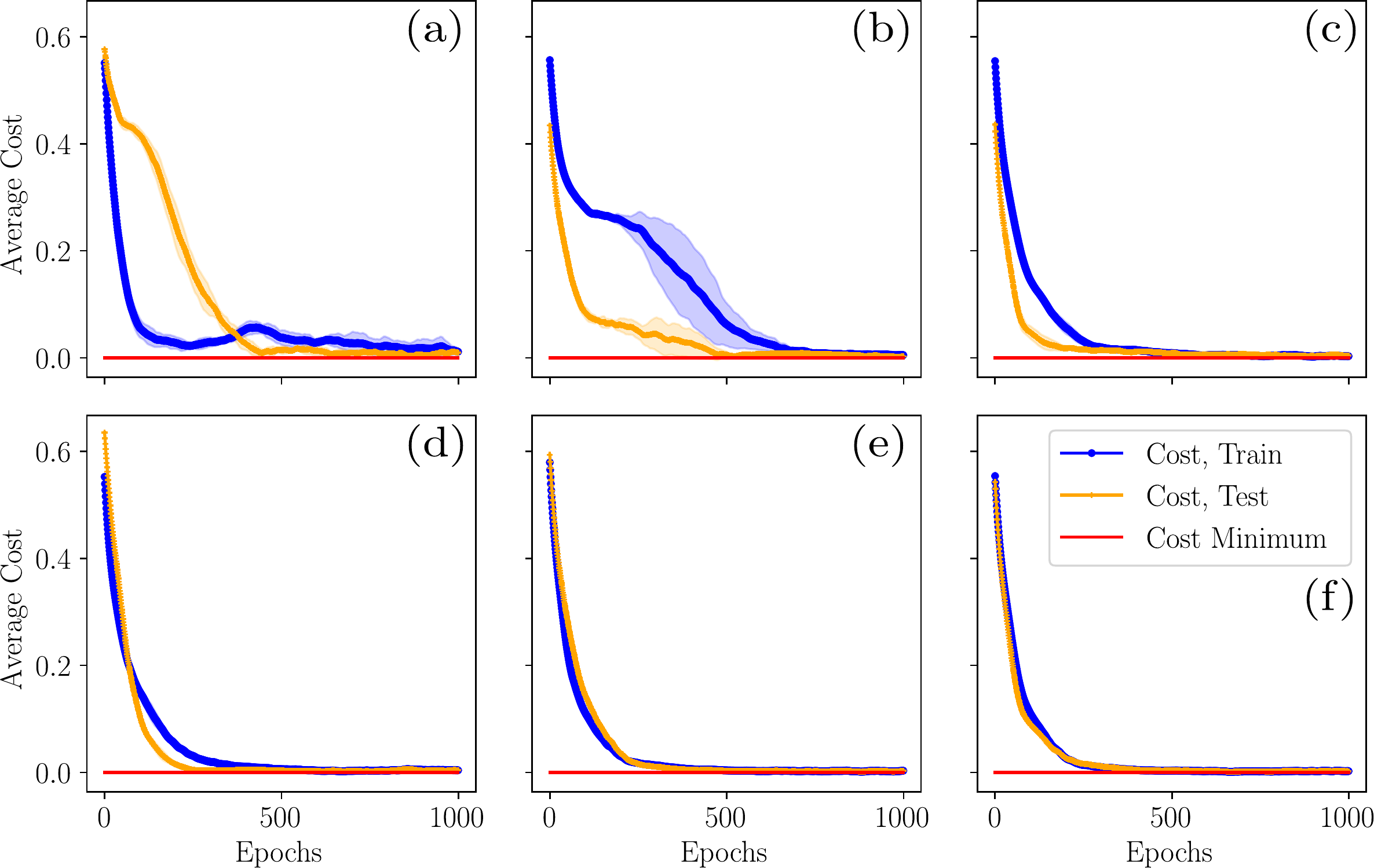}
    \caption{Sample complexity of VQC using the squared cost. We begin with a random initialisation of the structure learned circuit in \figref{fig:learned_vs_ideal_circ_on_hw}(c) and reoptimise the parameters using different sizes in the training/text set, and different mini-batch sizes. All of the following using a train/test split of $20\%$ and we denote the tuple $(i, j, k)$ as $i = $ number of training samples, $j = $  number of test samples, $k = $ batch size. (a) $(1, 1, 1)$, (b) $(4, 1, 2)$, (c) $(8, 2, 5)$, (d) $(16, 4, 8)$ (e) $(40, 10, 15)$, (f) $(80, 20, 20)$. }
    \label{fig:phase_cov_learned_circuit_sample_complexity}
        \end{center}
\end{figure}

\subsection{Local Cost Function Comparison} \label{app_ssec:connectivity_and_cost_function_comparison_aharanov}

In \figref{fig:local_vs_squared_cost_aharanoov_2to4}, we demonstrate the weakness of the local cost function, $\mathsf{C}_{\Lbs}$, in not enforcing symmetry strongly enough in the problem output, and how the squared cost function, $\mathsf{C}_{\mathsf{sq}}$ can alleviate this, for $2\rightarrow 4$ cloning specifically. Here we show the optimal fidelities found by VQC with a variable structure Anstaz, starting from a random structure. The local cost tends towards local minima, where one of the initial states ($\rho^1_\theta$) ends up with high fidelity, while the last qubit ($\rho^4_\theta$) has a low fidelity. This is alleviated with the squared cost function which is clearly more symmetric, on average, in the output fidelities. This is observed for both circuit connectivities we try (although a NN architecture is less able to transfer information across the circuit for a fixed depth).

\begin{figure}[ht]
    \begin{center}
        \includegraphics[width=0.85\columnwidth, height=0.3\columnwidth]{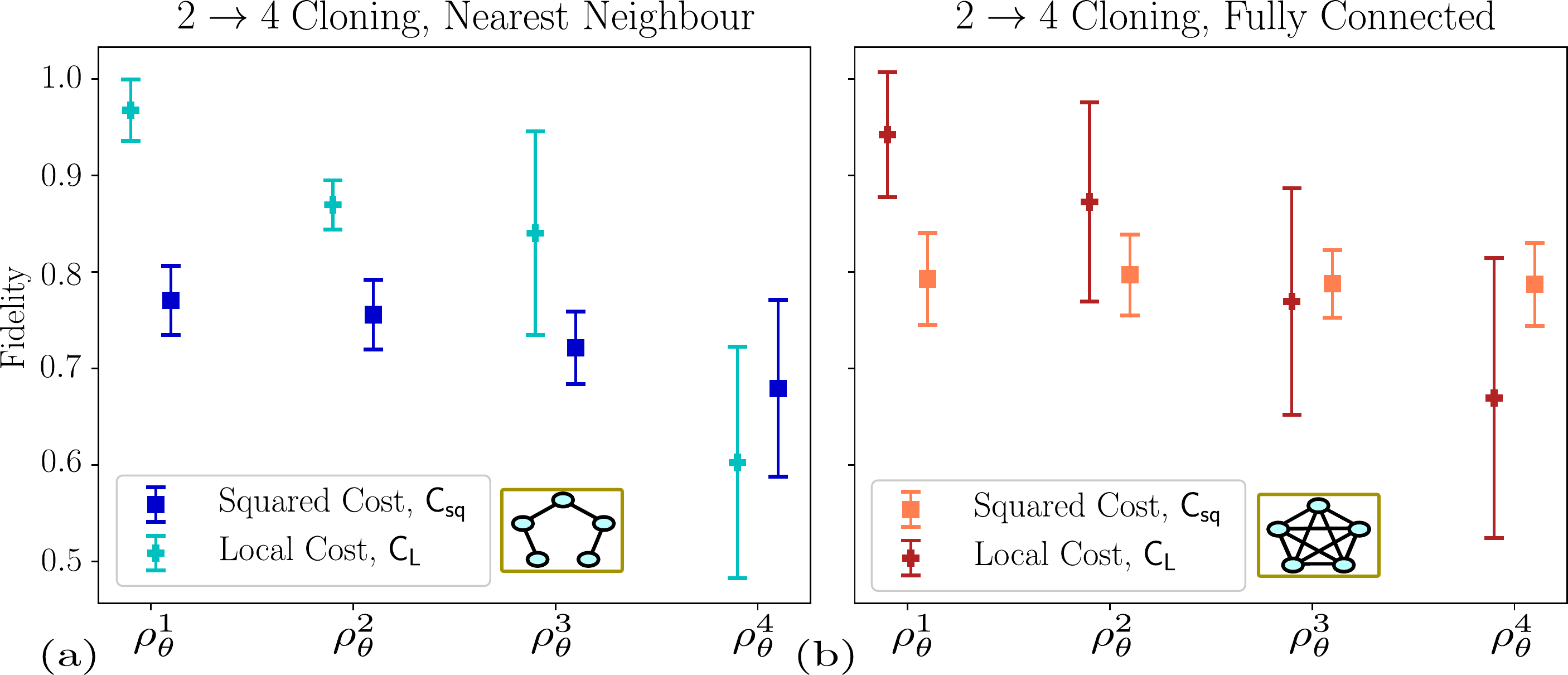}
    \caption{Comparison between the local (\eqref{eqn:local_cost_full}) and squared (\eqref{eqn:squared_local_cost_mton}) cost functions for $2\rightarrow 4$ cloning. (a) shows nearest neighbour (NN) and (b) has a fully connected (FC) entanglement connectivity allowed in the variable structure Ansatz. Again, we use the family of states in the protocol $\mathcal{P}_2$. Plots show the mean and standard deviation of the optimal fidelities found by VQC over 10 independent runs ($10$ random initial circuit structures). A sequence length of $35$ is used for $1\rightarrow 3$ and $40$ for $2\rightarrow 4$, with $50$ iterations of the variable structure Ansatz search in both cases. Here we use the same experiment hyperparameters as in \figref{fig:1to3_2to4_aharoanov_optimal_fidelities_plus_nn_vs_fc} in the main text.
    }
    \label{fig:local_vs_squared_cost_aharanoov_2to4}
        \end{center}
\end{figure}
%


\section{VQC Learned Circuits} \label{app_sec:vqc_learned_circuits}

Here we give the explicit circuits learned by VQC and which give the results in the main text. We mention as above, that these are only representative examples, and many alternatives were also found in each case.

\subsection{Ancilla-Free Phase Covariant Cloning} \label{app_ssec:ancialla_free_phase_covariant}

The circuits found in \secref{sec:results} in the main text to clone phase-covariant states are slightly more general than we may wish to use. In particular, the circuit \figref{fig:qubit_cloning_ideal_circ} also has the ability to clone \emph{universal} states, due to the addition of the ancilla, which can be used as a resource. However, it is well known that phase covariant cloning can be implemented economically, i.e.\@ \emph{without} the ancilla\cite{niu_two-qubit_1999, scarani_quantum_2005}. As such, we could compare against a shorter depth circuit which also does not use the ancilla. For example, the circuit from Ref.~\cite{du_experimental_2005} shown in \figref{fig:phase_covariant_cloning_circuits_2_qubits}(a) is also able to achieve the optimal cloning fidelities ($\sim 0.85$). An example VQC learned circuit for this task can be seen in \figref{fig:phase_covariant_cloning_circuits_2_qubits}(b) which has $2$ $\CZ$ gates. We note that this ideal circuit can be compiled to \emph{also} use $2$ $\CZ$ gates, so in this case VQC finds a circuit which is approximately comparable up to single qubit rotations.

\begin{figure}
    \centering
    \includegraphics[width=\columnwidth,height=0.15\textwidth]{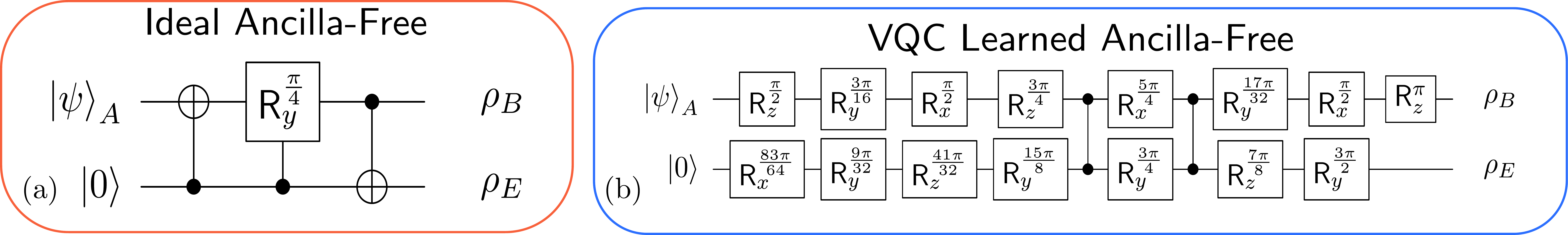}
    \caption{2 qubit circuits to clone phase-covariant states, without ancilla. (a) Optimal circuit from Ref.\cite{du_experimental_2005}, (b) circuit learned by VQC. In it can be checked that the ideal circuit in (a), can be compiled to \emph{also} use $2$ $\CZ$ plus single qubit gates, so VQC has found something close to optimal. The average fidelities for $B$, $E$ for the circuit in (b) is $F_{\mathsf{L},  \text{VQC}}^{B, \text{PC}}  \approx 0.854$ and  $F_{\mathsf{L},  \text{VQC}}^{E, \text{PC}}  \approx 0.851$ respectively, over $256$ input samples, $\ket{\psi}_A$ (comparing to the ideal fidelity of $F_{\mathsf{L},  \text{opt}}^{\text{PC}} = 0.853$).}
    \label{fig:phase_covariant_cloning_circuits_2_qubits}
\end{figure}

\subsection{State-Dependent Cloning Circuits} \label{app_ssec:state-dependent-circuits}

\figref{fig:mayers_1to2_cloning_vqc_circuit} shows the circuit used to achieve the fidelities in the attack on $\mathcal{P}_1^1$ in the main text. In training, we still allowed an ancilla to aid the cloning, but the example in \figref{fig:mayers_1to2_cloning_vqc_circuit} did not make use of it (in other words, VQC only applied gates which resolved to the identity on the ancilla), so we remove it to improve hardware performance. This repeats the behaviour seen for the circuits learned in phase-covariant cloning. We mention again, that some of the learned circuits did make use of the ancilla with similar performance.

\begin{figure}
    \centering
    \includegraphics[width=0.5\columnwidth,height=0.1\textwidth]{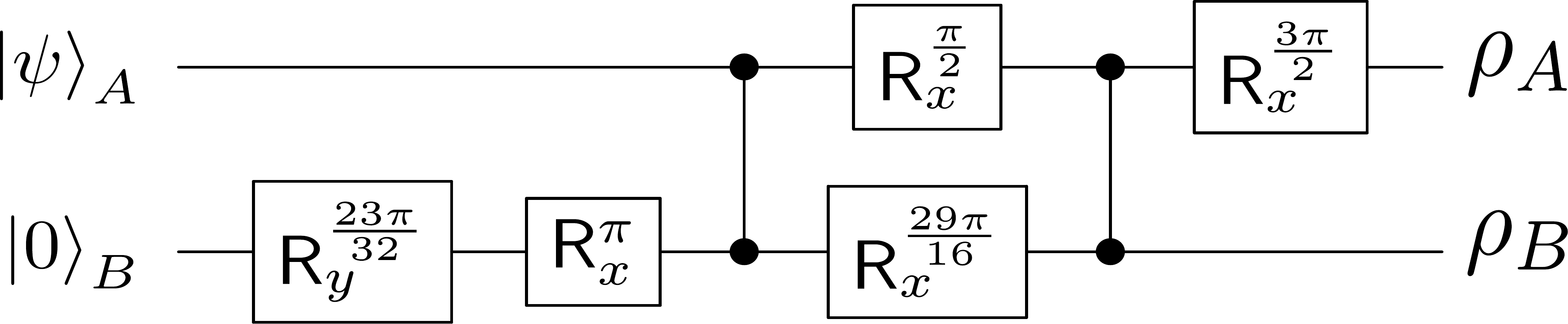}
    \caption{Circuit learned by VQC in to clone states, $\ket{\phi_0}, \ket{\phi_1}$, with an overlap $s=\cos\left(\pi/9\right)$ in the protocol, $\mathcal{P}_1$. For example, $\rho_A$ is the clone sent back to Alice, while $\rho_B$ is kept by Bob.}
    \label{fig:mayers_1to2_cloning_vqc_circuit}
\end{figure}

\figref{fig:aharonov__1to2_1to3_2to4_circuits} shows the circuits learned by VQC and approximately clone all four states in \eqref{eqn:aharonov_coinflip_states} in the protocol, $\mathcal{P}_2$, for $1\rightarrow 2, 1 \rightarrow 3$ and $2 \rightarrow 4$ cloning. These are the specific circuits used to produce the fidelities in \figref{fig:aharonov_1to2_cloning_fidelities_variational} and \figref{fig:1to3_2to4_aharoanov_optimal_fidelities_plus_nn_vs_fc}.

\begin{figure}
    \centering
    \includegraphics[width=0.7\columnwidth,height=0.5\textwidth]{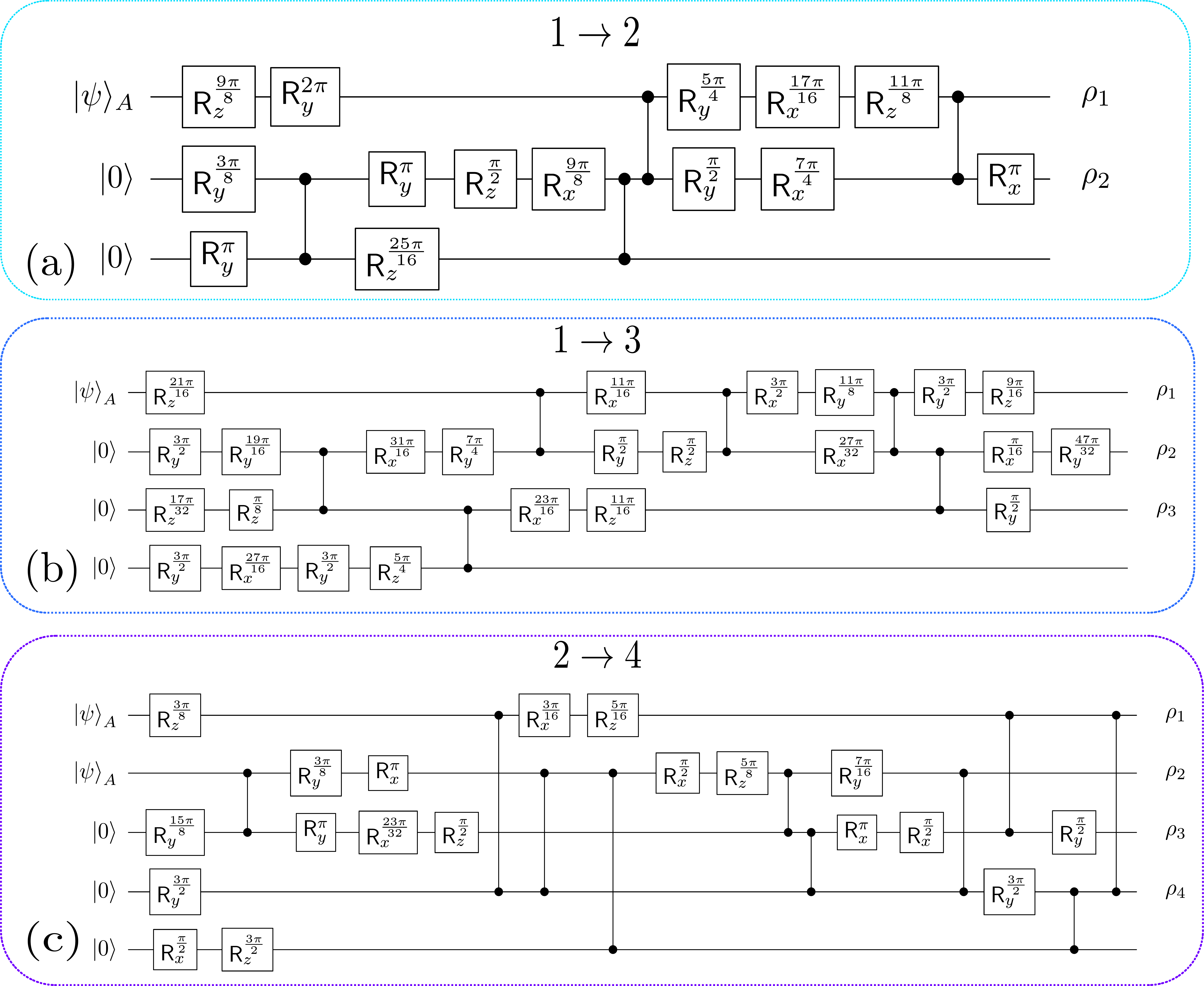}
    \caption{Circuits learned by VQC to clone states from the protocol, $\mathcal{P}_2$ for (a) $1\rightarrow 2$, (b) $1\rightarrow 3$ and (c) $2 \rightarrow 4$ cloning. These specific circuits produce the fidelities in \figref{fig:aharonov_1to2_cloning_fidelities_variational} for $1\rightarrow 2$, (using the local cost function), and in \figref{fig:1to3_2to4_aharoanov_optimal_fidelities_plus_nn_vs_fc} for $1\rightarrow 3$ and $2 \rightarrow 4$ (using the squared cost function). We allow an ancilla for all circuits, and $\rho_k$ indicates the qubit which will be the $k^{th}$ output clone.}
    \label{fig:aharonov__1to2_1to3_2to4_circuits}
\end{figure}

\end{document}